\documentclass{article}
\usepackage[american]{babel}
\usepackage[utf8]{inputenc}
\usepackage[T1]{fontenc}
\usepackage{csquotes}
\usepackage{amsmath, amssymb, amsfonts, amsthm, bm}
\usepackage{textcomp}
\usepackage{appendix}
\usepackage{color}
\usepackage{dsfont}
\usepackage{enumerate}
\usepackage[mathscr]{euscript}
\usepackage{graphicx}
\usepackage{latexsym}
\usepackage{mathrsfs}
\usepackage{enumitem}
\usepackage{mathtools}
\usepackage{eurosym}  

\usepackage[a4paper, top=3cm, bottom=3.3cm, left=3.5cm, right=3.5cm]{geometry}
\usepackage{setspace}

\usepackage{tikz}
\usepackage[font=small]{caption}

\usepackage[colorlinks=true,citecolor=blue,urlcolor=blue,linkcolor=blue]{hyperref}

\usepackage[backend=biber,giveninits=true,sorting=nyt,url=false,maxnames=10,maxcitenames=6,mincitenames=6]{biblatex}
\DeclareNameAlias{default}{family-given}
\renewbibmacro*{in:}{%
  \ifentrytype{article}
    {}
    {\printtext{\bibstring{in}\intitlepunct}}}
\DeclareNolabel{\nolabel{\regexp{[\p{P}\p{S}\p{C}\p{Z}]+}}}
\bibliography{Biblio, Biblio_ASMF}

\newcommand{\R}{\mathbb{R}}
\newcommand{\N}{\mathbb{N}}
\newcommand{\E}{\mathbb{E}}
\newcommand{\F}{\mathcal{F}}
\newcommand{\Q}{\mathbb{Q}}
\renewcommand{\P}{\mathbb{P}}
\newcommand{\1}{\mathds{1}}

\newcommand{\ds}{\displaystyle}
\newcommand{\itemref}[1]{$\textit{\ref{#1}}$}
\newcommand{\norm}[1]{|\!|#1|\!|}
\renewcommand{\ss}{\scriptstyle}

\def\a{{\alpha}}
\def\b{{\beta}}
\def\r{{\rho}}
\def\l{{\lambda}}
\def\Om{{\Omega}}
\def\om{{\omega}}
\def\s{{\sigma}}
\def\t{{\tau}}
\def\g{{\gamma}}
\def\eps{{\varepsilon}}

\numberwithin{equation}{section}
\mathtoolsset{showonlyrefs}

\newtheorem{theorem}{Theorem}
\newtheorem{corollary}[theorem]{Corollary}
\newtheorem{definition}[theorem]{Definition}
\newtheorem{lemma}[theorem]{Lemma}
\newtheorem{proposition}[theorem]{Proposition}
\newtheorem{example}[theorem]{Example}

\renewcommand{\d}{\,\mathrm{d}}

\theoremstyle{remark}
\newtheorem{remark}[theorem]{Remark}
\newtheorem{example_text}[theorem]{Example}
\newtheorem{notation}{Notation}

\DeclareMathOperator*{\esssup}{ess\,sup}
\DeclareMathOperator*{\essinf}{ess\,inf}

\title{Measuring Financial Resilience \\ 
Using Backward Stochastic Differential Equations\thanks{We are very grateful to Yacine A\"it-Sahalia, Jose Blanchet, Freddy Delbaen, Lukasz Delong, Damir Filipovic, Marco Frittelli, Peter Spreij, Mitja Stadje, Stefan Weber, Bert Zwart and participants at the 12th General AMaMeF Conference in Verona, 
the 2025 Vienna Congress on Mathematical Finance (VCMF), 
the 2025 ANR DREAMeS Conference in Paris,
the Workshop on Insurance and Financial Mathematics in Hannover,
Tsinghua University, 
the University of Tartu and the University of Amsterdam for useful comments and discussions.
This research was funded in part by
the Netherlands Organization for Scientific Research (NWO) under grant NWO Vici 2020--2027 (Laeven, Ferrari, Zullino).
Matteo Ferrari, Emanuela Rosazza Gianin and Marco Zullino are members of Gruppo Nazionale per l’Analisi Matematica, la Probabilità e le loro Applicazioni (GNAMPA)-INdAM, Italy. 
Emanuela Rosazza Gianin and Marco Zullino acknowledge the financial support of Gnampa Research Project 2024 (PRR-20231026-073916-203).}\bigskip}
\author{Roger J.~A.~Laeven\footnote{Corresponding author.} \\
{\footnotesize Dept.~of Quantitative Economics}\\
{\footnotesize University of Amsterdam, CentER}\\
{\footnotesize and EURANDOM, The Netherlands}\\
{\footnotesize \texttt{R.J.A.Laeven@uva.nl}}\\
\and Matteo Ferrari \\
{\footnotesize Dept.~of Quantitative Economics}\\
{\footnotesize University of Amsterdam, The Netherlands}\\
{\footnotesize \texttt{M.Ferrari@uva.nl}}\\
\and Emanuela Rosazza Gianin \\
{\footnotesize Dept.~of Statistics and Quantitative Methods}\\
{\footnotesize University of Milano-Bicocca, Italy}\\
{\footnotesize \texttt{emanuela.rosazza1@unimib.it}}\\
\and Marco Zullino \\
{\footnotesize Dept.~of Mathematics and Applications}\\
{\footnotesize University of Milano-Bicocca, Italy}\\
{\footnotesize \texttt{M.Zullino@campus.unimib.it}}\\
}

\date{\today}

\begin{document}

\maketitle

\begin{abstract}
We introduce the \emph{resilience rate} as a measure of financial resilience.
It captures the expected rate at which a dynamic risk measure recovers, i.e., bounces back, when the risk-acceptance set is breached.
We develop the corresponding stochastic calculus by establishing representation theorems for expected time-derivatives of solutions to backward stochastic differential equations (BSDEs) with jumps, evaluated at stopping times.
These results reveal that the resilience rate can be represented as a suitable  expectation of the generator of a BSDE.
We analyze the main properties of the resilience rate and the formal connection of these properties to the BSDE generator. 
We also introduce resilience-acceptance sets and study their properties in relation to both the resilience rate and the dynamic risk measure.
We illustrate our results in several canonical financial examples and highlight their implications via the notion of resilience neutrality.
\\[3mm]
\\[3mm]
\end{abstract}

\newpage
\onehalfspacing
\setcounter{tocdepth}{2}

\section{Introduction}

Over the past decades, an extensive literature across diverse fields has studied the question of how to measure financial risk, both in a static and a dynamic environment.
Measures of financial risk are employed for a wide variety of purposes, e.g., to determine optimal economic policies, set risk capital requirements, calculate insurance premia, and explain variation in asset prices (e.g.,\ \cite{Follmer+Schied_2016_Stochastic_finance,Barrieu+ElKaroui_2009_Pricing_hedging_optimally_designing_derivatives_minimization_risk_measures}).

Yet, irrespective of the quality of the risk measurement and the ensuing economic policies, adverse events are inevitable and continue to occur, with sufficient ferocity to put the entity under scrutiny --- whether individuals or households, firms, markets, financial institutions or society at large --- in jeopardy. 
The natural, and pivotal, question that arises then is: at what \textit{pace} can the entity be expected to \textit{recover from adverse events}, i.e., to what extent is the entity \textit{financially resilient}? 
Quite surprisingly, relatively little is known about measurement theory and methods to answer this fundamentally important question.

In this paper, we introduce a measure of financial resilience.
Our measure, which we will refer to as the \textit{resilience rate}, evaluates the expected \textit{rate} at which a \textit{dynamic risk measure recovers}, i.e., bounces back, upon the occurrence of an adverse shock in the form of a certain \textit{stress scenario}. 
More formally, our measure of financial resilience is defined as the expectation of the time-derivative of a dynamic risk measure at the first time the risk measure breaches a risk level deemed as acceptable by the financial investor, institution or regulatory authority.

Our measure builds on the theory of dynamic risk measurement using backward stochastic differential equations (BSDEs). 
Since the seminal work \cite{Pardoux+Peng_1990_Adapted_solution_backward_stochastic_differential_equation} on BSDEs, significant progress has been made. 
Main references dealing with BSDEs and their relation with dynamic risk measures in a purely Brownian filtration are \cite{ElKaroui+Peng+Quenez_1997_Backward_stochastic_differential_equations_finance,Kobylanski_2000_Backward_stochastic_differential_equations_partial_differential_equations_quadratic_growth,Peng_2005_dynamically_consistent_nonlinear_evaluations,Barrieu+ElKaroui_2009_Pricing_hedging_optimally_designing_derivatives_minimization_risk_measures}, while \cite{Buckdahn+Pardoux_1994_BSDEs_jumps_associated_integro-partial_differential_equations,Royer_2006_Backward_stochastic_differential_equations_jumps_related_non-linear_expectations,Quenez+Sulem_2013_BSDEs_jumps_optimization_applications_dynamic_risk_measures} deal with the enlargement of the Brownian filtration to include Poissonian noise.
Our starting point is a Brownian-Poissonian filtration, featuring random jumps coupled with a pure diffusion component. 
To establish that our measure of financial resilience is well defined, we develop the required stochastic calculus, by deriving representation theorems for expected time-derivatives of solutions to BSDEs with jumps, evaluated at stopping times.
These results demonstrate that our measure of financial resilience takes the form of a suitable expectation of the generator of a BSDE at a stopping time.

More specifically, this paper adopts a probabilistic setting in which dynamic risk assessment is conducted using a risk functional defined as (the first component of) the solution to a BSDE. 
Formally, given an $m$-dimensional standard Brownian motion $W$ and a compensated Poisson random measure $\tilde N$ on $\R_+\times \R^d_\ast:=[0,+\infty)\times(\R^d\setminus\{0\})$, we consider the following BSDE:
\begin{equation*}
    \rho_t (X)=X+\int_t^{T}g(s, \rho_s(X), Z_s, U_s)\d s - \int_t^{T} Z_s\cdot \d W_s-\int_{(t,T]\times \R^d_\ast}U_s{(x)} \d\tilde N(s,x),
\end{equation*}
where $(g,T,X)$ are the parameters of the equation, and the triple $\big(\rho(X),Z,U\big)$ is the (unknown) solution. 
The random variable $X$ is commonly referred to as the terminal condition of the BSDE, $g$ is called the driver or generator, while $T\in(0,+\infty)$ is the terminal time or horizon of the BSDE. 
As is well-known, many financially meaningful dynamic risk measures are induced via BSDEs, in diverse settings with e.g., Lipschitz continuous or quadratic drivers and possibly infinite activity jumps (see \cite{ElKaroui+Ravanelli_2009_Cash_sub-additive_risk_measures_interest_rate_ambiguity,Laeven+Stadje_2014_Robust_portfolio_choice_indifference_valuation,Laeven+RosazzaGianin+Zullino_2023_Dynamic_return_star-shaped_risk_measures_BSDEs,Pardoux_1997_Generalized_discontinuous_backward_stochastic_differential_equations}, among others).

Whereas the properties of $\rho_{t}(X)$ have been extensively studied over the past three decades, e.g., in relation to the minimum capital requirements needed to achieve risk-acceptance with respect to the underlying risk measure, to our knowledge, there is no detailed account of possible definitions of expected time-derivatives for solutions to BSDEs. 
In this paper, motivated by the conceptualization of the resilience rate, we introduce and analyze the concept of the expected time-derivative of (the first component of) the solution to a BSDE. 
Formally, we investigate the quantity
\begin{equation}\bm\dot\rho_\t(X):=\lim_{\varepsilon\to0^+}\frac{1}{\varepsilon}\mathbb{E}\left[\rho_{\t+\varepsilon}(X)-\rho_{\t}(X)|\tau<T\right],\label{eq:rm}\end{equation}
where $\t$ is a stopping time taking values in $[0,T]$.

In particular, we first demonstrate that, under standard Lipschitz or quadratic assumptions on the BSDE driver and for deterministic times $t\in[0,T)$, the above limit equals (minus) the expectation of the driver evaluated at the solution, i.e.,
\[
    \bm\dot\rho_t(X)=-\mathbb{E}\left[g(t,\rho_t(X),Z_t,U_t)\right],
\]
which we refer to as the \emph{bouncing drift}.
This representation admits a particularly intuitive interpretation: 
whereas $\rho_t(X)$ reflects the \emph{level} of risk associated with $X$ at time $t\in[0,T)$, its expected time-derivative $\bm\dot\rho_t(X)$ captures the \emph{drift} or \emph{intensity} at which the risk level changes over time. 
To further support this interpretation, the direct dependency on the driver $g$, whose role as an infinitesimal generator of local risk preferences has been studied in \cite{Barrieu+ElKaroui_2009_Pricing_hedging_optimally_designing_derivatives_minimization_risk_measures, ElKaroui+Ravanelli_2009_Cash_sub-additive_risk_measures_interest_rate_ambiguity}, offers a meaningful perspective: the bouncing drift is intimately related to the local risk preferences of a decision-maker. 

When proceeding to the more general setting of stopping times instead of deterministic times, the standard assumptions on the BSDE parameters are no longer sufficient to guarantee the well-posedness of the limit in \eqref{eq:rm}. 
Therefore, we identify sufficient conditions involving the path-regularity of both the driver $g$ and the solution $\big(\rho(X),Z,U\big)$ under which we demonstrate that
\[
    \bm\dot\rho_{\tau}(X)=-\E\left[g\big(\t,\rho_{\t}(X),Z_{\t^+},U_{\t^+}\big)\big|\t<T\right],
\]
which is mathematically more subtle, and we provide a detailed verification analysis of the required conditions.

Once the well-posedness results have been established, we investigate how the usual axioms for a (dynamic) risk measure, such as time consistency, positive homogeneity, cash-additivity, continuity, convexity or star-shapedness, are connected to analogous properties of its resilience rate.
We also establish a version of the comparison theorem that holds for the expected time-derivative of the first component solution to a BSDE.
In a spirit similar to \cite{Jiang_2008_Convexity_translation_invariance_subadditivity_g-expectations_related_risk_measures}, we also manage to show that --- under suitable conditions --- the implications can be reversed, so that the main properties of a BSDE-induced resilience rate are inherited by its driver.

Furthermore, we introduce the new concept of a family of \textit{resilience-acceptance sets}.
While the literature contains several ways to define acceptance sets for risk measures that fully characterize the associated risk functional, to our best knowledge, this is the first study to address resilience-acceptance sets. 
For levels of financial resilience $a\in\R$ and times $t\in[0,T)$, we set
\[
    \mathcal{R}^{a}_{tT}(\rho):=\big\{X\in L^{p}(\mathcal{F}_T) \, : \, \bm\dot\rho_t(X)\leq a\big\}.
\]
In particular, we demonstrate that the resilience rate is fully characterized by resilience-acceptance sets and we establish a correspondence between key properties of the resilience rate --- such as time consistency, cash-insensitivity and positive homogeneity --- and analogous properties of the resilience-acceptance sets.

Next, we study the behavior of the resilience rate in several canonical financial examples. 
Beginning with the analysis of the price of a European put option in a Brownian financial market, we derive a closed-form expression for its resilience rate. 
We also obtain a closed-form formula for the resilience rate of a zero-coupon bond price, assuming the short rate follows a mean-reverting, Vasicek model. 
We further consider a financial setting with ambiguity regarding the appropriate risk-free rate to discount a bond's cash flows and, in this context, we derive explicit formulas for the corresponding bouncing drift.
Moreover, we evaluate the resilience rate of an investor employing exponential utility-based risk measures --- specifically, the entropic risk measure ---, and we provide a detailed analysis of the resilience rate in an incomplete financial market framework where the underlying asset is subject to jump activity.

In the last part of the paper, we delve further into the financial interpretation and implications of the resilience rate.
Whereas acceptance sets for cash-additive risk measures support the interpretation of such measures as the level of capital that, when added to a financial position (and invested in a risk-free asset), renders the position risk-acceptable, a cash-insensitive measure of resilience can be interpreted as an additional drift that, when added to the driver $g$, makes the position resilience-acceptable.
Moreover, the time integral of the resilience rate defines a \textit{resilience risk adjustment} --- namely, a correction term that can be added to the risk assessment provided by the dynamic risk measure to ensure a \textit{resilience-neutral} attitude for the financial position.

The concept of resilience has many different facets.
It is widely studied in psychology, where psychological resilience refers to the process of adapting to adversity, through behavioral, emotional or mental flexibility (cf.\ \cite{APA_2018_Resilience}).
It is also extensively studied in e.g., ecology, where ecological resilience refers to the ability of an ecosystem to absorb adverse shocks and adapt to change while maintaining its current functions (cf.\ \cite{Holling_1973_Resilience_Stability_Ecological_Systems}).
In economics, insurance and finance, the concept of resilience is only nascent. 
We define financial resilience through the resilience rate: the expected rate at which a dynamic risk evaluation recovers, that is, bounces back, following adverse shocks in the guise of a financial stress scenario.

Resilience should be distinguished from robustness. 
Robust measures of risk take into account that the probabilistic model may be misspecified, as the true probabilistic model is unknown (cf.\ \cite{ Follmer+Schied_2016_Stochastic_finance}). 
They induce optimal economic policies that are robust against model misspecification, performing sufficiently well under a wide variety of probabilistic models (cf.\ \cite{Hansen+Sargent_2001_Robust_Control_Model_Uncertainty}). 
Measures of resilience induce policies, hence financial positions, that are resilient, i.e., they are expected to recover sufficiently quickly when adverse events nevertheless occur.
Resilience should also be distinguished from stability/solidity as the ability to withstand adverse events. 
Financial positions are stable if they are not much affected by adverse scenarios.
We refer to \cite{Brunnermeier_2021_Resilient_society} for a general discussion of resilient societies.  

This paper is organized as follows.
In Section~\ref{sec:prel}, we introduce preliminaries.
In Section~\ref{sec:insmeares}, we introduce our measure of resilience and establish representation theorems.
In Section~\ref{sec:properties}, we examine how the main properties of dynamic risk measures are connected to the resilience rate 
and introduce the concept of resilience-acceptance sets.
Section~\ref{sec:examples} contains examples.
Section~\ref{SEC:interpretation} expands further on the resilience rate's financial interpretation. 
Section~\ref{sec:con} provides a conclusion.
Some additional technical details are relegated to four online appendices.

\section{Preliminaries}\label{sec:prel}
In this section, we present the definitions of the functional spaces used in this paper and the basic definitions of dynamic risk measures, the required results on BSDEs and their relation with dynamic risk measures, and some useful preliminaries on Malliavin calculus.

\subsection{Functional spaces and dynamic risk measures}
\label{SEC:funct_spaces_dyn_risk_measures}

In the following, we denote $\overline\N_0:=\{n\in\mathbb Z \ : \ n\geq 0\}\cup\{+\infty\}$ and endow it with the power set $\s$-algebra.
We also let $\R_+:=[0,+\infty)$ and $\R^d_\ast:=\R^d\setminus\{0\}$, for $d\in\N$, with the induced Euclidean topology.
We denote the Euclidean scalar product between $a, b \in \mathbb{R}^d$ as $a \cdot b$, and the induced norm by $\norm{a}$. 
For a topological space $E$, we let $\mathscr B(E)$ be its Borel $\s$-algebra.
Moreover, $\ell_1$ denotes the Lebesgue measure on $\mathscr B(\R)$, or its restriction to Borel subsets of~$\R$.
We use the symbol $\otimes$ for the product of two $\s$-finite measures.

Let $(\Om,\F,\P)$ be a complete probability space.
As usual, the probability $\P$, when restricted to a sub-$\s$-algebra, will still be denoted by $\P$.
If $p\in[1,+\infty]$ and $\mathcal G$ is a sub-$\s$-algebra of $\F$, we denote by $L^{p}\left(\mathcal G\right)$ the normed space $L^{p}\left(\Omega,\mathcal G,\mathbb{P}\right)$ of real-valued, $\mathcal G$-measurable and $p$-integrable (or essentially bounded for $p=\infty$) random variables $X:\Om\to\R$, endowed with the usual norm topology.

Let
\[
    N:\Omega\times\mathscr B\big(\R_+\times \R^d_\ast \big)\to\overline\N_0
\]
be a Poisson random measure on $\R_+\times \R^d_\ast$ (see, for instance, \cite[Chapter~4]{Sato_1999_Levy_Processes_Infinitely_Divisible_Distributions}) that satisfies the following conditions.
There exists a (non-negative) measure $\nu$ on $\R^d_\ast$ such that:
\begin{enumerate}[noitemsep, topsep=0pt]
    \item 
        $\nu$ is $\s$-finite and $\ds\int_{\R^d_\ast }(1\wedge \norm{x}^2)\d \nu(x)<+\infty$.
    \item 
        The intensity measure (sometimes called mean measure, or compensator) of $N$ satisfies $\E[N(A)]=(\ell_1\otimes\nu)(A)$, for all $A\in\mathscr B\big(\R_+\!\times \R^d_\ast \big)$.
    \item 
        Denoting the compensated random measure of $N$ as 
        $\tilde N:=N-\ell_1\otimes\nu$, see, e.g., \cite[Chapter~13]{Cohen+Elliott_2015_Stochastic_calculus_applications},
        if $B\in\mathscr B(\R^d_\ast )$ is such that $\nu(B)<+\infty$, then the $(-\infty,+\infty]$-valued process $(\om,t)\mapsto \tilde N(\om,[0,t]\times B)$
        is a martingale.
\end{enumerate}
We introduce the linear space $\Lambda^2:=L^2\big(\R^d_\ast ,\mathscr B(\R^d_\ast ),\nu\big)$ with its natural inner product $\langle u,v \rangle^{}_{\Lambda^2}:=\int_{\R^d_\ast}uv \d \nu$, for $u,v\in \Lambda^2$, and the induced norm $\|\,\cdot\, \|_{\Lambda^2}$.
Furthermore, let $W:\Omega\times\R_+\to\R^m$ be a standard Brownian motion, with $m\in\N$.

Let $\tilde\F^W_t:=\s\big(W_r \ : \ r\in[0,t]\big)$ and ${\tilde\F_t:=\tilde \F^W_t\vee \s\big(N([0,r]\times A) \ : \ A\in\mathscr B(\R^d_\ast  ), \ r\in[0,t]\big)}$, for any $t\geq 0$.
Then we define the complete and right-continuous filtrations ${\bm \F}^W:=\big(\F^W_t\big)_{t\geq 0}$, and $\bm \F:=(\F_t)_{t\geq 0}$, on $(\Om,\F,\P)$, by augmenting $\big(\tilde\F^W_t\big)_{t\geq 0}$ and $(\tilde\F_t)_{t\geq 0}$, respectively.
Namely,
\[
    \F^W_t:=\tilde \F^W_t \vee \bigcap_{s>t}\tilde \F^W_s\vee \mathscr N, 
    \qquad
    \F_t:=\tilde \F_t \vee \bigcap_{s>t}\tilde \F_s\vee \mathscr N, \qquad \forall\, t\geq 0,
\]
where $\mathscr N$ is the family of all events in $\F$ of null $\P$-probability.
The filtrations $\bm \F^W$, $\bm\F$ will be referred to as the Brownian and Brownian-Poissonian filtration, respectively.
We denote by $\mathcal P$ (and $\mathcal P^W$) the $\s$-algebra on $\Om\times\R_+$ generated by continuous $\bm\F$-adapted (resp.\ $\bm\F^W$-adapted) processes $\Om\times\R_+\to\R$.
Unless otherwise stated, equalities and inequalities between random variables will be understood $\mathbb{P}$\text{-a.s.}, whereas equalities and inequalities for processes must be interpreted $\mathbb{P}\otimes \ell_1$-a.e.

When, referring to a function $\psi:\Om\times\R_+\times E\ni(\om,t,e)\mapsto \psi(\om,t,e)\in \R$, for some non-empty set $E$, we say that $\psi$ satisfies $\P\otimes\ell_1$-a.e.\ the property $P$ in $E$, we mean that, for $\P\otimes\ell_1$-a.e.\ $(\om,t)\in\Om\times\R_+$, the function $E\ni e\mapsto \psi(\om,t,e)$ satisfies the property $P$.
For any $T\in(0,+\infty)$, $p\geq 1$, and $E\in\{\R^m,\Lambda^2\}$, we define the following classes of stochastic processes:
{\small\begin{align}
    L^p_T
    &:=\left\{\F\otimes\mathscr B\big([0,T]\big)\textit{-measurable }H:\Om\times[0,T]\to \R \ : \ \E\left[\int_0^T|H_t|^p\d t \right]<+\infty\right\},\\
    \mathcal S_T^p&:=\left\{\textit{càdlàg $\bm\F$-adapted }Y:\Om\times[0,T]\to \R \ : \ \E\left[\sup_{t\in[0,T]}|Y_t|^p\right]<+\infty\right\},\\
    \mathcal S_T^\infty&:=\left\{\textit{càdlàg $\bm\F$-adapted }Y:\Om\times[0,T]\to \R \ : \ \esssup_{\om\in\Om}\sup_{t\in[0,T]}|Y_t(\om)|<+\infty\right\},\\
    \mathcal H^p_T(E)&:=\text{\large$\Bigg\{$} \mathcal P\textit{-measurable }Z:\Om\times[0,T]\to E\ :\ \E\left[\left(\int_0^T\|Z_t\|_{E}^2\d t \right)^{p/2}\right]<+\infty\text{\large$\Bigg\}$},\\
        BMO(E)&:=\Bigg\{\mathcal P\textit{-measurable }Z:\Om\times[0,T]\to E \ :\ \esssup_{\Om}\sup_{t\in[0,T]}\E\left[\left.\int_t^T\|{Z_s}\|_E^2\d s \right|\F_t\right]<+\infty\Bigg\},
\end{align}}
\hspace{-0.14cm}where $\|\,\cdot\,\|_E=\norm{\,\cdot\,}$ if $E=\R^m$, or $\|\,\cdot\,\|_E=\|\,\cdot\,\|_{\Lambda^2}$ for $E=\Lambda^2$.
With slight abuse of notation, we use the same symbols for the normed spaces obtained as quotients of the above under the equivalence relation induced by their respective natural (semi)norms.
With this convention, a process in $\mathcal S^p_T$ or $\mathcal S^\infty_T$ is unique up to indistinguishability, while processes in $L^p_T$, $\mathcal H^p_T(E)$, $BMO(E)$ are unique up to a set of null $\P\otimes\ell_1$-measure.
All these spaces can be defined also in the Brownian setting, by replacing $\bm\F$-adaptedness and $\mathcal P$-measurability with $\bm\F^W$-adaptedness and $\mathcal P^W$-measurability, respectively. 
A right-continuous $\bm\F$-adapted stochastic process $Y:\Om\times[0,T]\to \R$ is said to be of class~$(D)$ if the family of random variables $\{Y_\tau \, : \, \tau \text{ is a $[0,T]$-valued $\bm\F$-stopping time}\}$ is uniformly integrable; see \cite[Definition~4.8]{Karatzas+Shreve_1991_Brownian_motion_stochastic_calculus}.
Moreover, if $Y\in \mathcal S^p_T$ for some $p\in[1,+\infty]$, then $Y$ is of class~$(D)$.

\begin{definition}[Dynamic risk measure]
\label{DEF:dynamic_risk_measure}
    Assume $p\in[1,+\infty]$ and $T>0$. 
    A dynamic risk measure on $L^p(\F_T)$ is a family $\rho=(\rho_t)_{t\in[0,T]}$ such that $\r_t:L^p(\F_T)\to L^p(\F_t)$ for any $t\in[0,T]$.
\end{definition}
\noindent A static risk measure $\rho_0:L^p(\F_T)\to \R$ is defined as the evaluation at time $0$ of the dynamic risk measure $\rho$; see \cite{Delbaen_2012_Monetary_Utility_Functions, Follmer+Schied_2016_Stochastic_finance}.

Let us now fix $p\in[1,+\infty]$ and $T>0$.
We present a non-exhaustive list of well-known axioms for a dynamic risk measure $\r$ on $L^p(\F_T)$.
We adopt the actuarial sign convention and interpret risk measures as the risk evaluation of a claim; see \cite{Barrieu+ElKaroui_2009_Pricing_hedging_optimally_designing_derivatives_minimization_risk_measures,Frittelli+RosazzaGianin_2002_Putting_order_risk_measures}.
\begin{itemize}[label=\scriptsize$\bullet$, noitemsep, topsep=0pt,leftmargin=1em]
    \item 
        Normalization: $\rho_t(0)=0$ for any $t\in[0,T]$.
    \item 
        \label{PAGE:time_consistency}Time consistency:\label{tc} $\rho_t\big(\rho_s(X)\big)=\rho_t(X)$ for any $X\in L^{p}(\mathcal{F}_T)$ and $0\leq t\leq s\leq T$.
    \item 
        Monotonicity: For  any $X,Y\in L^p(\F_T)$, if $X\leq Y$, then $\rho_t(X)\leq\rho_t(Y)$ for any $t\in[0,T]$.
    \item 
        Cash-additivity: If $h\in L^p(\F_t)$, then $\rho_t(X+h)=\rho_t(X)+h$ for any $X\in L^p(\F_T)$ and $t\in[0,T]$.
    \item 
        Positive homogeneity: $\rho_t(\alpha X)=\alpha\rho_t(X)$ for any $X\in L^p(\F_T)$, $\a\geq 0$ and $t\in[0,T]$.
    \item 
        Convexity: $\rho_t\big(\lambda X+(1-\lambda)Y\big)\leq \lambda\rho_t(X)+(1-\lambda)\rho_t(Y)$ for any $X,Y\in L^{p}(\mathcal{F}_T)$, $\lambda\in[0,1]$ and $t\in[0,T]$.
    \item 
        Star-shapedness: $\rho_t(\lambda X)\leq \lambda\rho_t(X)+(1-\lambda)\rho_t(0)$ for any $X\in L^{p}(\mathcal{F}_T)$, $\l\in[0,1]$ and $t\in[0,T]$.
    \item Continuity from above (below):
        For any $X\in L^p(\F_T)$ and any $\P$-a.s.\ non-increasing (resp.\ non-decreasing) sequence $(X^n)_{n\in\N}$ in $L^p(\F_T)$ such that $X^n \longrightarrow X$, $\P$-a.s., we have, for any $t\in[0,T]$, that $\rho_t(X^n)\longrightarrow \rho_t(X)$, $\P\text{-a.s.}$
\end{itemize}
For further details on the financial interpretation and implications of the properties listed above, interested readers are referred to, e.g., \cite{ Follmer+Schied_2016_Stochastic_finance, Barrieu+ElKaroui_2009_Pricing_hedging_optimally_designing_derivatives_minimization_risk_measures, Delbaen_2012_Monetary_Utility_Functions,Bellini+Laeven+RosazzaGianin_2018_Robust_return_risk_measures}. 

In a Brownian setting, it is common to define a dynamic risk measure by assuming \textit{a priori} the additional property of monotonicity, as it typically constitutes a desirable feature in the evaluation of risk.
We follow the approach of \cite[Section 5]{Quenez+Sulem_2013_BSDEs_jumps_optimization_applications_dynamic_risk_measures}, where, in a Brownian-Poisson filtration, monotonicity is not assumed \textit{a priori}, since it does not follow directly from the structure of a generic BSDE.

In the recent literature, the property of cash-additivity has been proved to be too restrictive in the case of ambiguous interest rates, defaultable contingent claims and/or absence of a zero-coupon bond in the financial market. 
Thus, \cite{ElKaroui+Ravanelli_2009_Cash_sub-additive_risk_measures_interest_rate_ambiguity} suggested employing the axiom of cash-subadditivity, weakening cash-additivity. 
In addition, also positive homogeneity and convexity have been shown to be too restrictive, for example, when an economic agent wants to disentangle diversification and liquidity risk. 
Thus, a recent stream of literature proposed the axiom of star-shapedness (see \cite{Castagnoli+Cattelan+Maccheroni+Tebaldi+Wang_2022_Star-shaped_risk_measures,Laeven+RosazzaGianin+Zullino_2023_Dynamic_return_star-shaped_risk_measures_BSDEs,Laeven+RosazzaGianin+Zullino_2024_Law-invariant_return_star-shaped_risk_measures} for an extensive discussion of static and dynamic star-shaped risk measures).

If $\rho$ is a dynamic risk measure on $L^p(\F_T)$, and $X\in L^p(\F_T)$, we will often work with the naturally defined $\bm\F$-adapted stochastic process ${\rho(X):\Om\times[0,T]\ni(\om,t)\mapsto\rho_t(X)(\om)\in \R}$. 
This notation highlights the connection to BSDEs, which will be examined in the next subsection.

\subsection{Backward stochastic differential equations}
\label{SEC:BSDEs}

Throughout the remainder of the paper, we assume that $T>0$ is a fixed finite time horizon.
We now provide the main results on BSDEs that we will use in this paper. 
\begin{definition}
\label{DEF:solution_BSDE_jumps}
    We call a driver a $\mathcal P\otimes\mathscr B(\R)\otimes\mathscr B(\R^m)\otimes\mathscr B(\Lambda^2)$-measurable function 
    \[
        g:\Om\times[0,T]\times\R\times \R^m\times \Lambda^2\ni(\om,t,y,z,u)\mapsto g(\om,t,y,z,u)\in \R.
    \]
    If $g$ is a driver, $T'\in(0,T]$ and $X:\Om\to\R$ is $\F_{T'}$-measurable, then a solution to the BSDE with parameters $(g,T',X)$ is a triple $(Y,Z,U)\in\mathcal S_{T'}^2\times \mathcal H_{T'}^2(\R^m)\times\mathcal H_{T'}^2(\Lambda^2)$ such that, for any $t\in[0,T']$,
    \begin{equation}
    \label{EQ:BSDE}
        Y_t=X+\int_t^{T'}g(s, Y_s, Z_s, U_s)\d s - \int_t^{T'} Z_s\cdot \d W_s-\int_{(t,T']\times \R^d_\ast }U_s{(x)} \d\tilde N(s,x).
    \end{equation}
\end{definition}
In equation~\eqref{EQ:BSDE}, the integral with respect to the random measure $\tilde N$ is to be understood as in  \cite{Tang+Li_1994_Necessary_Conditions_Optimal_Control_Stochastic_Systems_Random_Jumps};
see also \cite[Chapter~13]{Cohen+Elliott_2015_Stochastic_calculus_applications}.
Let us recall from \cite[Section II.3]{Ikeda+Watanabe_1989_Stochastic_differential_equations_diffusion_processes} that, if $U\in \mathcal H^2_T(\Lambda^2)$, then the process $\int_{[0,\,\cdot\,]\times\R^d_\ast}U\d\tilde N$ is a square integrable $\bm\F$-martingale.

In the following, we occasionally restrict our analysis to a Brownian filtration. 
In this case, a \textit{Brownian} driver is defined as a $\mathcal P^W\otimes\mathscr B(\R)\otimes\mathscr B(\R^m)$-measurable function ${g:\Om\times[0,T]\times\R\times \R^m\to \R}$.
Moreover, for $T'\in(0,T]$ and $X\in L^2\big(\F^W_{T'}\big)$, a solution to the \textit{Brownian} BSDE with parameters $(g,T',X)$ is a pair $(Y,Z)\in\mathcal S_{T'}^2\times \mathcal H_{T'}^2(\R^m)$ such that, for any $t\in[0,T']$,
\begin{equation}
\label{EQ:brownian_BSDE}
    Y_t=X+\int_t^{T'}g(s, Y_s, Z_s)\d s - \int_t^{T'} Z_s\cdot \d W_s.
\end{equation}

\begin{remark}
\label{REM:induced_dyn_risk_meas}
    Consider a driver $g$ and a time $T'\in[0,T]$.
    If, for any $X\in L^2(\F_{T'})$, the BSDE \eqref{EQ:BSDE} with parameters $(g,T',X)$ admits a solution, then the first components of these solutions naturally induce a conditional non-linear expectation, see \cite{Tang+Wei_2012_Representation_dynamic_time-consistent_convex_risk_measures_jumps,Peng_1997_Backward_SDE_related_g-expectation}, or, according to Definition~\ref{DEF:dynamic_risk_measure}, a dynamic risk measure.
    Indeed, if $\rho(X)\in \mathcal S^2_{T'}$ denotes the first component of the solution to the BSDE with parameters $(g,T',X)$, then the function ${\rho_t:L^2(\F_{T'})\ni X \mapsto \rho_t(X)\in L^2(\F_t)}$, is well-defined for any $t\in[0,T']$, and the family $\rho:=(\rho_t)_{t\in[0,T']}$ is a dynamic risk measure on $L^2(\F_{T'})$.
    We say that $\rho$ is induced on $L^2(\F_{T'})$ by the driver $g$.
    The same reasoning holds for the dynamic risk measure on $L^2\big(\F^W_{T'}\big)$ induced by a Brownian driver $g$, for which the Brownian BSDE \eqref{EQ:brownian_BSDE} with parameters $(g,T',X)$ admits a solution for each $X\in L^2\big(\F^W_{T'}\big)$.
    This remark can be generalized to dynamic risk measures on $L^p(\F_{T'})$ (and $L^p\big(\F^W_{T'}\big)$) for any $p\in[1,+\infty]$ if, for any $X\in L^p(\F_{T'})$ (resp.\ $\F^W_{T'}$), a solution to the associated BSDE exists and is proved to be in $\mathcal S^p_{T'}$.
\end{remark}

We recall some theorems of existence, uniqueness and regularity for solutions to BSDEs. 
In the Brownian-Poissonian setting, the following assumptions have been considered.
\begin{enumerate}[label=(\textbf{L}), leftmargin=1.8em]
    \item\label{IT:L_condition}
        Let $g$ be a driver such that the following holds. 
        There exist $K>0$, $\a\in\R$, and a non-negative $\varphi\in L^2_T$, such that, for all $(t,y,z,u)\in [0,T]\times\R\times\R^m\times\Lambda^2$ and all $(y',z',u')\in\R\times\R^m\times\Lambda^2$:
        \begin{align}
            &|g(t,y,z,u)|\leq \varphi_t+K\big(|y|+\norm z + \|u\|_{\Lambda^2}\big),\\
            &|g(t,y,z,u)-g(t,y,z',u')|\leq K\big(\norm{z-z'}+\|u-u'\|_{\Lambda^2}\big),\\
            &(y-y')\big(g(t,y,z,u)-g(t,y',z,u)\big)\leq \a(y-y')^2,\\
            &\R\ni y \mapsto g(\,\cdot\,,y,z,u) \text{ is }\P\otimes\ell_1\text{-a.e. continuous}.
        \end{align}
        Then, for every $X\in L^2(\F_T)$ there exists a unique solution to the BSDE \eqref{EQ:BSDE} with parameters $(g,T,X)$. 
        See \cite{Pardoux_1997_Generalized_discontinuous_backward_stochastic_differential_equations, Royer_2006_Backward_stochastic_differential_equations_jumps_related_non-linear_expectations}.
    \end{enumerate}
    \begin{enumerate}[label=(\textbf{Q})]
    \item\label{IT:Q_condition}
        Assume that the driver $g$ is independent of the $y$-component and that there exist a process $\varphi\in BMO(\R)$, a $\P\otimes\ell_1\otimes\nu$-essentially bounded process $\psi\in BMO(\Lambda^2)$, constants $K_1,K_2,K_3,K_4,\eps> 0$, such that, for all $(t,z,u)\in[0,T]\times\R^m\times\left(\Lambda^2\cap L^\infty(\R^d_\ast, \mathscr B(\R^d_\ast),\nu)\right)$:
        \begin{align}
            &g(t,0,0)=0, \qquad \P\text{-a.s},\\
            &0\leq g(t,z,u)\leq K_1\left(1+\norm{z}^2\right)+K_2\int_{\R^d_\ast}\left[\exp\left(\frac{u(x)}{K_2}\right)-\frac{u(x)}{K_2}-1\right]\d \nu (x), \qquad \P\text{-a.s},\\
            &\R^m\times(\Lambda^2\cap L^\infty(\nu))\ni(z,u)\mapsto g(t,z,u)\text{ is }\P\text{-a.s.\ Fréchet-differentiable in $\R^m\times\Lambda^2$},\\
            &|\partial_zg(t,z,u)|\leq \varphi_t+K_3\norm z, \qquad \P\text{-a.s.},\quad
            |\partial_ug(t,z,u)|\leq \psi_t+K_4|u|,\quad \P\otimes\nu\text{-a.e.},\\
            &\psi_t\geq -1+\eps,\qquad \P\otimes\nu\text{-a.e.}
        \end{align}
        Then, for every $X\in L^\infty(\F_T)$, there exists a unique solution $(Y,Z,U)$ to the BSDE \eqref{EQ:BSDE} with parameters $(g,T,X)$.\footnote{The assumptions can be relaxed further by letting $u$ vary in the Orlicz heart corresponding to the function $\R^d_\ast\ni x \mapsto e^{\norm x}-\norm x-1$. 
        This space is indeed a subset of $\Lambda^2$ larger than $\Lambda^2\cap L^\infty\left(\R^d_\ast,\mathscr B(\R^d_\ast),\nu\right)$, 
        with inclusions dense in the norm $\|\,\cdot\,\|_{\Lambda^2}$.
        Moreover, the hypothesis of Fréchet-differentiability can be replaced by the one of convexity and the partial derivatives by elements of the subgradient.}
        Moreover, $(Y,Z,U)\in \mathcal S^\infty_T\times BMO(\R^m)\times BMO(\Lambda^2)$, and $U$ is $\P\otimes\ell_1\otimes\nu$-essentially bounded.
        See \cite[Appendix~A]{Laeven+Stadje_2014_Robust_portfolio_choice_indifference_valuation} for this result, while a similar statement can be found in \cite[Section~5]{ElKaroui+Matoussi+Ngoupeyou_2016_Quadratic_Exponential_Semimartingales_Application_BSDEs_jumps}.
\end{enumerate}
In the Brownian setting instead, the most common settings that grant existence and uniqueness to solutions are the following.
\begin{enumerate}[label=(\textbf{BL})]
    \item
    \label{IT:BL_condition}
        Let $g$ be a Brownian driver such that $g(\,\cdot\,,0,0)\in L^2_T$ and such that, for some $K\geq 0$ and all $(y,z),(y',z')\in\R\times\R^m$:
        \[
                |g(\,\cdot\,,y,z)-g(\,\cdot\,,y',z')|\leq K\big(|y-y'|+\norm{z-z'}\big), \qquad \P\otimes\ell_1\text{-a.e.}
        \]
        Then, for every $X\in L^2(\F^W_T)$, there exists a unique solution $(Y,Z)$ to the Brownian BSDE \eqref{EQ:brownian_BSDE} with parameters $(g,T,X)$.
        Moreover, $Y$ is $\P$-a.s.\ continuous.
        See \cite{Pardoux+Peng_1990_Adapted_solution_backward_stochastic_differential_equation} and \cite[Theorems~4.2.1, 4.3.1]{Zhang_2017_Backward_stochastic_differential_equations}.
\end{enumerate}
\begin{enumerate}[label=(\textbf{BQ})]
    \item
    \label{IT:BQ_condition}
        Let a Brownian driver $g$ satisfy the following assumptions. 
        There exist $K,K'>0$ such that, for any $(t,y,z)\in[0,T]\times\R\times\R^m$ and $(y',z')\in\R\times\R^m$:
        \begin{align}
            &|g(t,y,z)|\leq K\big(1+|y|+\norm z^2\big),\\
            &|g(t,y,z)-g(t,y',z')|\leq K' \big[|y-y'|+(1+|y|+|y'|+\norm{z}+\norm{z})\norm{z-z'}\big].
            \end{align}
        Then, for every $X\in L^\infty(\F^W_T)$, there exists a unique solution $(Y,Z)$ to the Brownian BSDE \eqref{EQ:brownian_BSDE} with parameters $(g,T,X)$.
        Moreover, $(Y,Z)\in \mathcal S^\infty_T\times BMO(\R^m)$ and $Y$ is $\P$-a.s.\ continuous.
        See \cite{Kobylanski_2000_Backward_stochastic_differential_equations_partial_differential_equations_quadratic_growth} and \cite[Theorems~7.2.1,~7.3.3]{Zhang_2017_Backward_stochastic_differential_equations}.
\end{enumerate}

As is well-known, the properties of dynamic risk measures induced via BSDEs are mainly dictated by the corresponding drivers.
In the following list, we make explicit the relation between the driver $g$ and the induced dynamic risk measure (see, for instance, \cite{Tang+Wei_2012_Representation_dynamic_time-consistent_convex_risk_measures_jumps,Quenez+Sulem_2013_BSDEs_jumps_optimization_applications_dynamic_risk_measures,Calvia+RosazzaGianin_2020_Risk_measures_progressive_enlargement_filtration_BSDE_approach} or, in the Brownian setting, \cite{Barrieu+ElKaroui_2009_Pricing_hedging_optimally_designing_derivatives_minimization_risk_measures,Jiang_2008_Convexity_translation_invariance_subadditivity_g-expectations_related_risk_measures,Laeven+RosazzaGianin+Zullino_2023_Dynamic_return_star-shaped_risk_measures_BSDEs,RosazzaGianin_2006_Risk_measures_g-expectations}). 
Let $g$ be a driver satisfying the condition~\ref{IT:L_condition} and 
let $\rho$ denote the dynamic risk measure induced by $g$ on $L^2(\F_T)$, 
as explained in Remark~\ref{REM:induced_dyn_risk_meas}.
Then the following properties hold.
\begin{enumerate}[noitemsep, topsep=0pt, leftmargin=1.3em]
    \item Zero-one law: If $g(\,\cdot\,,0,0,0)=0$, then $\rho_t(\1_AX)=\1_A\rho_t(X)$ for any $t\in[0,T]$, $A\in\F_t$ and $X\in L^2(\F_T)$. 
    In particular, $\r$ is normalized.
    \item \label{PAGE:flow_property} Flow property:\footnote{Note that, in a Brownian setting and under the normalization condition on the driver $g(\,\cdot\,,0,0)=0$, the flow property reduces to time consistency as formulated on page~\pageref{tc} (see \cite{DiNunno+RosazzaGianin_2024_Fully_dynamic_risk_measures_horizon_risk_time-consistency_relations_BSDEs_BSVIEs} for details).}  
    For all $s\in(0,T]$, if $(\rho_{t,s})_{t\in[0,s]}$ is the dynamic risk measure induced by $g$ on $L^2(\F_s)$, then for any $X\in L^2(\F_T)$ and $t\in[0,s]$, we have $\rho_{t,T}(X)=\rho_{t,s}\big(\rho_{s,T}(X)\big)$.
    \item Monotonicity: $\r$ is monotone if $g$ satisfies the following condition.
        \begin{enumerate}[label=$\mathbf{(C)}$]
            \item 
            \label{IT:hp_comparison_jump}
                There exist $C_1\in(-1,0]$ and $C_2\geq 0$ such that, for all $(y,z)\in\R\times\R^m$ and all $u_1,u_2\in \Lambda^2$, there exists a $\mathcal P$-measurable process $\Gamma:\Om\times[0,T]\to \Lambda^2$ such that 
            \vspace{-0.5em}
            \begin{align*}
                &C_1(1\wedge |x|)\leq \Gamma_t(\om)(x)\leq C_2(1\wedge|x|), \qquad \P\otimes\ell_1\text{-a.e. }(\om,t), \ \nu\text{-a.e. }x\\
                &g(\,\cdot\,,y,z,u_1)-g(\,\cdot\,,y,z,u_2)\leq \big\langle u_1-u_2, \Gamma\big\rangle_{\Lambda^2}^{}.
            \end{align*}
        \end{enumerate}
    \vspace{-1em}
    \item Convexity: 
    \begin{itemize}[label=\scriptsize$\bullet$, noitemsep, topsep=0pt,leftmargin=1em]
        \item 
            If $g$ is $\P\otimes\ell_1$-a.e.\ convex in $(y,z,u)$,
            then $\rho$ is convex.\label{PAGE:convex}
        \item 
            If, for $\P\otimes\ell_1$-a.e.\ $(\om,t)$ we have $g(\om,t,y,0,0)=0$ for $\ell_1$-a.e.\ $y\in\R$, then $\rho$ is convex if and only if $g$ is $\P\otimes\ell_1$-a.e.\ convex in $(y,z,u)$.
        \item 
            If $g$ satisfies the condition \itemref{IT:hp_comparison_jump}, then $\rho$ is convex if and only if $g$ is independent of $y$ (i.e., it is defined on $\Om\times[0,T]\times\R^m\times\Lambda^2$), it is $\P\otimes\ell_1$-a.e.\ convex in $(z,u)$, and $g(\,\cdot\,,0,0)=0$. 
    \end{itemize}
    \item Positive homogeneity: 
        \begin{itemize}[label=\scriptsize$\bullet$, noitemsep, topsep=0pt,leftmargin=1em]
            \item   
                If $g$ is $\P\otimes\ell_1$-a.e.\ positively homogeneous in $(y,z,u)$,
                then $\rho$ is positively homogeneous.
            \item If $\P\otimes\ell_1$-a.e., we have $g(\,\cdot\,,y,0,0)=0$ for $\ell_1$-a.e.\ $y\in\R$, then $\rho$ is positively homogeneous if and only if $g$ is $\P\otimes\ell_1$-a.e.\ positively homogeneous in $(y,z,u)$. 
        \end{itemize}        
    \item Cash-additivity: If $g$ does not depend on $y$, then $\rho$ is cash-additive.
\end{enumerate}

The previous results still hold for the Brownian filtration, if $g$ satisfies the condition~\ref{IT:BL_condition}.
However, for such a Brownian driver $g$ and its induced dynamic risk measure $\rho$, the following additional properties hold.
\begin{enumerate}[resume, noitemsep, topsep=0pt, leftmargin=1.4em]
    \item[$6'$.]
        Cash-additivity: If $\P\otimes\ell_1$-a.e., we have $g(\,\cdot\,,y,0)=0$ for all $y\in\R$, then $\rho$ is cash-additive if and only if $g$ is independent of $y$.
    \item 
        Star-shapedness: $\rho$ is star-shaped if and only if $g$ is $\P\otimes\ell_1$-a.e.\ star-shaped in $(y,z)$.
\end{enumerate}

\begin{remark}
\label{REM:condition_C}
    Let us note that condition \itemref{IT:hp_comparison_jump} is required to apply the BSDE comparison theorem in the Brownian-Poissonian filtration, see \cite[Theorem~2.5]{Royer_2006_Backward_stochastic_differential_equations_jumps_related_non-linear_expectations}.
    If one reduces to the Brownian setting, then this condition is trivially satisfied, and can be omitted. 
    In other words, a dynamic risk measure induced by a Brownian driver is automatically monotone.
\end{remark}

There are sufficient conditions under which dynamic cash-additive risk measures can be represented via BSDEs; see \cite[Proposition~3.3]{Tang+Wei_2012_Representation_dynamic_time-consistent_convex_risk_measures_jumps} and \cite[Theorem~4.6]{Royer_2006_Backward_stochastic_differential_equations_jumps_related_non-linear_expectations}, or \cite{Peng_2004_Nonlinear_expectations_nonlinear_evaluations_risk_measures,Peng_2005_dynamically_consistent_nonlinear_evaluations,RosazzaGianin_2006_Risk_measures_g-expectations} for the Brownian setting.\footnote{In the Brownian setting, similar arguments can be applied to positively homogeneous risk measures; see \cite{Laeven+RosazzaGianin+Zullino_2025_Geometric_bsdes}.
In particular, it is possible to employ the one-to-one correspondence in \cite{Laeven+RosazzaGianin+Zullino_2025_Geometric_bsdes} to derive sufficient conditions under which a positively homogeneous risk measure can be induced via geometric BSDEs.}

\subsection{Malliavin calculus}

For $n\in\N$, let us denote by $C^{\infty}_b(\R^n)$ the space of all smooth functions $f:\R^n\to\R$ that are bounded and have bounded derivatives of all orders. 
We let $\mathscr S$ denote the space of all random variables of the form 
\[
    F=f\left(\int_0^T\varphi_1(t)\cdot\d W_t,\dots, \int_0^T\varphi_n(t)\cdot\d W_t\right),
\]
for some $n\in\N$, $f\in C^\infty_b(\R^n)$ and $\varphi_1,\dots,\varphi_n\in L^2(0,T;\R^m)$. 
We call a \textit{Malliavin derivative operator} a mapping $D:\mathscr S\to L^2_T$ such that:
\[
    DF:=\sum_{j=1}^n\varphi_j\partial_jf\left(\int_0^T\varphi_1(t)\cdot\d W_t,\dots, \int_0^T\varphi_n(t)\cdot\d W_t\right), \qquad \forall\, F\in\mathscr S.
\]
Next, we introduce the norm $\|F\|_{1,2}:=\mathbb{E}\left[|F|^2+\int_0^T|D_tF|^2\d t\right]$, for $F\in \mathscr S$, and we let $\mathbb D^{1,2}$ denote the completion of $\mathscr S$ in $L^2\big(\F^W_T\big)$ under $\|\,\cdot\,\|_{1,2}$.
It can be shown, see, e.g., \cite{Nualart_2006_Malliavin_calculus_related_topics}, that $D$ is a densely defined, closed linear operator from $\mathscr S\subset \mathbb D^{1,2}$ to $L^2_T$, with a unique extension to $\mathbb D^{1,2}$, still denoted by $D$.
For each $F\in \mathbb D^{1,2}$, we call a Malliavin derivative of $F$ the stochastic process $DF:\Om\times [0,T]\ni(\om,t)\mapsto (D_tF)(\om)\in\R$. 

We now recall two basic results that will be used later in the paper, without further mentioning. 
For every $X\in\mathbb{D}^{1,2}$, any $h\in C^1(\R)$, and $\ell_1$-a.e. times $t,s\in[0,T]$, we have: $D_t\mathbb{E}\left[\left.X\right|\mathcal F^W_s\right]=\mathbb{E}\left[\left.D_tX\right|\F^W_s\right]\1_{[0,s]}(t)$ and $D_th(X)=h'(X)D_t X$.

The general theory of Malliavin calculus can be found in \cite{DiNunno+Oksendal+Proske_2009_Malliavin_calculus_Levy_processes_applications_finance,Nualart_2006_Malliavin_calculus_related_topics}, among others.
An important stream of research has also investigated the link between the solutions of BSDEs and Malliavin calculus. 
Specifically, it has been shown that, under various assumptions on the parameters of the Brownian BSDE, it is possible to express the second component of the solution in terms of the Malliavin derivative of the first component. 
In this subsection, we do not detail the exact results, as they are established under different sets of assumptions in the literature. 
Instead, we present a brief list of the most relevant works for our purposes that address this topic, namely \cite{ElKaroui+Peng+Quenez_1997_Backward_stochastic_differential_equations_finance, Wu+Yu_2012_Backward_stochastic_viability_related_properties_Z_BSDEs_applications, Ma+Zhang_2002_Representation_theorems_backward_stochastic_differential_equations}. 
We mention that Malliavin calculus can also be defined in the Brownian-Poissonian filtration and that, in this setting, also the third component of the solution to a BSDE can be characterized as the Brownian-Poissonian Malliavin derivative of the first component, see \cite{Oksendal+Sulem_2019_Applied_stochastic_control_jump_diffusions,DiNunno+Oksendal+Proske_2009_Malliavin_calculus_Levy_processes_applications_finance}.

\section{A Measure of Financial Resilience}
\label{sec:insmeares}
Fix $T\in(0,+\infty)$ and define the set
\[
    \mathcal{T}_T:=\big\{\bm\F\textit{-stopping times } \tau:\Omega\to[0,T] \ : \ \P(\tau<T)>0\big\}.
\]
Let $\rho$ be a dynamic risk measure on $L^p(\F_T)$, for some $p\in[1,+\infty]$, and let $X\in L^p(\mathcal{F}_T)$ be a risky claim.
For $c \in \mathbb{R}$, let us define 
\[
    \tau^c := T\wedge\inf\big\{t \in [0, T] : \ \rho_t(X) \geq c\big\},
\]
with the convention $\inf\{\emptyset\}=+\infty$. 

Clearly, $\tau^c$ is an $\bm\F$-stopping time, which can be interpreted as the first time at which the risk of the claim $X$, evaluated through the dynamic risk measure $\rho$, exceeds the threshold $c \in \mathbb{R}$.
The threshold level can be chosen in such a way that the probability of a hit in $[0,T)$ is non-zero. 
Then, $\t^c\in\mathcal T_T$.
When $c \equiv 0$ and $\rho$ is a dynamic normalized cash-additive risk measure (cf.\ \cite{Delbaen_2012_Monetary_Utility_Functions, Follmer+Schied_2016_Stochastic_finance}), $\tau^0$ admits a natural interpretation: it represents the first time the risk of $X$ becomes unacceptable, i.e., it breaches the risk-acceptance set, requiring an additional risk capital margin to cover the risk of insolvency, as stipulated e.g., by financial regulators. 
An increasingly important regulatory concern is how long this capital margin must be maintained in the financial portfolio, and whether it is expected to increase or decrease in the near future. 
We note that the threshold level $c$ may be generalized to depend on time, i.e., to be maturity-dependent, in a continuous and deterministic way.

In this section, we introduce a measure of financial resilience, capturing the expected \textit{rate} at which the risk associated with $X$, evaluated using the \textit{dynamic risk measure} $\rho$, \textit{recovers} (i.e., \emph{bounces back}) upon the occurrence of a \textit{stress scenario}.
Clearly, not all dynamic risk measures will exhibit favorably signed resilience. 
In the sequel, we provide a precise formal definition of our measure of financial resilience.

\subsection{Resilience rate}

Assume that $\rho(X)$ is $\bm\F$-progressively measurable.
Then, we can consider the family of random variables (thanks to \cite[Proposition~2.18]{Karatzas+Shreve_1991_Brownian_motion_stochastic_calculus})
\[
    \rho^{\varepsilon}_{\tau^c}(X):=\frac{1}{\varepsilon} \big( \rho_{(\tau^c + \varepsilon){\wedge T}}(X) - \rho_{\tau^c}(X) \big), \qquad \forall \, \eps>0.
\]
If $\rho^{\varepsilon}_{\tau^c}(X)$ is positive for all $\eps>0$, then the risk evaluation of $X$ increases after $\tau^c$, meaning $\rho(X)$ is not recovering. 
On the other hand, if $\rho^{\varepsilon}_{\tau^c}(X)$ is negative for all $\eps>0$, then the risk evaluation of $X$ diminishes, indicating that $\rho(X)$ is exhibiting favorably signed resilience.

While the limit as $\eps\to 0^+$ of $\rho^{\varepsilon}_{\tau^c}(X)$ might seem appealing, it is well-known that stochastic processes (and solutions to BSDEs) generally do not possess differentiable paths. 
For this reason --- and also because we aim to identify a scenario-independent quantity, i.e., one that is immediately readable and interpretable for regulatory purposes, and allows for an ``instantaneous'' evaluation of resilience --- we propose defining the expected time-derivative of a dynamic risk measure at a stopping time, in the sense specified by the following definition.

\begin{definition}[Resilience rate]
\label{DEF:resilience_rate_jumps}
    Let $\rho$ be a dynamic risk measure on $L^p(\F_T)$, for some $p\in[1,+\infty]$, and let $X\in L^p(\F_T)$. 
    Assume that the process $\rho(X)$ is of class~$(D)$.
    For $\t\in\mathcal T_T$, we define the \textbf{resilience rate} of $\rho_\t(X)$ as the limit
    \begin{equation*}
        \bm\dot \rho_\t(X):=\lim_{\varepsilon\to0^+}\frac{1}{\varepsilon}\mathbb{E}\left[\rho_{(\tau+\varepsilon){\wedge T}}(X)-\rho_{\tau}(X){\big|\tau<T}\right],
    \end{equation*}
    whenever it exists in $\overline\R:=\R\cup\{\pm\infty\}$.
    Similarly, for $\s,\t\in\mathcal T_T$ such that $\s\leq\t$, we define the \textbf{conditional resilience rate} of $\rho_\t(X)$, conditioned at time $\s$, as  
    \begin{equation*}
        \bm\dot\rho_{\tau|\s}(X):=
        L^p(\F)\text{ \!-\!}\lim_{\varepsilon\to0^+}\frac{\E\left[\rho_{(\tau+\varepsilon)\wedge T}(X)-\rho_{\tau}(X)\big|\F_\s\right]}{\eps\,\P(\tau<T)},
    \end{equation*}
    whenever the limit exists in $L^p(\F)$.
    Here, $\mathcal{F}_{\s}$ denotes the sigma-algebra of $\s$-past, which is defined as $\F_\s:=\big\{A\in \F \ : \ A\cap \{\s\leq t\}\in\F_t, \ \forall\, t\in[0,T] \big\}$.
\end{definition}

We point out that, if $\tau\in\mathcal T_T$, then $(\tau+\eps)\wedge T$ is a $[0,T]$-valued $\bm\F$-stopping time for every $\eps>0$.
Consequently, if $\rho(X)$ is of class~$(D)$, both $\rho_{(\tau+\varepsilon)\wedge T}(X)$ and $\rho_{\tau}(X)$ are well-defined, integrable random variables and their difference converges to $0$ in $L^1(\F)$ as $\eps\to 0^+$.
This observation gives meaning to the quantities under consideration.
Nevertheless, it does not suffice to guarantee the existence of the limits appearing in the definition.
The assumptions under which these limits do exist will be analyzed in Sections~\ref{SEC:bouncing_drift}, \ref{subsec:VT} and in Appendix~\ref{APP:verif_assu_brownian}.

Whenever $\bm\dot\rho_{\tau}(X)$ is well-defined, its financial interpretation is clear: it measures the expected instantaneous rate of change for the risk evaluation of $X$ immediately after the condition specified by the stopping time $\tau$ is met, provided this occurs before maturity.
If $\bm\dot\rho_{\tau}(X) < 0$, we say that the risk evaluation of $X$ is \emph{resilient}, since $\rho_{\tau}(X)$ is expected to recover (i.e., decrease) directly after $\tau$.
Instead, if $\bm\dot\rho_{\tau}(X) > 0$, the risk evaluation is expected to increase, and the position $X$ does not exhibit resilience.
Finally, whenever $\bm\dot\rho_{\tau}(X)=0$, the risk evaluation is said to be \emph{resilience neutral}.
We further discuss the role of resilience neutrality in Section~\ref{SEC:interpretation}.

It is worth noting that if $\rho(X)$ is a class~$(D)$ submartingale (resp.\ supermartingale, martingale), then for any $\bm\F$-stopping times $\s\leq \t\leq T$ the optional sampling theorem (cf. \ \cite[Theorem~17, Chapter~I]{Protter_2005_Stochastic_integration_differential_equations}) yields $\E[\rho_{\t}(X)]\geq \E[\rho_{\s}(X)]$ (resp.\ $\leq$, $=$). 
In particular, for any $\tau\in \mathcal T_T$, the function $\eps\mapsto \E\big[\rho_{(\t+\eps)\wedge T}(X)\big]$ is non-decreasing (resp.\ non-increasing, constant), thus, whenever $\bm\dot\rho_\tau(X)$ is well-defined, it is non-negative (resp.\ non-positive, zero).
We refer to \cite[Section~3.1]{Delbaen_2006_Structure_m-stable_sets_particular_set_risk_neutral_measures} for sufficient conditions on a dynamic risk measure that ensure the existence of a càdlàg modification satisfying the submartingale property.

\begin{example}[VaR and ES]
\label{EX:VaR_ES}
For $\alpha\in(0,1)$, let us define the dynamic version of the Value-at-Risk as follows (see, e.g., \cite{Vogelpoth_2006_some_results_dynamic_risk_measures, Follmer+Schied_2016_Stochastic_finance}):
\[
    \operatorname{VaR}_t^\a(X)
    :=\essinf \left\{Y\in L^2\big(\F^W_t\big) \ : \ \P\big(X\geq Y\big|\F^W_t\big)\leq\alpha\right\}, \quad \forall\, X\in L^2\big(\F_T\big), \ t\in[0,T].
\]
Then we can compute the resilience rate of the (dynamic) Value-at-Risk for the $\F^W_T$-measurable random variable
\[
    X=x+\int_0^T\mu_s\d s + \int_0^T\s_s\d W_s,
\]
where $x\in \R$, and $\mu,\s:[0,T]\to \R$ are continuous functions. 
Since $X$ is normally distributed, it is easy to verify that, for all $t\in[0,T]$,
\begin{equation}
\label{EQ:VAR}
    \operatorname{VaR}_t^\a(X)
    =x+\int_0^T\mu_s\d s + \int_0^t\s_s\d W_s - z_\a\sqrt{\int_t^T\s^2_s\d s},
\end{equation}
where $z_\a$ is the $\alpha$-quantile of the standard normal distribution.
A straightforward computation shows that, for any $\t\in\mathcal T_T$,
\[
    \bm\dot{\operatorname{VaR}^\a_\t}(X)=
    \!
    \lim_{\eps\to0^+}
    \!
    -\frac{z_\a}{\eps}\E\left[\left.\sqrt{\int_{{(\t+\eps)}\wedge T}^T
    \!\!\!\!\!\!\!\!\!
    \s^2_s\d s}-\sqrt{\int_\t^T\s^2_s\d s}\right|\t<T\right]
    \!
    =\frac{z_\a}{2}\E\left[\left.\frac{\s^2_\t}{\sqrt{\int_\t^T\s^2_s\d s}}\right|\t<T\right].
\]
The result is negative if $\a\in(0,1/2)$, zero at $\alpha=1/2$, and positive otherwise.

Similarly, let us define the dynamic version of the Expected Shortfall as follows:
\[
    \operatorname{ES}_t^\a(X):=\frac1\a\int_0^\a\operatorname{VaR}_t^\g(X)\d\g, \qquad \forall\, X\in L^2\big(\F_T\big), \ t\in[0,T].
\]
For the same choice of $X$ as before, we have
\begin{equation}
\label{EQ:ES}
    \operatorname{ES}_t^\a(X)
    =x+\int_0^T\mu_s\d s + \int_0^t\s_s\d W_s +\frac1\a\varphi(z_\a)\sqrt{\int_t^T\s^2_s\d s},
\end{equation}
where $\varphi$ is the probability density function of the standard normal distribution, and we used the known relation $x\varphi(x)=-\varphi'(x)$ to integrate $z_\g$ for $\gamma\in(0,\a]$.
A straightforward computation shows that, for any $\t\in\mathcal T_T$,
\[
    \bm\dot{\operatorname{ES}^\a_\t}(X)
    =-\frac{\varphi(z_\a)}{2\a}\E\left[\left.\frac{\s^2_\t}{\sqrt{\int_\t^T\s^2_s\d s}}\right|\t<T\right].
\]The result is negative for any choice of $\a\in(0,1)$.  
\end{example}

In Definition~\ref{DEF:resilience_rate_jumps}, we may limit our study to the case of deterministic times, namely for $\tau(\om)=t$ and $\s(\om)=s$ for all $\om\in\Om$ and some $0\leq s\leq t<  T$.
In this setting, the assumption that $\rho(X)$ is of class~$(D)$ and the conditioning on the event $\{\tau<T\}=\Om$ can be omitted.
Moreover, for any $t\in[0,T)$ there exists $\delta>0$ such that, for all $\eps\in(0,\delta)$, we have $t+\eps<T$. 
Thus, the time evaluation of the difference quotient can be done at $t+\eps$ instead of $(t+\eps)\wedge T$.
Consequently, for deterministic times, the resilience rate and its conditional version reduce to:
\begin{align}
    &\bm\dot \rho_t(X):=\lim_{\varepsilon\to0^+}\frac{1}{\varepsilon}\mathbb{E}\left[\rho_{t+\eps}(X)-\rho_{t}(X)\right],\\
    &\bm\dot\rho_{t|s}(X):=
    L^p(\F)\text{ \!-\!}\lim_{\varepsilon\to0^+}\frac1\eps\E\left[\rho_{t+\eps}(X)-\rho_{t}(X)\big|\F_s\right].
\end{align}
\begin{remark}
\label{REM:rate_derivative}
    For $t,s\in[0,T]$ with $s\leq t$, both the resilience rate $\bm\dot\r_t(X)$ and its conditional version $\bm\dot\r_{t|s}(X)$ are well-defined and finite if and only if the respective maps
    \[
        [0,T)\ni r\mapsto \E[\r_r(X)]\in\R, \qquad [0,T)\ni r\mapsto \E[\r_r(X)|\F_s]\in L^p(\F),
    \]
    are right-differentiable at $t$. 
    In this case, they equal the respective right-derivatives at $t$.
    Analogously, if $\t\in\mathcal T_T$, then the resilience rate $\bm\dot\r_\t(X)$, and its conditional version $\bm\dot\r_{\t|\s}(X)$, for fixed $\s\in\mathcal T_T$ \textit{s.t.} $\s\leq\t$, are well-defined and finite if and only if the respective maps 
    \begin{align}
        L^\infty(\F)\supset\mathcal T_T\ni \nu&\mapsto \E[\r_{\nu\wedge T}(X)|\nu<T]\in\R,\\
        L^\infty(\F)\supset\mathcal T_T\ni \nu&\mapsto \frac{\E[\r_{\nu\wedge T}(X)|\F_\s]}{\P(\nu<T)}\in L^p(\F),
    \end{align}
    are right-Gateaux-differentiable at $\t$ in the direction $\1_\Om$.
    In this case, the resilience rates equal the right-Gateaux-derivatives at $\t$ in the direction $\1_\Om$ of the respective functions.
\end{remark}
\begin{remark}
\label{REM:relation_conditional_resilience_jumps}
    Let us remark that, in the setting of Definition~\ref{DEF:resilience_rate_jumps}, if $\bm\dot\rho_{\t|\t}(X)$ exists in $L^p(\F)$, then $\bm\dot\rho_{\t}(X)$ exists, is finite and $\bm\dot\rho_{\t}(X)=\E[\bm\dot\rho_{\t|\t}(X)]$.
\end{remark}

We also highlight that our measure of resilience $\bm\dot\rho_\tau(X)$, for any fixed choice of the parameters for which it is well-defined, is a deterministic quantity independent of the paths of the risk measure $\rho$ or of the random realizations of $X$ and $\t$.
A random quantity arises only when considering the conditional resilience measure $\bm\dot\rho_{\t|\s}(X)$, where the source of randomness stems from the conditional expectation of the rate at the future time $\t$, given the information available at the present time $\s$.
Consequently, the conditional version can be viewed as a stochastic process either as $\big(\bm\dot\rho_{t|s}(X)\big)_{s\in[0,t]}$ for fixed $t\in(0,T)$, or as $\big(\bm\dot\rho_{t|s}(X)\big)_{t\in[s,T)}$ for fixed $s\in[0,T)$.

Let us note that, in the first part of Definition~\ref{DEF:resilience_rate_jumps}, one could alternatively consider the conditional expectation to the sigma-algebra generated by the event $\{\tau<T\}$, instead of the expectation conditioned to the same event.
This alternative would entail defining the resilience rate as a random variable that coincides with the current definition of $\bm\dot\rho_\tau(X)$ on the event $\{\tau < T\}$, and is identically zero on the complementary event $\{\tau = T\}$.
However, the latter is not relevant for the purposes of our study, as the recovery rate of the risk evaluation of the claim $X$ is only meaningful before maturity.
Therefore, we have opted for the current definition.

Note also that our measure of financial resilience, just like the dynamic risk measure it is derived from, reflects possibly subjective risk preferences, with decision-theoretic foundations, and is not a purely objective statistical property of the underlying stochastic financial process.
Further discussion on this point will be provided in the examples; see Section~\ref{sec:examples}.

To aid the mathematical formalization, we present a preliminary proposition that can be regarded as an extension of Proposition~2.2 in \cite{Jiang_2008_Convexity_translation_invariance_subadditivity_g-expectations_related_risk_measures} which, instead, is limited to the statement \itemref{IT:PROP:integral_average_1} for the case $q>1$ (with strict inequality).

\begin{proposition}
\label{PROP:integral_average}
    The following statements hold true.
    \begin{enumerate}[label=(\roman*)]
    \item \label{IT:PROP:integral_average_1}
         Let $\psi\in L^q_T$ for some $q\geq 1$. 
         For $\ell_1$-a.e.\ $t\in[0,T)$, we have:
        \begin{align} 
        \label{EQ:PROP:integral_average_determ_time}
            \psi_t&=L^q(\F)\text{ \!-\!}\lim_{\varepsilon\to0^+}\frac{1}{\varepsilon}\int_{t}^{t+\varepsilon}\psi_r\d r,\\
            \E[\psi_{t}|\F_s]
            &=L^q(\F)\text{ \!-\!}\lim_{\varepsilon\to0^+}\E\left[\left.\frac{1}{\varepsilon}\int_{t}^{t+\varepsilon}\psi_r\d r\right|\F_s\right], \qquad \forall\, s\in[0,t].
        \end{align}
        The limits hold for all $t\in[0,T)$, if $\psi$ has $\P$-a.s.\ right-continuous trajectories.
    \item \label{IT:PROP:integral_average_2}
        Assume that the $\F\otimes\mathscr B\big([0,T]\big)$-measurable process $\psi:\Om\times[0,T]\to\R$ satisfies the following conditions.
        \begin{itemize}[noitemsep, topsep=0pt, leftmargin=1.3em]
            \item $\P$-a.s., for any $t\in[0,T)$, the right limit $\psi_{t^+}:=\lim_{s\to t^+}\psi_s$ exists.
            \item $\psi$ has $\P$-a.s.\ at most countably many discontinuities.
            \item There exists $q\geq 1$ such that $\E\Big[\sup_{s\in[0,T]}|\psi_s|^{q}\Big]<+\infty$.
        \end{itemize}
        Then, for any $[0,T]$-valued $\bm\F$-stopping times $\tau,\s$ such that $\s\leq\t$, we have:
        \begin{align}
        \label{EQ:PROP:integral_average_stopping_time}
                 \psi_{\tau^+}
                 &=L^q(\F)\text{ \!-\!}\lim_{\varepsilon\to0^+}\frac{1}{\varepsilon}\int_{\tau}^{\tau+\varepsilon}\psi_{r{\wedge T}}\d r,\\
                 {\E[\psi_{\t^+}|\F_\s]}&{=L^q(\F)\text{ \!-\!}\lim_{\varepsilon\to0^+}\mathbb{E}\left[\left.\frac{1}{\eps}\int_{\tau}^{\tau+\eps}\psi_{r\wedge T}\d r\right|\mathcal{F}_{\s}\right],}
        \end{align}
        where $\psi_{\tau^+}(\omega):=\ds\lim_{s\to (\tau(\om))^+}\psi_{s{\wedge T}}(\om)$, $\P\text{-a.e. }\om\in\Om.$
    \end{enumerate}
\end{proposition}
\begin{proof}
    The first convergence in statement~\itemref{IT:PROP:integral_average_1} directly follows from the Lebesgue differentiation theorem for the Bochner-integrable function $[0,T]\ni t \mapsto \psi_t\in L^q(\F)$ (cf.\ \cite[Theorem~9, Chapter~II]{Diestel+Uhl_1977_Vector_measures}).
    
\textit{Step $1$.}
    We now proceed to prove the first convergence in the second statement. 
    To begin with, we show that $\sup_{[0,T]}|\psi|^q$ is $\F$-measurable, thus giving sense to its expectation.
    For $\P$-a.e.\ $\om\in\Om$, let $D(\om)\subset[0,T]$ denote the countable set of discontinuities of $\psi(\om)$.
    Then $\psi(\om)$ is continuous at any $t\in C(\om):=[0,T]\setminus D(\om)$.
    The supremum of $|\psi(\om)|^q$ over the set $C(\om)$ equals the supremum over the dense countable subset $C(\om)\cap \Q$, thanks to the continuity of $|\psi(\om)|^q$ over $C(\om)$.
    Therefore, we have:
    \[
        \sup_{t\in[0,T]}|\psi_t(\om)|^q=\max\left\{\sup_{t\in D(\om)}|\psi_t(\om)|^q, \sup_{t\in\Q\cap C(\om)}|\psi_t(\om)|^q\right\},
    \]
    which shows that $\sup_{t\in[0,T]}|\psi_t|^q$ is the supremum over a countable family of $\F$-measurable random variables, hence is $\F$-measurable itself.
    
    Let us now fix an infinitesimal sequence of real positive numbers $(\eps_n)_{n\in\N}$.
    We aim to apply Lebesgue's dominated convergence theorem, with respect to integration in probability, to the sequence of random variables
    \begin{equation}
    \label{EQ:PROP:integral_average_sequence}
        \left(\left|\frac{1}{\eps_n}\int_\tau^{\tau+\eps_n}\psi_{r\wedge T}\d r -\psi_{\tau^+}\right|^q\right)_{n\in\N}.
    \end{equation}
 
    Let $\tilde\Om$ denote the set of all $\om\in\Om$ such that, for $t\in[0,T)$, the right limit of $\psi(\om)$ at time $t$ exists in $\R\cup\{\pm\infty\}$, denoted by $\psi_{t^+}(\om)$.
    Fix $\omega\in\tilde \Omega$.
    For any $n\in\N$, there exist $t^n_1(\omega),t^n_2(\omega)\in [\tau(\omega),\tau(\omega)+\eps_n]$ such that 
    \begin{equation}
    \label{EQ:PROP:integral_average_bounds_sequences}
        \psi_{t^n_1(\omega)\wedge T}(\omega)\leq \frac{1}{\eps_n}\int_{\tau(\omega)}^{\tau(\omega)+\eps_n}\psi_{r\wedge T}(\omega)\d r \leq \psi_{t^n_2(\omega)\wedge T}(\omega).
    \end{equation}    
    Let us momentarily fix $n\in\N$ and prove the existence of $t_1^n(\omega)$ that satisfies the left inequality above.
    The right inequality can then be proved similarly, \textit{mutatis mutandis}.
    Denote, for the sake of brevity, $I:= \inf_{t\in[\tau(\omega),\tau(\omega)+\eps_n]}\psi_{t\wedge T}(\omega)$ and $\langle \psi\rangle:= \frac{1}{\eps_n}\int_{\tau(\omega)}^{\tau(\omega)+\eps_n}\psi_{s\wedge T}(\omega)\d s$. 
    By direct inspection, we have $I\leq \langle \psi\rangle$.
    If $I=-\infty$, then for any $M\in\R$, there exists ${t\in [\tau(\omega),\tau(\omega)+\eps_n]}$ such that $\psi_{t\wedge T}(\omega)\leq M$. 
    Choose $M=\langle \psi\rangle$ and conclude.
    Suppose now that $I\in\R$. 
    If $I<\langle\psi\rangle$, then for any $\eps>0$ there exists $t\in [\tau(\omega),\tau(\omega)+\eps_n]$ such that $\psi_{t\wedge T}(\omega)\leq I+\eps$.
    Choose $\eps=\langle \psi\rangle -I$ and conclude.
    If $I=\langle\psi\rangle$, then ${\int_{\tau(\omega)}^{\tau(\omega)+\eps_n}(\psi_{r\wedge T}(\omega)-I)\d r=0}$, and since the integrand is non-negative, we infer that the function $\psi_\cdot(\omega)$ is constant.
    Were this the case, then the claim would be trivially satisfied.
    
    Moreover, for $i=1,2$, the bounds $\tau(\omega)\leq t^n_i(\omega)\leq \tau(\omega)+\eps_n$, $n\in\N$, imply that $t^n_i(\omega)\longrightarrow \big(\t(\omega)\big)^+$ as $n\to\infty$.
    Hence, for the fixed $\om$ above, we have 
    \[
        \lim_{n\to\infty}\psi_{t^n_i(\omega)\wedge T}(\omega)=\lim_{s\to(\t(\om))^+}\psi_{s\wedge T}(\om)=:\psi_{\tau^+}(\omega), \qquad i=1,2.
    \]
    By the squeeze theorem, the middle term in  \eqref{EQ:PROP:integral_average_bounds_sequences} converges to $\psi_{\tau^+}(\omega)$ as well, as $n\to\infty$.
    Namely, we proved that
    \[
        \lim_{n\to\infty}\frac{1}{\eps_n}\int_{\tau}^{\tau+\eps_n}\psi_{r\wedge T}\d r =\psi_{\tau^+}, \qquad \P\text{-a.s.},
    \]
    in particular, the sequence in \eqref{EQ:PROP:integral_average_sequence} is $\P\text{-a.s.}$ infinitesimal.

    Finally, for any $n\in\N$, we have $\P\text{-a.s.}$
    \begin{align}
       \left|\frac{1}{\eps_n}\int_\tau^{\tau+\eps_n}\psi_{r\wedge T}\d r-\psi_{\t^+}\right|^q
       \leq2^q\sup_{r\in[0,T]}|\psi_r|^q,
    \end{align}
    thanks to the fact that $\tau$ takes values in the interval $[0,T]$. 
    By assumption, the random variable on the right-hand side is integrable with respect to $\P$, and allows us to apply Lebesgue's dominated convergence theorem to the sequence in \eqref{EQ:PROP:integral_average_sequence}:
    \[
        \lim_{n\to\infty}\E\left[\left|\frac{1}{\eps_n}\int_\tau^{\tau+\eps_n}\psi_{r\wedge T}\d r -\psi_{\tau^+}\right|^q\right]=0.
    \]
    
    Eventually, the thesis follows from the arbitrariness of the sequence $(\eps_n)_{n\in\N}$.
    
\textit{Step $2$.}
    The second convergence in statement~\itemref{IT:PROP:integral_average_2} follows from the $\F_\s$-measurability of $\E[\psi_{\t^+}|\F_\s]$ and Jensen's inequality for conditional expectations:
     \begin{align}
        \E\left[\left|\E\left[\left.\frac1\eps\int_\t^{\t+\eps}\psi_{r\wedge T}\d r\right|\F_\s\right]-\E[\psi_{\t^+}|\F_\s]\right|^q\right]
        \leq \E\left[\left|\frac1\eps\int_\t^{\t+\eps}\psi_{r\wedge T}\d r-\psi_{\t^+}\right|^q\right],
    \end{align}
    where the second term is infinitesimal, as $\eps\to 0^+$, by the first convergence in statement~\itemref{IT:PROP:integral_average_2}.
    Similarly, one can prove that the second convergence in the first statement follows from the first convergence.
\end{proof}

\begin{example}
\label{EX:entropic_1}
    We consider the entropic risk measure as a standing example throughout this paper.
    It is intimately connected to the Kullback-Leibler divergence (see \cite{Csiszar_1975_I-divergence_geometry_probability_distributions_minimization_problems,Ben-Tal+Teboulle_1986_Expected_utility_penalty_functions_duality_stochastic_nonlinear_programming,Ben-Tal+Teboulle_1987_Penalty_functions_duality_stochastic_programming_phi_divergence_functionals,Frittelli_2000_Introduction_theory_value_coherent_no-arbitrage_principle,Laeven+Stadje_2013_Entropy_coherent_entropy_convex_measures_risk, Follmer+Schied_2016_Stochastic_finance}), and it is used by \cite{Hansen+Sargent_2001_Robust_Control_Model_Uncertainty} to obtain robustness in macroeconomic models.
    We recall (see, e.g., \cite{Barrieu+ElKaroui_2009_Pricing_hedging_optimally_designing_derivatives_minimization_risk_measures}) that the dynamic entropic risk measure, with parameter $\g>0$, is defined as follows:
    \begin{equation}
    \label{EQ:dyn_entr_risk_meas}
        e^\g_t(X):=\frac1\g\ln\E\big[e^{\g X}\big|\F_t\big],\qquad\forall\, t\in[0,T],
    \end{equation}
    for any $\F_T$-measurable random variable $X$ such that $\E[e^{p X}]<\infty$ for all $p\geq \g$.
    We now proceed to compute its resilience rate, according to Definition~\ref{DEF:resilience_rate_jumps}.
    To this purpose, for fixed $\g$ and $X$ as above, we denote by $M$ the càdlàg version (cf.\ \cite[Corollary~5.1.9]{Cohen+Elliott_2015_Stochastic_calculus_applications})  of the square-integrable martingale
    \begin{equation}
    \label{EQ:EX:entropic_1:def_M}
        M_t:=\E\big[e^{\g X}\big|\F_t\big], \qquad\forall\, t\in[0,T].
    \end{equation}
    Then, by the martingale representation theorem (see, e.g., \cite[Theorem~14.5.7]{Cohen+Elliott_2015_Stochastic_calculus_applications} or \cite[Lemma~4.24, Chapter~III]{Jacod+Shiryaev_2003_Limit_theorems_stochastic_processes}), there exist processes $H\in \mathcal H^2_T(\R^m)$, $K\in \mathcal H^2_T(\Lambda^2)$ such that:
    \begin{equation}
    \label{EQ:EX:entropic_1:martingale_repr_th}
        M_t=\E[e^{\g X}]+\int_0^tH_s\cdot \d W_s + \int_{[0,t]\times\R^d_\ast}K_s(x)\d \tilde N(s,x), \qquad \forall\, t\in[0,T].
    \end{equation}
    By direct inspection, we have $\P(M_t>0)=1$ for all $t$, hence we can apply It\^o's generalized formula (see, e.g., \cite[Theorem~14.2.3]{Cohen+Elliott_2015_Stochastic_calculus_applications}) to explicitly compute $\ln M_t$, yielding:
    \begin{equation}
    \label{EQ:entropic}
    \begin{aligned}
        \ln M_t
        =&\, \ln\E[e^{\g X}]+\int_0^t\frac{H_s}{M_{s^-}}\cdot\d W_s +\int_{[0,t]\times\R^d_\ast}\ln\left(1+\frac{K_s(x)}{M_{s^-}}\right)\d \tilde N(s,x) \\
        &+\int_0^t\left(- \frac12\frac{\norm{H_s}^2}{M^2_{s^-}}+\int_{\R^d_\ast}\left[\ln\left(1+\frac{K_s(x)}{M_{s^-}}\right)-\frac{K_s(x)}{M_{s^-}}\right]\d \nu(x)\right)\d s,
    \end{aligned}
    \end{equation}
    where all the integrals are $\P$-a.s.\ well-defined and absolutely convergent (cf.\ \cite[Lemma~14.2.2]{Cohen+Elliott_2015_Stochastic_calculus_applications}).
    If we further assume that $X$ is bounded, we are able to compute the resilience rate of the dynamic entropic risk measure, at $\ell_1$-a.e. $t\in[0,T)$, by means of Proposition~\ref{PROP:integral_average}\itemref{IT:PROP:integral_average_1}:
    \begin{align}
        \bm\dot e^\g_t(X)
        :=&\lim_{\eps\to 0^+}\frac{1}{\g\eps}\E\left[\ln M_{t+\eps}-\ln M_t\right]\\
        =&\lim_{\eps\to 0^+}\frac{1}{\g\eps}\E\left[\int_t^{t+\eps}\left(-\frac12\frac{\norm{H_s}^2}{M^2_{s^-}}+\int_{\R^d_\ast}\left[\ln\left(1+\frac{K_s(x)}{M_{s^-}}\right)-\frac{K_s(x)}{M_{s^-}}\right]\d \nu(x)\right)\d s
        \right]\\
        =&\frac1\g\E\left[-\frac12\frac{\norm{H_t}^2}{M^2_{t^-}}+\int_{\R^d_\ast}\left[\ln\left(1+\frac{K_t(x)}{M_{t^-}}\right)-\frac{K_t(x)}{M_{t^-}}\right]\d \nu(x)\right].
    \label{EQ:rate_entropic_1}
    \end{align}
    In the second equality above, the stochastic integrals vanished because they have zero expectation.
    We postpone to Appendix~\ref{APP:example} the rigorous justification for applying the cited proposition.
    The computation of the resilience rate at stopping times, by means of Proposition~\ref{PROP:integral_average}\itemref{IT:PROP:integral_average_2}, requires further technical conditions that will be discussed later in Section~\ref{subsec:VT}.
    However, if we further assume that $H,K$ are càglàd (continuous on the left with finite right limits), then the hypotheses of Proposition~\ref{PROP:integral_average}\itemref{IT:PROP:integral_average_2} are met, and, with a similar argument as before, we infer that, for any $\tau\in\mathcal T_T$,
    \[
        \bm\dot e^\g_\t(X)=\frac1\g\E\left[\left.-\frac12\frac{\norm{H_{\t^+}}^2}{M^2_\t}+\int_{\R^d_\ast}\left[\ln\left(1+\frac{K_{\t^+}(x)}{M_\t}\right)-\frac{K_{\t^+}(x)}{M_\t}\right]\d \nu (x)\right|\t<T\right],
    \]
    were we implicitly used the trivial fact that $\lim_{t\to \t^+}M_{t^-}=M_\t$.
\end{example}

\subsection{Bouncing drift}
\label{SEC:bouncing_drift}
Let us now assume that the dynamic risk measure is induced by a BSDE. 
Then, Proposition~\ref{PROP:integral_average} can be used to formally justify the definition of our measure of financial resilience, as we establish in the following theorem. 
\begin{theorem}
\label{TH:resilience_jumps}
    Let $\big(\rho(X),Z,U\big)$ be a solution to the BSDE \eqref{EQ:BSDE} with parameters $(g,T,X)$. 
    \begin{enumerate}[label=(\roman*)]
        \item \label{IT:th:resilience_jumps:deterministic}
            If there exists $q\geq1$ such that
            \begin{equation}
            \label{EQ:th:resilience_ass}
                \E\left[\int_0^T\big|g(\rho_t(X),Z_t,U_t)\big|^q\d t\right]<+\infty,  
            \end{equation}
            then, for $\ell_1$-a.e.\ $t\in[0,T)$, and all $s\in[0,t]$, both $\bm\dot\rho_{t|s}(X)$ and $\bm\dot\rho_t(X)$ exist, are finite, and satisfy:
            \begin{align}
                \bm\dot\rho_{t|s}(X)&=-\E\big[g\big(t,\rho_t(X),Z_t,U_t\big)\big|\F_s\big],\qquad\P\text{-a.s.}, \\
            \label{eq:bd}
                \bm\dot\rho_t(X)&=-\mathbb{E}\left[g\big(t,\rho_t(X),Z_t,U_t\big)\right]. 
            \end{align}
        \item \label{IT:th:resilience_jumps:stopping_times}
            Assume the following:
            \begin{itemize}[noitemsep, topsep=0pt, leftmargin=1.3em]
                \item $\P$-a.s., for any $t\in[0,T)$, the right limit 
                $\lim_{s\to t^+}g\big(s,\rho_s(X),Z_s,U_s\big)$ exists. 
                \item $g\big(\,\cdot\,,\rho(X),Z,U\big)$ has $\P$-a.s. at most countably many discontinuities.
                \item There exists $q\geq 1$ such that $\E\Big[\sup_{t\in[0,T]}\big|g(t,\rho_t(X),Z_t,U_t)\big|^q\Big]<+\infty$.
            \end{itemize}
            Then, for any $\tau,\s\in\mathcal{T}_T$ such that $\s\leq \t$, we have:
            \begin{align}
                \bm\dot\rho_{\tau|\s}(X)&=-\frac{\E\left[\ds\1_{\{\t<T\}}\left.\lim_{s\to \t^+}g\big(s,\rho_{s}(X),Z_{s},U_{s}\big)\right|\F_\s\right]}{\P(\t<T)},\qquad\P\text{-a.s.},\\
                \bm\dot\rho_{\tau}(X)&=-\E\left[\lim_{s\to \t^+}g\big(s,\rho_{s}(X),Z_{s},U_{s}\big){\Big|\tau<T}\right].
            \end{align}
    \end{enumerate}
\end{theorem}
\begin{proof}[Proof of \itemref{IT:th:resilience_jumps:deterministic}]
    We prove the first part of item~\itemref{IT:th:resilience_jumps:deterministic} by means of statement~\itemref{IT:PROP:integral_average_1} of Proposition~\ref{PROP:integral_average}. 
    For any $t\in[0,T)$ and $\eps>0$, we have, by direct inspection from equation~\eqref{EQ:BSDE}:
    \[
        \rho_{t+\eps}-\rho_t=-\int_t^{t+\varepsilon}g(r,\rho_r,Z_r,U_r)\d r+\int_t^{t+\varepsilon}Z_r\cdot \d W_r+\int_{(t,t+\eps]\times \R^d_\ast }U_r{(x)}\d\tilde N(r,x),\quad\P\text{-a.s.},
    \]
    where we suppress the dependence of $\rho$ on $X$ for simplicity.
    For $s\in[0,t]$, we take the conditional expectation with respect to $\F_s$, and divide by $\eps$:
    \begin{equation}
    \label{EQ:proof:th:resilience_jumps}
        \frac1\eps\E\left[\left.\rho_{t+\eps}-\rho_t\right|\mathcal{F}_s\right]
        =-\frac{1}{\eps}\mathbb{E}\left[\left.\int_t^{t+\eps}g(r,\rho_r,Z_r,U_r)\d r\right|\mathcal{F}_s\right],
    \end{equation}
    where both the It\^o integral and the integral with respect to the Poisson random measure have disappeared due to the martingale property.
    Recall indeed that $Z\in\mathcal H^2_T(\R^m)$ and $U\in\mathcal H^2_T(\Lambda^2)$ from Definition~\ref{DEF:solution_BSDE_jumps}, hence both their stochastic integral processes are martingales, see \cite[Theorem~1.33, Chapter~II]{Jacod+Shiryaev_2003_Limit_theorems_stochastic_processes}.
    
    Since $(Z,U)$ is $\mathcal P$-measurable and $\rho$ is càdlàg and $\bm{\mathcal F}$-adapted, $g(\rho,Z,U)$ is $\F\otimes\mathscr B\big([0,T]\big)$-measurable.
    This, together with condition \eqref{EQ:th:resilience_ass}, yields $g(\,\cdot\,,\rho,Z,U)\in L^q_T$.
    Thus, by the first thesis of Proposition~\ref{PROP:integral_average}, we get, for $\ell_1\text{-a.e. } t\in[0,T)$ and any $s\in[0,t]$,
    \[
        \frac1\eps\E\left[\left.\int_t^{t+\eps}g(r,\rho_r,Z_r,U_r)\d r\right|\F_s\right]\xrightarrow[]{\eps\to 0^+}\E[g(t,\rho_t,Z_t,U_t)|\F_s], \quad \textit{in }L^q(\F).
    \]
    Eventually, the first part of \itemref{IT:th:resilience_jumps:deterministic} follows from Equation~\eqref{EQ:proof:th:resilience_jumps} and from Definition~\ref{DEF:resilience_rate_jumps} of resilience rate, where $\tau=t$ identically. 

    The statement $\bm\dot\rho_t=-\E\left[g(t,\rho_t,Z_t,U_t)\right]$ follows from Remark~\ref{REM:relation_conditional_resilience_jumps}.
\end{proof}
\begin{proof}[Proof of \itemref{IT:th:resilience_jumps:stopping_times}]
    Let us now fix $\tau\in\mathcal T_T$ and denote by $\psi$ the process $g(\,\cdot\,,\rho(X),Z,U)$.
    We have
    \begin{align}
        \frac1\eps\big[\rho_{(\tau+\eps)\wedge T}-\rho_\t\big]
        &= -\frac1\eps\int_\t^{(\t+\eps)\wedge T}\!\!\!\!
\!\!\!\!\!\!\psi_s\d s + \frac 1\eps\int_\t^{(\tau+\eps)\wedge T}\!\!\!\!\!\!\!\!Z_s\cdot\d W_s + \frac 1\eps\int_{(\tau,(\t+\eps)\wedge T]\times\R^d_\ast}\!\!\!\!\!\!\!\!U_s(x)\d \tilde N(s,x),
    \end{align}
    where we again implied the dependence of $\rho$ in $X$.
    If we take the expectation on both sides of the above expression, we get, by Doob's optional sampling theorem (see, for instance \cite[Theorems~16,~18, Chapter~I]{Protter_2005_Stochastic_integration_differential_equations}): 
    \begin{align}
        &\ \frac1\eps\E\big[\rho_{(\tau+\eps)\wedge T}-\rho_\t\big]
        =-\E\left[\frac1\eps\int_\t^{(\t+\eps)\wedge T}\psi_s\d s\right]
        =-\E\left[\frac1\eps\int_\t^{(\t+\eps)\wedge T}\psi_{s\wedge T}\d s\right].
    \end{align}
    Here, after employing the martingale property in the first equality, we have noticed that $\psi_s=\psi_{s\wedge T}$ for all $s\in[\tau,(\tau+\eps)\wedge T]$.
    We now proceed by splitting the interval of integration depending on the values of the stopping time $\tau$:
    \begin{align}
        \ &-\E\left[\1_{[0,T-\eps]}(\t)\frac1\eps\int_\t^{\t+\eps}\psi_{s\wedge T}\d s + \1_{(T-\eps,T)}(\t)\frac1\eps\int_\t^{T}\psi_{s\wedge T}\d s\right]\\
        =\ &-\E\left[\Big(\1_{[0,T-\eps]}(\t)+\1_{(T-\eps,T)}(\t)\Big)\frac1\eps\int_\t^{\t+\eps}\psi_{s\wedge T}\d s - \1_{(T-\eps,T)}(\t)\frac1\eps\int_T^{\t+\eps}\psi_T\d s\right]\\
        =\ &-\E\left[\1_{[0,T)}(\t)\frac1\eps\int_\t^{\t+\eps}\psi_{s\wedge T}\d s\right] + \E\left[\1_{(T-\eps,T)}(\t)\frac{\t+\eps-T}{\eps}\psi_T\right].
    \end{align}
    For the first equality we used the additivity of the second integral and noticed that, if $\tau\in(T-\eps,T)$, then $\psi_{s\wedge T}=\psi_T$ for $s\in[T,\tau+\eps]$.
    In the last passage, we employed the linearity of expectation, reassembled the first two indicator functions and computed the second integral.
    
    In view of the hypotheses on the process $\psi=g(\,\cdot\,,\r,Z,U)$, we can apply Proposition~\ref{PROP:integral_average}\itemref{IT:PROP:integral_average_2} to the integral in the last line of the above chain of equalities, and conclude that the first expectation converges to $-\E\left[\1_{[0,T)}(\t)\lim_{s\to \t^+}\psi_{s}\right]$.
    
    As far as the second expectation is concerned, we observe that $\tau\in(T-\eps,T)$ implies $\tau+\eps-T< \eps$, hence 
    \[
        \left|\1_{(T-\eps,T)}(\t)\frac{\t+\eps-T}{\eps}\psi_T\right|
        \leq \1_{(T-\eps,T)}(\t)|\psi_T|\xrightarrow{\eps\to 0^+}0, \qquad \P\text{-a.s.}
    \]
    The boundedness of the indicator function and the hypotheses on $\psi=g(\,\cdot\,,\rho,Z,U)$ allow us to apply the dominated convergence theorem to infer that the second expectation converges to $0$, as $\eps\to 0^+$.
    Therefore, we conclude that
    \[
        \frac1\eps\E[\rho_{(\tau+\eps)\wedge T}-\rho_\tau]\xrightarrow[]{\eps\to 0^+} -\E\left[\1_{[0,T)}(\t)\lim_{s\to \t^+}\psi_{s}\right],
    \]
    and the thesis follows after dividing by $\P(\tau<T)$.

    The statement concerning the conditional resilience rate is proved analogously.
\end{proof}

We highlight that the hypothesis on the number of discontinuities of the process $\psi$ in Proposition~\ref{PROP:integral_average}\itemref{IT:PROP:integral_average_2} and of $g\big(\,\cdot\,,\rho(X),Z,U\big)$ in Theorem~\ref{TH:resilience_jumps}\itemref{IT:th:resilience_jumps:stopping_times} is sufficient to assure that the supremum in time is an $\F$-measurable random variable.

\begin{example}[continued from Example~\ref{EX:entropic_1}]
\label{EX:ecntropy_cont_2}
    Let us now recall that the dynamic entropic risk measure, in our filtration $\bm\F$, is actually induced by the following BSDE:
    \begin{equation}
    \label{EQ:entropy_BSDE}
        e^\g_t(X)=X+\int_t^T g(Z_s,U_s)\d s - \int_t^T Z_s \cdot\d W_s-\int_{(t,T]\times\R^d_\ast}U_s(x)\d\tilde N(s,x),
    \end{equation}
    $\forall\,t\in[0,T]$, with driver
    \begin{equation}
    \label{EQ:EX:entropic:driver}
        g:\R^m\times\Lambda^2\ni(z,u)\mapsto\frac\g2 \norm z^2 + \frac1\g\int_{\R^d_\ast}\left(e^{\g u}-\g u-1\right)\d \nu;
    \end{equation}
    see \cite[Theorem~4.1]{Becherer_2006_Bounded_solutions_backward_SDEs_jumps_utility_optimization_indifference_hedging}, \cite[Proposition~3.2]{ElKaroui+Matoussi+Ngoupeyou_2016_Quadratic_Exponential_Semimartingales_Application_BSDEs_jumps} or \cite[Example~4.1]{Mabitsela+Guambe+Kufakunesu_2022_Note_representation_BSDE-based_dynamic_risk_measures_dynamic_capital_allocations}.
    It satisfies condition~\ref{IT:Q_condition}, thus admitting a unique solution for any bounded terminal condition.
    Choose $X\in L^\infty(\F_T)$ and denote by $(e^\g(X),Z,U)$ the unique solution to the BSDE with parameters $(g,T,X)$, where $e^\g(X)$ has the explicit expression in equation~\eqref{EQ:dyn_entr_risk_meas}.
    We are going to show that the process $g(Z,U)$ is in $L^1_T$, which allows us to apply Theorem~\ref{TH:resilience_jumps}\itemref{IT:th:resilience_jumps:deterministic} and conclude that, for $\ell_1$-a.e. $t\in[0,T)$:
    \begin{equation}
    \label{EQ:rate_entropic_2}
        \bm\dot e^\g_t(X)=-\E\left[g(Z_t,U_t)\right]=-\E\left[\frac\g2\norm{Z_t}^2+\frac1\g\int_{\R^d_\ast}\left(e^{\g U_t(x)}-\g U_t(x)-1\right)\d \nu (x)\right].
    \end{equation}

    One easily verifies that $\norm{Z}^2\in L^1_T$ because $Z\in BMO(\R^m)\subset \mathcal H^2_T(\R^m)$, thus, we will only show that $\int_{\R^d_\ast}h(\g U)\d \nu\in L^1_T$, where $h(u):=e^u-u-1$ for $u\in\R$.
    To begin with, it is easily proved that $|h(u)|\leq e^{|u|}u^2/2$ for all $u\in\R$.
    Therefore, we obtain
    \begin{align}
        \E\left[\int_0^T\left |\int_{\R^d_\ast}h\big(\g U_t(x)\big)\d \nu(x)\right|\d t \right]
        \leq \frac{\g^2}{2}e^{\g C}\E\left[\int_0^T\int_{\R^d_\ast}U^2_t(x)\d\nu(x)\d t\right],
    \end{align}
    where $C>0$ is the constant that $\P\otimes\ell_1\otimes\nu$-essentially bounds $U$.
    The last expectation is finite because $U\in BMO(\Lambda^2)\subset \mathcal H^2_T(\Lambda^2)$.

    Let us now observe that formulas~\eqref{EQ:rate_entropic_1} and \eqref{EQ:rate_entropic_2} indeed coincide, providing two alternative representations of the same quantity: the first in terms of the coefficients $H$ and $K$ from the martingale representation theorem applied to the process $M=\exp(\g e^\g)$, and the second in terms of the solutions $Z$ and $U$ of the BSDE that induces $e^\g$.
    To show this, it is sufficient to compare the BSDE~\eqref{EQ:entropy_BSDE} with the explicit forward definition of $e^\g=(1/\g)\ln M$ from equation~\eqref{EQ:entropic}. 
    It is then clear that $H/M_-=\g Z$ and $\ln(1+K/M_-)=\g U$.
    The application of Theorem~\ref{TH:resilience_jumps}\itemref{IT:th:resilience_jumps:stopping_times} to the dynamic entropic risk measure requires further structure that will be considered later in Section~\ref{subsec:VT}.

    We emphasize that equation~\eqref{EQ:rate_entropic_2} reduces to a much simpler version when the filtration is generated only by the Brownian motion. 
    Indeed, in this case the dependence of the driver on the $u$-component vanishes and the resilience rate is given by
    \[
        \bm\dot e^\g_t(X)=-\frac\g2\E\left[\norm{Z_t}^2\right], \qquad \forall\, X\in L^\infty\big(\F^W_T\big), \ \ell_1\text{-a.e. }t\in[0,T).
    \]
\end{example}

In the case of a Brownian setting, the previous Theorem~\ref{TH:resilience_jumps} has the following stronger version.
\begin{corollary}
\label{COR:resil_brownian}
    Let $(\rho,Z)$ be a solution to the Brownian BSDE \eqref{EQ:brownian_BSDE} with parameters $(g,T,X)$. 
    \begin{enumerate}[label=(\roman*)]
        \item 
        \label{IT:cor:resil_brown:all_t}
            If there exists $q\geq1$ such that:
                \begin{equation}
                    \E\left[\int_0^T\big|g(t, \rho_t,Z_t)\big|^q\d t\right]<+\infty,   
                \end{equation}
            then for $\ell_1$-a.e. $t\in[0,T)$, $\bm\dot\rho_t(X)$ is well-defined, finite, and satisfies:
            \begin{equation}
            \label{EQ:brown_bounc_drift}
                \bm\dot\rho_{t}(X)=-\E[g(t,\rho_t,Z_t)].
            \end{equation}
            In addition, if the process $g(\,\cdot\,,\rho,Z):\Om\times[0,T]\ni(\om,t)\mapsto g\big(\om,t,\rho_t(\om),Z_t(\om)\big)$ has $\P$-a.s.\ continuous trajectories, then equation~\eqref{EQ:brown_bounc_drift} holds for all $t\in[0,T)$.
        \item 
        \label{IT:cor:resil_brown:tau}
            Assume the following:
            \begin{itemize}[noitemsep, topsep=0pt, leftmargin=1.3em]
                \item $g\big(\,\cdot\,,\rho,Z\big)$ has $\P$-a.s. continuous trajectories.
                \item There exists $q\geq 1$ such that $\E\Big[\sup_{t\in[0,T]}\big|g(t,\rho_t,Z_t)\big|^q\Big]<+\infty$.
            \end{itemize}
            Then, for any $\tau,\s\in\mathcal T_T$ such that $\s\leq \t$, we have:
            \begin{align}
                \bm\dot\rho_{\t|\s}(X)&=-\frac{\E\big[\1_{\{\tau<T\}}g(\t,\rho_\t,Z_\t)\big|\F_\s\big]}{\P(\tau<T)}, \qquad \P\text{-a.s.,}\\
                \label{EQ:bounc_det}
                \bm\dot\rho_{\t}(X)&=-\E[g(\t,\rho_\t,Z_\t)|\tau<T]. 
            \end{align}
    \end{enumerate}
\end{corollary}

In the sequel, we will refer to the right-hand side of \eqref{eq:bd}, \eqref{EQ:brown_bounc_drift} as the \textbf{\emph{bouncing drift}}.

\begin{remark}
\label{REM:COR_brown_res_rate}    
    Corollary~\ref{COR:resil_brownian} is a straightforward application of Theorem~\ref{TH:resilience_jumps} to the special case of a Brownian filtration. 
    Indeed, the further assumption on the $\P$-a.s.\ pathwise continuity for $g(\,\cdot\,,\r,Z)$ yields several simplifications of the statement. 
    First, it ensures that we can apply the second part of Proposition~\ref{PROP:integral_average}\itemref{IT:PROP:integral_average_1} to obtain the validity of equation~\eqref{EQ:brown_bounc_drift} for all times.
    Second, it allows us to apply Theorem~\ref{TH:resilience_jumps}\itemref{IT:th:resilience_jumps:stopping_times}.
    Finally, it simplifies the expressions for the resilience rates from Theorem~\ref{TH:resilience_jumps}\itemref{IT:th:resilience_jumps:stopping_times}.

    The respective hypotheses can indeed be verified in the Brownian setting (see Proposition~\ref{PROP:verification_stopptimes_brown}); however, we are not aware of any non-trivial sufficient conditions that ensure continuity  of the solutions to Brownian-Poissonian BSDEs.
    Therefore, even though the last corollary could, in principle, be correctly stated in the Brownian-Poissonian filtration, it would be of limited practical use, as the required assumptions are too restrictive to be satisfied in non-trivial cases.
\end{remark}

We emphasize here that the deterministic setting with $t \in [0,T)$ and the setting with stopping times $\tau \in \mathcal{T}_T$ allow for a similar interpretation. 
That is, instead of employing stopping times as in Definition~\ref{DEF:resilience_rate_jumps}, one could alternatively condition upon a stress scenario where the risk-acceptance set is breached at a fixed deterministic time $t$. 
In this case, the claim under management, $X$, is evaluated conditionally upon, for example, $Y$ exceeding a certain threshold at time $t \in [0,T)$, and the rate at which the resulting risk measure recovers, i.e., bounces back is analyzed. 
This is a typical approach in stress testing using risk models.
As is evident from the previous theorem, upon comparing items \itemref{IT:th:resilience_jumps:deterministic} and \itemref{IT:th:resilience_jumps:stopping_times}, the approach with deterministic times requires fewer conditions on the pathwise regularity of the trajectories, but it is mathematically and practically less elegant and appealing.

\begin{remark}
    \label{REM:greek1}
    In this remark, we stress the difference between our definition of the resilience rate and the Greek $\Theta$ as usually defined in the option pricing literature. 
    It is well known that, when considering a price $(S_t)_{t\in[0,T]}$ of a given underlying financial asset, the solution $P_t$ at time $t\in[0,T]$ to a BSDE with terminal condition $X=f(S_T)$ and a carefully chosen driver $g$ can be interpreted as the fair price (or the hedging cost) at time $t$ of a European-style derivative written on the underlying $S$ with payoff $X$ (see also Section~\ref{sec:examples}). 
    In general, if $p:\R_{+}\times[0,T]\to\R_+$ is the pricing functional for a financial instrument (namely, $p(S_t,t)=P_t$ for all $t\in[0,T]$), the Greek $\Theta$ is defined as $\Theta(S,t):=\partial_{t}p(S,t)$. 
    At first glance, both the Greek $\Theta$ and the resilience rate $\bm\dot P$ in this setting consist of a time-derivative of the fair price of the financial instrument.
    However, we note that, for the Greek $\Theta$, this time-derivative is taken as a partial derivative of the price, assuming that the underlying $S$ is fixed. 
    By contrast, the resilience rate accounts for the (expected) \textit{total} derivative of the price, explicitly incorporating its time-dependence through $t$, rather than treating $S$ as a fixed quantity.
    In this context, it is important to note that while the pricing functional is a deterministic differentiable function, the definition of the resilience rate (cf.\ Definition~\ref{DEF:resilience_rate_jumps}) involves taking the time-derivative of the expected price of the instrument. 
    This is necessary because the price process itself usually exhibits non-differentiable trajectories.
\end{remark}

\subsection{Verification of the assumptions}
\label{subsec:VT}

In this subsection, we provide natural, sufficient conditions under which the assumptions in Theorem~\ref{TH:resilience_jumps} and Corollary~\ref{COR:resil_brownian} are verified.
We first deal with the case of deterministic times.

\begin{proposition}
\label{PROP:verification_deterministic}
    Assume any of the following two conditions:
    \begin{itemize}[label=\scriptsize$\bullet$, noitemsep, topsep=0pt]
        \item The driver $g$ satisfies the condition~\ref{IT:L_condition} and $X\in L^2(\F_T)$.
        \item The driver $g$ satisfies the condition~\ref{IT:Q_condition} and $X\in L^\infty(\F_T)$.
    \end{itemize}
    Then the assumption of Theorem~\ref{TH:resilience_jumps}\itemref{IT:th:resilience_jumps:deterministic} is satisfied.\\
    Moreover, assume any of the following two conditions:
    \begin{itemize}[label=\scriptsize$\bullet$, noitemsep, topsep=0pt]
        \item The Brownian driver $g$ satisfies the condition~\ref{IT:BL_condition} and $X\in L^2\big(\F^W_T\big)$.
        \item The Brownian driver $g$ satisfies the condition~\ref{IT:BQ_condition} and $X\in L^\infty\big(\F^W_T\big)$.
    \end{itemize}
    Then the first assumption of Corollary~\ref{COR:resil_brownian}\itemref{IT:cor:resil_brown:all_t} is satisfied.
\end{proposition}
\begin{proof}
    The first estimate for the driver $g$ in condition~\ref{IT:L_condition}, together with the minimal regularity imposed for a solution to the BSDE~\eqref{EQ:BSDE}, cf.\ Definition~\ref{DEF:solution_BSDE_jumps}, trivially implies $g(\,\cdot\,,\rho(X),Z,U)\in L^2_T$.
    Hence, the assumption of Theorem~\ref{TH:resilience_jumps}\itemref{IT:th:resilience_jumps:deterministic} is satisfied with $q=2$.
    Similarly, the Lipschitz assumption on the Brownian driver $g$ in condition~\ref{IT:BL_condition} yields
    \[
        \big|g(\,\cdot\,,\rho(X),Z)\big|\leq |g(\,\cdot\,,0,0)|+K\big(|\rho(X)|+\norm{Z}\big),\qquad\P\otimes\ell_1\text{-a.e.},
    \]
    which again implies $g(\,\cdot\,,\rho(X),Z)\in L^2_T$, thanks to the hypothesis $g(\,\cdot\,,0,0)\in L^2_T$.
    
    In case of quadratic growth for a Brownian $g$ and bounded terminal condition, the first estimate in condition~\ref{IT:BQ_condition} yields $g(\,\cdot\,,\r(X),Z)\in L^1_T$.
    Similarly, if $g$ satisfies the condition~\ref{IT:Q_condition} and $X\in L^\infty(\F_T)$, then $(Z,U)\in \mathcal H^2_T(\R^m)\times\mathcal H^2_T(\Lambda^2)$ and $U$ is ${\P\otimes\ell_1\otimes\nu\text{-essentially bounded}}$.
    Thus, reasoning as in Example~\ref{EX:ecntropy_cont_2}, we get $g(\,\cdot\,,Z,U)\in L^1_T$.
\end{proof}

\subsubsection{Stopping times in Brownian-Poissonian filtration}
\label{SEC:verif_assu_brown-poisson}
We now analyze when the additional conditions in item~\itemref{IT:th:resilience_jumps:stopping_times} of Theorem~\ref{TH:resilience_jumps} are verified.
Let $n\in \N$ and assume the following:
\begin{gather}
    \mu:\R^n\to \R^n, \qquad \s:\R^n\to \R^{n\times m}, \qquad \g:\R^n\times \R^d_\ast \to \R^{n}, \qquad h:\R^n\to\R,\\
    f:\R^n\times\R\times\R^m\times\Lambda^2\to \R.
\end{gather}
We consider the following system of forward-backward stochastic differential equations (FBSDE) with coefficients $\mu,\s,\g,h,f$, initial point $x\in\R^n$ and horizon $T>0$:
\begin{equation}\label{EQ:FBSDE_jumps}
    \begin{cases}
        \ds X_t=x+\int_0^t\mu(X_{s^-})\d s +\int_0^t\s(X_{s^-}) \d W_s+\int_{[0,t]\times \R^d_\ast }\gamma(X_{s^-},\xi)\d\tilde N(s,\xi),\\
        \ds Y_t=h(X_T)+\int_t^Tf(X_s,Y_s,Z_s, U_s)\d s -\int_t^TZ_s\cdot \d W_s-\int_{(t,T]\times \R^d_\ast } U_s{(\xi)}\d\tilde N(s,\xi).
    \end{cases}
\end{equation}
The second equation of the above system is a special case of the BSDE \eqref{EQ:BSDE}, where the driver $g$ takes the form
\[
    g(\om,t,y,z,u)=f(X_t(\om),y,z,u), \qquad \forall\, (\om,t,y,z,u)\in\Om\times[0,T]\times\R\times\R^m\times\Lambda^2,
\]
and the terminal condition is $h(X_T)$.
Further assume the following regularities for the coefficients:
\begin{enumerate}[label=$(A\arabic*)$, noitemsep, topsep=0pt]
    \item $\mu\in C^3(\R^n;\R^n)$ with bounded partial derivatives of order $1,2,3$.
    \item $\s\in C^3(\R^n;\R^{n\times m})$ with bounded partial derivatives of order $1,2,3$.
    \item $\g:\R^n\times \R^d_\ast \to\R^n$ measurable \textit{s.t.}:
    \begin{itemize}[label=\scriptsize$\bullet$, noitemsep, topsep=0pt, leftmargin=1em]
        \item For $\xi\in \R^d_\ast $, the map $\R^n\ni x \mapsto \g(x ,\xi)\in\R^n$ has continuous and bounded partial derivatives of order $1,2,3$.
        \item There exists $K_1>0$ such that $\norm{\g(0,\xi)}\leq K_1(1\wedge \norm{\xi})$ for all $\xi\in \R^d_\ast $.
        \item There exists $K_2>0$ such that for $(x,\xi)\in\R^n\times \R^d_\ast$:
        \[
            \left|\!\left|\frac{\partial^{|\a|}}{\partial x^\a}\g(x,\xi)\right|\!\right|\leq K_2(1\wedge \norm{\xi}),\qquad\forall\,  \a\in \N_0^n \text{ s.t. } 1\leq \sum_{i=1}^n\a_i\leq 3.
        \]
    \end{itemize}
    \item $h\in C^3(\R^n)$ with partial derivatives with polynomial growth.
    \item 
    $f\in C^3(\R^n\times \R\times \R^m\times \Lambda^2)$ with bounded partial (Fréchet) derivatives of order $1,2,3$.
\end{enumerate}

Then, see \cite[Section~3]{Buckdahn+Pardoux_1994_BSDEs_jumps_associated_integro-partial_differential_equations} and \cite[Section~5.2]{Matoussi+Sabbagh+Zhou_2015_Obstacle_problem_semilinear_parabolic_partial_integro-differential_equations}, we have the following result. 
For any $x\in \R^n$, there exists a unique solution $(X,Y,Z,U)$ to the FBSDE~\eqref{EQ:FBSDE_jumps} with coefficients $\mu$, $\s$, $\g$, $h$, $f$, starting point $x$ and horizon $T$ s.t.\ ${(X,Y,Z,U)\in \mathcal S^p_T\times \mathcal S^p_T\times \mathcal H^p_T(\R^m)\times \mathcal H^p_T(\Lambda^2)}$ for any $p\geq 1$.
Moreover, there exists a left-continuous version of $(Z,U):\Om\times[0,T]\to\R^m\times\Lambda^2$ and a function $u\in C^{1,2}\big([0,T]\times\R^n\big)$ with partial derivatives of polynomial growth such that the following representation holds $\P$-a.s., for $\ell_1\otimes\nu$-a.e.\ $(s,\xi)$:
\begin{equation}\label{EQ:representation}
    Y_s=u(s,X_s),\quad Z_s=\s^{\ss{\top}}(X_{s^-})\nabla u(s,X_{s^-}),\quad U_s{(\xi)}=u\big(s,X_{s^-}+\g(X_{s^-},\xi)\big)-u(s,X_{s^-}).
\end{equation}
This representation for the solution allows us to give sufficient conditions for the verification of the assumptions of item~\itemref{IT:th:resilience_jumps:stopping_times} in Theorem~\ref{TH:resilience_jumps}.
\begin{proposition}
\label{PROP:verification_stopptimes}
    Assume $(A1),\dots,(A5)$ above, and fix $x\in\R^n$.
    Let $(X,Y,Z,U)$ be the solution to the FBSDE~\eqref{EQ:FBSDE_jumps} with coefficients $\mu,\s,\g,h,f$, starting point $x$ and horizon $T$.
    Then the following properties hold:
    \begin{itemize}[noitemsep, topsep=0pt]
        \item $\P$-a.s., for any $t\in[0,T)$, the right limit 
        $\lim_{s\to t^+}f\big(X_s,Y_s,Z_s,U_s\big)$ exists. 
        \item $f\big(X,Y,Z,U\big)$ has $\P$-a.s.\ at most countably many discontinuities.
        \item There exists $q\geq 1$ such that $\E\Big[\sup_{t\in[0,T]}\big|f(X_t,Y_t,Z_t,U_t)\big|^q\Big]<+\infty$.
    \end{itemize}
    In particular, for any $\tau,\s\in\mathcal{T}_T$ such that $\s\leq \t$, we have:
    \begin{align}
        \bm\dot Y_{\tau|\s}\big(h(X_T)\big)&=-\frac{\E\left[\ds\1_{\{\t<T\}}\left.f\big(X_\t,Y_\t,Z_{\t^+},U_{\t^+}\big)\right|\F_\s\right]}{\P(\t<T)},\qquad\P\text{-a.s.},\\
        \bm\dot Y_{\tau}\big(h(X_T)\big)&=-\E\left[f\big(X_\t,Y_\t,Z_{\t^+},U_{\t^+}\big){\Big|\tau<T}\right].
    \end{align}
\end{proposition}
\begin{proof}
    First, let us show the existence of the limit. 
    Let $\overline\Om$ be the set of full probability where the trajectories of $X,Y$ are right-continuous and the representation in equation~\eqref{EQ:representation} holds true.
    Let us now fix $\om\in\overline\Om$, implied in the notation for the sake of clarity, and $t\in[0,T)$.
    The following limit exists by the representation \eqref{EQ:representation} and because $X(\om)$ is càdlàg:
    \[
        \lim_{s\to t^+}Z_s
        =\lim_{s\to t^+}\s^{\ss{\top}}(X_{s^-})\nabla u(s,X_{s^-})
        =\s^{\ss{\top}}(X_{t})\nabla u(t,X_{t}),
    \]
    and similarly for $\lim_{s\to t^+}U_s$ in $\Lambda^2$.
    This shows in particular that the trajectories of $(Z,U)$ are $\P$-a.s.\ càglàd.
    Hence, by continuity of $f$ and right-continuity of $X(\om),Y(\om)$:
    \begin{align*}
        \lim_{s\to t^+}f(X_s,Y_s,Z_s,U_s)
        &=f\left(X_t,Y_t,\lim_{s\to t^+}Z_s,\lim_{s\to t^+}U_s\right),
    \end{align*}
    the limit for $U$ being interpreted in $\Lambda^2$.

    We now discuss the number of discontinuities for the trajectories of $f(X,Y,Z,U)$.
    We have proved above that $(Z,U)$ are $\P$-a.s.\ càglàd, and $(X,Y)$ are $\P$-a.s.\ càdlàg. 
    Hence, $(X,Y,Z,U)$ admits $\P$-a.s.\ finite right and left limits at any time, and the same holds for the composition with the continuous function $f$.  
    Therefore, for $\P$-a.e.\ $\om \in \Om$, the $\om$-trajectory of $f(X,Y,Z,U)$ on $[0,T]$ is a regulated function, which are known to have at most countably many discontinuities (see, e.g., \cite[Theorem~4.7]{Banas+Kot_2017_regulated_functions}).

    Concerning the estimate, the boundedness of the first-order partial derivatives of $f$ implies the uniform Lipschitz condition in all its variables, thus, for a certain $L_f\geq 0$ we have
    \begin{align}\label{EQ:verification_f}
        \E\left[\sup_{t\in[0,T]}|f(X_t,Y_t,Z_t,U_t)|\right]
        \leq &\ \, |f(0,0,0,0)|\nonumber\\
        &+L_f\, \E\left[\sup_{t\in[0,T]}\big(\norm{X_t}+|Y_t|+\norm{Z_t}+\|U_t\|_{\Lambda^2}\big)\right].
    \end{align}
    If we once again use the representation \eqref{EQ:representation}, along with assumption $(A2)$ and the power-growth regularity on the partial derivatives of  $u$, we get, for constants $L_\s\geq 0$, $C_{1,2}\geq 0$ and a power $\a\geq 0$:
    \begin{align}
        \norm{Z_t}
        = \norm{\s^{\ss{\top}}(X_{t^-})\nabla u(t,X_{t^-})}
        &\leq C_1\big(\norm{\s^{\ss{\top}}(0)}+L_\s\norm{X_{t^-}}\big)\big(1+|t|^\a+\norm{X_{t^-}}^\a\big)\\
        &\leq C_2\big(1+|t|^\a+\norm{X_{t^-}}^{\a+1}\big).
    \label{EQ:verification_Z}
    \end{align}
    Similarly, for a constant $C>0$ possibly different at every line, and for a power $\b\geq 0$:
    \begin{align}
        \|U_t\|^2_{\Lambda^2}
        &= \int_{\R^d_\ast }\big|\!\big|u\big(t,X_{t^-}+\g(X_{t^{-}},\xi)\big)-u(t,X_{t^-})\big|\!\big|^2\d\nu(\xi)\\
        \notag &\leq C\int_{\R^d_\ast }\norm{\g(X_{t^-},\xi)}^2\sup_{\norm{x}\leq R_t(\xi)}\norm{\nabla u(t,x)}^2\d\nu(\xi)\\
        &\leq C\int_{\R^d_\ast }\norm{\g(X_{t^-},\xi)}^2\big(1+|t|^\b+\norm{X_{t^-}}^\b+\norm{\g(X_t^{-},\xi)}^\b\big)\d\nu(\xi),\label{EQ:verification_U}
    \end{align}
    where the supremum is among $x\in \R^n$ such that $\norm{x}\leq R_t(\xi):=\norm{X_{t^-}}\vee \norm{X_{t^-}+\g(X_{t^{-}},\xi)}$.
    The term with $\g$ can be estimated as follows by means of the assumption $(A3)$:
    \begin{align*}
        \norm{\g(X_{t^-},\xi)}
        \leq \norm{\g(0,\xi)}+\norm{X_{t^-}}\sup_{\norm{x}\leq \norm{X_{t^-}}}\norm{D_x \g(x,\xi)}
        \leq K(1\wedge \norm{\xi})(1+\norm{X_{t^-}}),
    \end{align*}
    where the supremum is among $x\in \R^n$ s.t. $\norm x\leq\norm{X_{t^-}}$, and $D_x\g(x,\xi)$ is the Jacobian of $\g$ with respect to its first variable.
    By plugging this estimate into equation~\eqref{EQ:verification_U}, we get, for a different $C>0$,
    \begin{align}\label{EQ:verification_U_2}
        \|U_t\|^2_{\Lambda^2}
        &\leq C\big(1+\norm{X_{t^-}}^{2+\b}\big)\big(1+|t|^\b\big)\int_{\R^d_\ast }(1\wedge \norm{\xi}^2)\d\nu(\xi),
    \end{align}
    where the integral is finite thanks to the initial assumptions on the measure $\nu$.
    Inserting both \eqref{EQ:verification_Z} and \eqref{EQ:verification_U_2} back into \eqref{EQ:verification_f}, we find a constant $C_T>0$ such that 
    \begin{align*}
        \E\left[\sup_{t\in[0,T]}|f(X_t,Y_t,Z_t,U_t)|\right]
        \leq \, C_T \E\left[\sup_{t\in[0,T]}\big(\norm{X_t}^q+|Y_t|\big)\right],
    \end{align*}
    where $q:=1+\max\{\a,\b/2\}$.
    Since $X\in \mathcal S^p_T$ for all $p\geq 1$, the right-hand side of the last equation is finite.
\end{proof}

Of course, Proposition~\ref{PROP:verification_stopptimes} can also be applied to the special case of Brownian FBSDEs.
However, the assumptions of Corollary~\ref{COR:resil_brownian}\itemref{IT:cor:resil_brown:tau} are verified under slightly more general conditions, which are discussed in detail in Appendix~\ref{APP:verif_assu_brownian}.

\section{Properties of the Resilience Rate}
\label{sec:properties}

In Section~\ref{SEC:funct_spaces_dyn_risk_measures}, we recalled the main properties of dynamic risk measures. 
Considerable attention has been devoted in the literature to characterizing these properties for BSDE-induced risk measures in terms of the underlying driver (see Section~\ref{SEC:BSDEs} for a brief overview). 
It is natural to address the same problem for our measure of resilience. 
Thus, in this section, we study properties of the resilience rate and identify conditions on the underlying BSDE and dynamic risk measure that are sufficient and/or necessary for these properties to hold.

\begin{example}[Continued from Examples \ref{EX:entropic_1}, \ref{EX:ecntropy_cont_2}]
\label{EX:entropy_cont_3}
    Let us consider once again the dynamic entropic risk measure of equation~\eqref{EQ:dyn_entr_risk_meas}. 
    Recall the definitions of time consistency and cash-additivity from Section~\ref{SEC:funct_spaces_dyn_risk_measures}, which are well known to be satisfied by the dynamic entropic risk measure, and recall the notation introduced in Example~\ref{EX:entropic_1}. 
    We show that the associated resilience rate satisfies properties similar to time consistency and cash-additivity.
    
    One may choose to evaluate the entropic risk measure at time $t$ for the random variable $e^\g_s(X)$, for some $s\in(t,T)$, instead of for the terminal random variable $X$.
    This is indeed feasible as $e^\g_s(X)$ is $\F_s$-measurable, hence $\F_T$-measurable, and, for any $p\geq \g$, we have 
    \[
        \E\left[e^{pe^\g_s(X)}\right]
        =\E\left[\left(\E[e^{\g X}|\F_s]\right)^{p/\g}\right]\leq \E\left[\E\left[\left.e^{pX}\right|\F_s\right]\right]
        =\E[e^{pX}]<+\infty,
    \]
    where we first used equation~\eqref{EQ:dyn_entr_risk_meas}, followed by the conditional Jensen inequality, then we invoked elementary  properties of the conditional expectation, and finally we recalled the hypotheses on $X$ from Example~\ref{EX:entropic_1}.
    In this case, we have 
    \[
        M_t
        =\E\left[\left.e^{\g e^\g_s(X)}\right|\F_t\right]
        =\E\left[\left.\E[e^{\g X}|\F_s]\right|\F_t\right]
        =\E\left[\left.e^{\g X}\right|\F_t\right],
    \]
    which shows that the martingale $M$ does not change with the proposed modification of the terminal random variable $X$ (which actually also derives from the time consistency property of the dynamic entropic risk measure).
    Therefore, the processes $H$ and $K$ representing $M$ in equation~\eqref{EQ:EX:entropic_1:martingale_repr_th} remain unchanged, and so does the resilience rate in equation~\eqref{EQ:rate_entropic_1}. 
    Hence, we have proved that the resilience rate inherits some kind of time consistency from the entropic risk measure:
    \[
        \bm\dot e^\g_t\left(e^\g_s(X)\right)=\bm\dot e^\g_t(X), \qquad \forall\, s\in(t,T).
    \]

    Similarly, if $h\in L^\infty(\F_t)$, we can translate the terminal random variable $X$ to $X+h$, as $X+h$ is $\F_T$-measurable and $\E[e^{p(X+h)}]$ is finite for all $p\geq \g$.
    With this modification, the martingale $M$ becomes
    \[
        \tilde M_t
        :=\E\left[\left.e^{\g(X+h)}\right|\F_t\right]
        =e^{\g h}M_t,
    \]
    where we used the definition of $M$ in equation~\eqref{EQ:EX:entropic_1:def_M}, and the $\F_t$-measurability of $h$.
    Consequently, from equation~\eqref{EQ:EX:entropic_1:martingale_repr_th}, we see that the processes $\tilde H$ and $\tilde K$ representing $\tilde M$ are $\tilde H:= e^{\g h} H$ and $\tilde K := e^{\g h} K$, thus the resilience rate from equation~\eqref{EQ:rate_entropic_1} remains unchanged, as it only depends on the fractions $\tilde H/\tilde M_-=H/M_-$ and $\tilde K/\tilde M_-=K/M_-$.
    To conclude, we have showed that the resilience rate enjoys a cash-invariance property of the following kind:
    \[
        \bm\dot e^\g_t(X+h)=\bm\dot e^\g_t(X), \qquad \forall\, h\in L^\infty(\F_t).
    \]
\end{example}

\begin{notation}
\label{NOT:res}
    If $g$ is a driver that induces a dynamic risk measure $\rho(g)$ on $L^p(\F_T)$, for some $p\in[1,+\infty]$ (see Remark~\ref{REM:induced_dyn_risk_meas}), and if $X\in L^p(\F_T)$, then $\big(\rho(g,X), Z(g,X), U(g,X)\big)$
    denotes the solution to the BSDE \eqref{EQ:BSDE} with parameters $(g,T,X)$.
    For $t\in[0,T)$ and $X\in L^p(\F_T)$, the resilience rate of $\rho_t(g,X)$ will be denoted as $\bm\dot\rho_t(g,X)$ and will naturally define two functions: $\bm\dot\rho(g,X):[0,T)\ni t\mapsto \bm\dot\rho_t(g,X)\in\R$ for fixed $X\in L^p(\F_T)$, and $\bm\dot\rho_t(g):L^p(\F_T)\ni X\mapsto \bm\dot\rho_t(g,X)\in\R$ for fixed time $t\in[0,T)$.
    The dependence on $g$ or $X$ will sometimes be suppressed if it is clear from the context.
\end{notation}

Although the results in this section are stated for deterministic times, most of them (except for Corollaries~\ref{COR:prop_brown_res_rate}, \ref{COR:conv_RAS_prop} and Section~\ref{SEC:continuity}) can be extended to stopping times, as soon as the conditions for the existence of the resilience rates are satisfied.
In more detail, whenever the existence of the resilience rate is assumed \textit{a priori}, the statements remain true by replacing deterministic times in $[0,T)$ with stopping times in $\mathcal T_T$. 
Instead, if the existence of the resilience rate follows from a BSDE --- thanks to Theorem~\ref{TH:resilience_jumps}\itemref{IT:th:resilience_jumps:deterministic} or Corollary~\ref{COR:resil_brownian}\itemref{IT:cor:resil_brown:all_t} --- then the solution of the BSDE needs to satisfy the stronger assumptions of Theorem~\ref{TH:resilience_jumps}\itemref{IT:th:resilience_jumps:stopping_times} or Corollary~\ref{COR:resil_brownian}\itemref{IT:cor:resil_brown:tau}, in order to allow for stopping times.
In case of Corollaries~\ref{COR:prop_brown_res_rate}, \ref{COR:conv_RAS_prop} the generalization to stopping times does not hold, as the proof relies on differentiability in time of the map $t\mapsto \E[\rho_t(X)]$.
Section~\ref{SEC:continuity} deals with time integration of $t\mapsto \bm \dot \rho_t(X)$, hence a statement for stopping times is not viable.

\subsection{Time consistency}
\label{SEC:time_cons}
The following proposition shows that our resilience rate satisfies a special form of time consistency (e.g., \cite{Bion-Nadal_2008_Dynamic_risk_measures_time_consistency__risk_measures_BMO_martingales,Delbaen+Peng+RosazzaGianin_2010_Representation_penalty_term_dynamic_concave_utilities,Follmer+Schied_2016_Stochastic_finance}), which we denote as $(\bm\dot\rho,\rho)$-time consistency.

\begin{proposition}
\label{PROP:time_consistency}
    Let $p\in[1,+\infty]$ and let $\rho$ be a dynamic risk measure on $L^p(\F_T)$.
    If $\rho$ is time consistent, then the following ${(\bm\dot\rho,\rho)\text{-time consistency}}$ is satisfied: If $X\in L^p(\F_T)$ and $t\in[0,T)$ are such that $\bm\dot\rho_t(X)$ exists, then 
    \begin{equation} 
    \label{EQ:time_consistency_jumps}
        \bm\dot\rho_{t}\big(\rho_s(X)\big)= \bm\dot\rho_{t}(X), \qquad \forall\, s\in(t,T).
    \end{equation}
\end{proposition}
\begin{proof}
    Fix $0\leq t< s<T$.
    The thesis follows from direct inspection of Definition~\ref{DEF:resilience_rate_jumps}, once observed that the time consistency of $\rho$ implies that
    \begin{align}
        \frac1\eps\E\left[\rho_{(t+\eps)\wedge T}\big(\rho_s(X)\big)- \rho_{t}\big(\rho_s(X)\big)\right]
        =\frac1\eps{\E\left[\rho_{(t+\eps)\wedge T}(X)- \rho_{t}(X)\right]}, \qquad \forall\, \eps>0.
    \end{align}
\end{proof}
The proposition above admits a clear interpretation: the resilience evaluation at time $t$ of a risky claim can equivalently be computed starting from its risk assessment at any later time $s>t$.

\begin{remark}
\label{REM:strong_time_consistency}
    We recall that if the risk measure $\r$ is induced by a BSDE, it enjoyes a property stronger than the usual time consistency of page~\pageref{PAGE:time_consistency}, namely the flow property (or strong time consistency) stated on page~\pageref{PAGE:flow_property} (see \cite{DiNunno+RosazzaGianin_2024_Fully_dynamic_risk_measures_horizon_risk_time-consistency_relations_BSDEs_BSVIEs}).
    Let us assume we are in this setting.
    Then one can easily obtain a somewhat stronger thesis for Proposition~\ref{PROP:time_consistency}.
    For each $s\in(0,T]$, let $[0,s]\times L^2(\F_s)\ni (t,X)\mapsto \rho_{t,s}(X)\in L^2(\F_t)$  denote the dynamic risk measure induced by a driver $g$ on $L^p(\F_s)$, for some $p\in[2,+\infty]$.
    Then, for $X\in L^p(\F_T)$ and for $\ell_1$-a.e. $t\in[0,s)$, we have that
    \[
        \bm\dot\rho_{t,s}\big(\rho_{s,T}(X)\big)=\bm\dot\rho_{t,T}(X), \qquad \forall\, s\in(t,T).
    \] 
\end{remark}

\subsection{Cash-insensitivity and positive homogeneity}
\label{SEC:properties_cash-ins_pos-hom}

As discussed in the Introduction, two fundamental properties of risk measures are cash-additivity and positive homogeneity. 
In this subsection, we show that they lead to interesting consequences concerning resilience rates.

\begin{proposition}
\label{PROP:res_rate_properties}
    Let $\rho$ be a dynamic risk measure on $L^p(\F_T)$ for some $p\in[1,+\infty]$.
    \begin{enumerate}[label=(\roman*)]
    \item   
    \label{IT:PROP:res_rate_properties:cash_add}
        If $\rho$ is cash-additive and if $\bm\dot\rho_t(X)$ exists for some $t\in[0,T)$ and $X\in L^p(\F_T)$, then $\bm\dot\rho_t(X)$ is cash-insensitive; that is, for any $h\in L^p(\F_t)$, $\bm\dot\rho_t(X+h)$ exists and 
        \[
            \bm\dot\rho_t(X+h)=  \bm\dot\rho_t(X).
        \]
    \item   
    \label{IT:PROP:res_rate_properties:pos_hom}
        If $\rho$ is positively homogeneous and if $\bm\dot\rho_t(X)$ exists for some $t\in[0,T)$ and $X\in L^p(\F_T)$, then $\bm\dot\rho_t(X)$ is positively homogeneous; that is, for any $\a\geq 0$, $\bm\dot\rho(\a X)$ exists and 
        \[
            \bm\dot\rho_t(\alpha X)= \alpha \bm\dot\rho_t(X).
        \]
    \end{enumerate}
\end{proposition}
\begin{proof}
    Let us show \itemref{IT:PROP:res_rate_properties:cash_add}.
    Fix $t\in[0,T)$ and $X\in L^p(\F_T)$ such that $\bm\dot\rho_t(X)$ exists.
    If $h\in L^p(\F_t)$, then the cash-additivity of $\rho$ yields $\rho_s(X+h)=\rho_s(X)$ for all $s\in[t,T]$. 
    Hence, for any $\eps\in(0,T-t]$, we have
    \begin{align}
        \frac1\eps[\rho_{t+\eps}(X+h)-\rho_t(X+h)]
        =\frac1\eps[\rho_{t+\eps}(X)-\rho_t(X)].
    \end{align}
    By definition, the right-hand side of the above equality converges to $\bm\dot\rho_t(X)$ as $\eps\to 0^+$. 
    In particular, we infer that $\bm\dot\rho_t(X+h)$ exists and it equals $\bm\dot\rho_t(X)$.
    Item \itemref{IT:PROP:res_rate_properties:pos_hom} is proved analogously.
\end{proof}

Let us highlight that the previous proposition does not rely on the fact that the dynamic risk measure is induced by a BSDE. 
However, if one were to restrict to the study of dynamic risk measures induced by Brownian BSDEs, a much stronger result can be proved. 
Indeed, it is known that cash-additivity and positive homogeneity are closely related to the structure of the BSDE driver, specifically, they are equivalent to the driver being independent of the $y$-component or being positively homogeneous in the spatial variables, respectively (see Section~\ref{SEC:BSDEs} for a formal statement).
We now show that, in the Brownian setting, analogous equivalences hold for the resilience rate.

\begin{corollary}
\label{COR:prop_brown_res_rate}
    Let $g$ satisfy condition~\ref{IT:BL_condition} and $p=2$, or let $g$ satisfy condition~\ref{IT:BQ_condition} and $p=\infty$.
    Further assume:
    \begin{itemize}[label=\scriptsize$\bullet$, noitemsep, topsep=0pt]
        \item $\P\otimes\ell_1$-a.e.\ we have $g(\,\cdot\,,y,0)=0$ for all $y\in\R$.
        \item For all $X\in L^p\big(\F^W_T\big)$, the process $g\big(\,\cdot\,,\rho(X),Z(X)\big)$ is $\P\text{-a.s.}$ continuous.
    \end{itemize}
    Then the following assertions hold true.
    \begin{enumerate}[label=(\roman*)]
        \item 
        \label{IT:COR:prop_brown_res_rate:cash_ins_converse}
        The following conditions are equivalent:
            \begin{itemize}[label=\scriptsize$\bullet$, noitemsep, topsep=0pt, leftmargin=1em]
                \item $g$ does not depend on $y$, i.e., it is defined on $\Om\times[0,T]\times\R^m$.
                \item $\rho$ is cash-additive.
                \item For all $X\in L^p\big(\F^W_T\big)$ and $t\in[0,T)$, $\bm\dot\rho_t(X)$ is cash-insensitive.
            \end{itemize}
        \item 
        \label{IT:COR:prop_brown_res_rate:pos_hom_converse}
        The following conditions are equivalent:
            \begin{itemize}[label=\scriptsize$\bullet$, noitemsep, topsep=0pt, leftmargin=1em]
                \item $g$ is $\P\otimes\ell_1$-a.e.\ positively homogeneous in $(y,z)$.
                \item $\rho$ is positively homogeneous.
                \item For all $X\in L^p(\F^W_T)$ and $t\in[0,T)$, $\bm\dot\rho_t(X)$ is positively homogeneous.
            \end{itemize}  
    \end{enumerate}
\end{corollary}
\begin{proof}
    Let us prove part \itemref{IT:COR:prop_brown_res_rate:cash_ins_converse}, as part \itemref{IT:COR:prop_brown_res_rate:pos_hom_converse} can be proved analogously.
    
    In the Brownian setting, and with the additional hypothesis of $g(\,\cdot\,,y,0)=0$ for any $y\in\R$, we already know from \cite[Theorem~3.1]{Jiang_2008_Convexity_translation_invariance_subadditivity_g-expectations_related_risk_measures} and \cite[Theorem~5.3]{Zheng+Li_2018_Representation_theorem_generators_quadratic_BSDEs} that the cash-additivity of $\rho$ is equivalent to $g$ being independent of $y$.
    
    Moreover, $\bm\dot\rho_t(X)$ exists for any choice of $t\in[0,T)$ and $X\in L^p\big(\F^W_T\big)$ by Corollary~\ref{COR:resil_brownian}\itemref{IT:cor:resil_brown:all_t}, which can be applied thanks to the assumption on the $\P$-a.s.\ pathwise continuity of the process $g\big(\,\cdot\,,\rho(X),Z(X)\big)$, for $X\in L^p\big(\F^W_T\big)$.
    
    It only remains to prove that the cash-insensitivity of the resilience rate implies that $g$ does not depend on the $y$ variable.
    Let us assume that $\bm\dot\rho_t(X)$ is cash-insensitive for all $t\in[0,T)$ and $X\in L^p\big(\F^W_T\big)$. 
    It is straightforward to verify that 
    \begin{equation}
    \label{EQ:COR:iff_1}
        \bm\dot\rho_t(X) =\frac{\d}{\d t }\mathbb{E}\left[\rho_t(X)\right],\qquad \forall\, t\in(0,T).
    \end{equation} 
    Indeed, by pathwise continuity of  $g(\,\cdot\,,\rho,Z)$ and the integral mean theorem, it holds $\P$-a.s.\ that
    \[
        \lim_{\eps\to0}\frac{1}{\varepsilon}\int_t^{t+\varepsilon}g(s,\rho_s,Z_s)\d s = g(t,\rho_t, Z_t), \qquad \forall\, t\in(0,T).
    \]
    Arguing exactly as in the second part of the proof of Proposition~2.2 in \cite{Jiang_2008_Convexity_translation_invariance_subadditivity_g-expectations_related_risk_measures}, it is easy to see that
    \[
        \lim_{\eps\to0}\frac{1}{\varepsilon}\mathbb{E}\left[\int_t^{t+\varepsilon}g(s,\rho_s,Z_s)\d s\right] = \mathbb{E}[g(t,\rho_t,Z_t)], \qquad \forall\, t\in(0,T).
    \]
    The last convergence immediately yields that, for all $t\in(0,T)$,
    \begin{align}
        \lim_{\eps\to0}\frac{1}{\eps}\mathbb{E}\left[\rho_{t+\varepsilon}(X)-\rho_t(X)\right]
        =-\lim_{\eps\to0}\frac{1}{\varepsilon}\mathbb{E}\left[\int_t^{t+\varepsilon}g(s,\rho_s,Z_s)\d s\right]  
        = -\mathbb{E}[g(t,\rho_t,Z_t)],
    \end{align}
    where the left-hand side defines the derivative in equation \eqref{EQ:COR:iff_1}, while the right-hand side equals the resilience rate $\bm\dot\rho_t(X)$ by Corollary~\ref{COR:resil_brownian}\itemref{IT:cor:resil_brown:all_t}, which applies  thanks to Proposition~\ref{PROP:verification_deterministic}.
    By the previous point and cash-insensitivity, we have, for any $c\in\R$, 
    \begin{equation}
    \label{EQ:COR:iff_2}
        \frac{\d}{\d t}\mathbb{E}\left[\rho_t(X)\right] 
        =\bm\dot\rho_t(X)
        =\bm\dot\rho_t(X+c)
        = \frac{\d}{\d t}\mathbb{E}\left[\rho_t(X+c)\right], \qquad \forall\, t\in(0,T).
    \end{equation}
    In particular, the real-valued function $t\mapsto \mathbb{E}\left[\rho_t(X)-\rho_t(X+c)\right]$ is differentiable on $(0,T)$ with null derivative. 
    It follows that there exists $K\in\R$ such that
    \begin{equation}
    \label{EQ:COR:iff_3}
        \mathbb{E}\left[\rho_t(X)\right] = \mathbb{E}\left[\rho_t(X+c)\right] + K, \qquad\forall\, t\in(0,T).
    \end{equation} 
    Furthermore, for all $\xi\in L^p\big(\F^W_T\big)$, the map $t\mapsto \mathbb{E}\left[\rho_t(\xi)\right]$ is continuous on the closed interval $[0,T]$. 
    To see this, let us recall that, in the Brownian setting, $\rho(\xi)$ is $\P$-a.s.\ pathwise continuous and verifies $\mathbb{E}\left[\sup_{t\in[0,T]}|\rho_t(\xi)|^2\right]<+\infty$,
    hence the family $\big(\rho_t(\xi)\big)_{t\in[0,T]}$ is uniformly $\P$-integrable. 
    By Vitali's theorem, we infer the continuity of $[0,T]\ni t\mapsto \mathbb{E}\left[\rho_t(\xi)\right]$, obtaining 
    \[
        \mathbb{E}\left[\rho_t(X)\right] = \mathbb{E}\left[\rho_t(X+c)\right] + K, \qquad \forall\, t\in [0,T].
    \]
    If we now take $t=T$, we have
    \[
        \mathbb{E}[X] 
        = \mathbb{E}\left[\rho_T(X)\right] 
        = \mathbb{E}\left[\rho_T(X+c)\right] + K 
        = \mathbb{E}[X+c] +K
        = \mathbb{E}[X] +c  +K, 
    \]
    which proves that $K=-c$, yielding $\mathbb{E}\left[\rho_t(X)\right] = \mathbb{E}\left[\rho_t(X+c)\right] -c $. 
    If we take $t=0$, and recall that $\F^W_0$ is the trivial $\s$-algebra, it results that $\rho_0(X+c) = \rho_0(X) + c$. 
    By \cite[Theorem~3.1]{Jiang_2008_Convexity_translation_invariance_subadditivity_g-expectations_related_risk_measures} and \cite[Theorem~5.3]{Zheng+Li_2018_Representation_theorem_generators_quadratic_BSDEs}, we know that the last condition is equivalent to $g$ being independent of $y$.
    This completes the chain of equivalences.
\end{proof}

The reason why the last corollary cannot be stated in the Brownian-Poissonian setting is that Corollary~\ref{COR:resil_brownian} is formulated only for the Brownian framework (see also Remark~\ref{REM:COR_brown_res_rate}).
Consequently, if we were to adopt a general Brownian-Poissonian filtration, Theorem~\ref{TH:resilience_jumps}\itemref{IT:th:resilience_jumps:deterministic} would imply the validity of equations~\eqref{EQ:COR:iff_1} and~\eqref{EQ:COR:iff_2} only for $\ell_1$-a.e.\ $t \in [0,T)$. This would not be sufficient to infer the existence of a single constant $K \in \mathbb{R}$ satisfying equation~\eqref{EQ:COR:iff_3}, not even for $\ell_1$-a.e.\ $t \in [0,T)$.
Indeed, the Cantor function provides a simple counterexample: it is everywhere right-differentiable and almost everywhere differentiable with zero derivative, yet it is not almost everywhere constant.

\subsection{Convexity and star-shapedness}
    Let us note that not all the standard properties of risk measures are naturally inherited by the resilience rate. 
    Whereas convexity or star-shapedness are relevant characteristics when measuring profits and losses, this does not necessarily pertain to evaluating financial resilience. 
    Indeed, while diversification benefits and liquidity risk are crucial when determining capital requirements, there is no fundamental financial reason to expect the same principles to apply when measuring the speed at which the risk evaluation is expected to recover after a stress scenario. 
    Nevertheless, we can still establish some interesting connections between the convexity or star-shapedness of the dynamic risk measure and the concavity or anti-star-shapedness of the resilience rate, though under rather restrictive assumptions.

\begin{proposition}
\label{PROP:convexity}
    Let $g$ satisfy any of the conditions \ref{IT:BL_condition}, \ref{IT:L_condition}, \ref{IT:BQ_condition}, or \ref{IT:Q_condition}, and the following assumptions:
    \begin{enumerate}[label=\arabic*., noitemsep, topsep=0pt]
        \item $g$ does not depend on $(z,u)$, i.e., $g:\Om\times[0,T]\times\R\ni(\om,t,y)\mapsto g(\om,t,y)\in\R$.
        \item \label{IT:COR:mon} $g$ is $\P\otimes\ell_1$-a.e.\ non-decreasing in $y$,
        \item \label{IT:COR:conc} $g$ is $\P\otimes\ell_1$-a.e.\ convex in $y$.
    \end{enumerate}
    Then for $\ell_1$-a.e.\ $t\in[0,T)$, $\bm\dot\rho_t$ is concave, namely, for all terminal conditions\footnote{The terminal conditions belong to the appropriate spaces specified in Proposition~\ref{PROP:verification_deterministic}.} $X_1,X_2$ and all $\l\in[0,1]$,
    \[
        \l\bm\dot\rho_t(X_1) + (1-\l)\bm\dot\rho_t(X_2)\leq \bm\dot\rho_t\big(\l X_1+(1-\l)X_2\big).
    \]
    Moreover, if we replace the property \itemref{IT:COR:conc} above with
    \begin{enumerate}[noitemsep, topsep=0pt]
        \item[3\,$'$.] $g$ is $\P\otimes\ell_1$-a.e. star-shaped in $y$,
    \end{enumerate}
    then for $\ell_1$-a.e.\ $t\in[0,T)$, $\bm\dot\rho_t$ is anti-star-shaped, namely, for all terminal conditions\footnotemark[\value{footnote}]  $X$ and all $\l\in[0,1]$,
    \[
         \l\bm\dot\rho_t(X) + (1-\l)\bm\dot\rho_t(0)\leq \bm\dot\rho_t(\l X).
    \]
\end{proposition}
\begin{proof}
    Let us fix terminal conditions $X_1,X_2$ and $\lambda\in[0,1]$.
    By Proposition~\ref{PROP:verification_deterministic} and Theorem~\ref{TH:resilience_jumps}\itemref{IT:th:resilience_jumps:deterministic} we have, for $\ell_1$-a.e.\ $t\in[0,T)$,
    \begin{align*}
        \lambda\bm\dot\rho_t(X_1) + (1-\lambda)\bm\dot\rho_t(X_2)
        &= -\mathbb{E}\big[\lambda g\big(t,\rho_t(X_1)\big)\big]-(1-\lambda)\E\big[g\big(t,\rho_t(X_2)\big)\big] \\
        &= -\mathbb{E}\big[\lambda g\big(t,\rho_t(X_1)\big)+(1-\lambda)g\big(t,\rho_t(X_2)\big)\big] \\
        &\leq -\mathbb{E}\big[g\big(t,\lambda\rho_t(X_1)+(1-\lambda)\rho_t(X_2)\big)\big] \\
        &\leq -\mathbb{E}\left[g\big(t,\rho_t(\l X_1+(1-\l)X_2)\big)\right] \\
        &=\bm\dot\rho_t\big(\lambda X_1+(1-\lambda)X_2\big).
    \end{align*}
    The first inequality holds by \itemref{IT:COR:conc}
    The second inequality is due to \itemref{IT:COR:mon} and to the convexity of $\rho_t$ at all times $t\in[0,T]$, which again follows from \itemref{IT:COR:conc} (see the overview on page~\pageref{PAGE:convex}).
    The final equality holds thanks to Theorem~\ref{TH:resilience_jumps}\itemref{IT:th:resilience_jumps:deterministic}. \\
    \noindent The second part of the statement can be proved analogously.
\end{proof}

Note that the previous corollary is useful, for instance, in the case of linear drivers. 
By Girsanov’s theorem, a linear BSDE can be transformed into a BSDE whose driver no longer depends on $(z,u)$, allowing to isolate exposure to, for example, ambiguous interest rates under the risk-neutral equivalent martingale measure. 
By Proposition~\ref{PROP:convexity}, the resilience rate is then automatically concave. 
In particular, any derivative priced using the Black-Scholes model will exhibit a concave resilience rate.
See also Section~\ref{sec:examples} for further details.

\subsection{Monotonicity}
The comparison theorem is a fundamental tool in the theory of BSDEs. 
One of its main applications in the context of risk measures is in establishing their monotonicity property.
We show that a similar result can be derived for the resilience rate, under appropriate assumptions.

\begin{proposition}
\label{PROP:res_rate_comp_th}
    Let $g_1,g_2$ be drivers satisfying the condition~\ref{IT:L_condition}
    and the following comparison condition: For $y_1,y_2\in\R$,
    \begin{equation}
    \label{EQ:PROP:RAS_neq_mon_cond}
       y_1\geq y_2\quad\Longrightarrow\quad 
       \inf_{(z,u)\in\R^m\times\Lambda^2} g_1(\,\cdot\,,y_1,z,u)\geq \sup_{(z,u)\in\R^m\times\Lambda^2} g_2(\,\cdot\,,y_2,z,u).
    \end{equation}
    Further, assume that $g_1$ satisfies the condition~\ref{IT:hp_comparison_jump}.
    Then the corresponding resilience rates $\bm\dot\rho(g_1)$, $\bm\dot\rho(g_2)$ are monotone --- that is, they satisfy a form of comparison principle --- in the following sense: 
    For $X_1, X_2\in L^2(\F_T)$,
    \[
        X_1\geq X_2 \quad \Longrightarrow\quad  \bm\dot\rho_t(g_1,X_1)\leq \bm\dot\rho_t(g_2,X_2),\qquad \ell_1\text{-a.e. } t\in[0,T).
    \]
\end{proposition}
\begin{proof}
    Assume the hypotheses and fix $X_1, X_2\in L^2(\F_T)$ such that $X_1\geq X_2$. 
    Let us highlight that the condition \itemref{IT:hp_comparison_jump} on $g_1$ and our comparison condition in equation~\eqref{EQ:PROP:RAS_neq_mon_cond} allow us to apply the comparison theorem for BSDEs with jumps (see \cite[Theorem~2.5]{Royer_2006_Backward_stochastic_differential_equations_jumps_related_non-linear_expectations}) and infer that, for any $t\in[0,T]$, $\rho_t(g_1,X_1)\geq \rho_t(g_2,X_2)$, $\P$-a.s.
    Then, using Theorem~\ref{TH:resilience_jumps}\itemref{IT:th:resilience_jumps:deterministic} and equation~\eqref{EQ:PROP:RAS_neq_mon_cond} we obtain:
    \begin{align*}
        \bm\dot\rho_t(g_1,X_1)
        &=  -\mathbb{E}\left[g_1\big(t,\rho_t(g_1,X_1),Z_t(g_1,X_1), U_t(g_1,X_1)\big)\right] \\
        &\leq -\mathbb{E}\left[g_2\big(t,\rho_t(g_2,X_2),Z_t(g_2,X_2), U_t(g_2,X_2)\big)\right] \\
        &= \bm\dot\rho_t(g_2,X_2), \qquad \ell_1\text{-a.e. } t\in[0,T].  
    \end{align*}
\end{proof}

The comparison condition introduced in equation~\eqref{EQ:PROP:RAS_neq_mon_cond} is stronger than the pointwise order relation in $y$ for the functions $g_1$, $g_2$, as can be seen by taking $y_1=y_2$. 
Also, it is strictly stronger, as the simple counterexample $g_i(\om,t,y,z,u)=-i+y$, for $i=1,2$, shows.
Moreover, in the simpler case where $g_1=g_2=:g$, the condition reduces to $g$ being $\P\otimes\ell_1$-a.e.\ non-decreasing in $y$ and independent of $(z,u)$.

Let us remark that, in the Brownian setting, the assumption on $g_1$ satisfying the condition \itemref{IT:hp_comparison_jump} can be omitted, as the Brownian version of the comparison theorem in \cite[Theorem~4.4.1]{Zhang_2017_Backward_stochastic_differential_equations} can be applied in place of its Brownian-Poissonian counterpart in \cite[Theorem~2.5]{Royer_2006_Backward_stochastic_differential_equations_jumps_related_non-linear_expectations}, and the same result still holds. 
See also Remark~\ref{REM:condition_C}.

\subsection{Continuity}
\label{SEC:continuity}

Another important property for risk measures is continuity from above/below, which is intimately related to the concepts of lower semi-continuity and the Lebesgue property of the risk functional (see, e.g., \cite{Biagini+Frittelli_2009_Extension_Namioka-Klee_theorem_Fatou_property_risk_measures} for an exhaustive overview of these topics). 
The financial meaning of continuity from above/below is intuitive: it ensures that the limit of a sequence of acceptable (or unacceptable) positions remains acceptable (or unacceptable).
We aim to extend this concept to resilience rates: given a sequence of resilient parameters, their limit should also be resilient (see also Section~\ref{SEC:res_acc_fam} for further discussion).
In the following, we provide sufficient conditions under which this intuition holds.

We start with a general lemma for BSDEs with jumps; Step~1 in its proof extends Theorem~4.2.3 in \cite{Zhang_2017_Backward_stochastic_differential_equations}.
\begin{lemma}
\label{LEM:stability_jumps}
    Let us consider a sequence of drivers $(g^n)_{n\in\N}$ such that $g^n(\,\cdot\,,0,0,0)\in L^2_T$, for all $n\in\N$, and satisfying the following uniform Lipschitz condition:
    \begin{enumerate}[label=\arabic*., noitemsep, topsep=0pt]
        \item 
        \label{IT:LEM:stability_jumps:uniform_lipschitz}
        There exists $K\geq 0$ such that, for all $n\in\N$, $y_1,y_2\in\R$, $z_1,z_2\in\R^m$, $u_1,u_2\in\Lambda^2$:
        \[
            \big|g^n(\,\cdot\,,y_1,z_1,u_1)-g^n(\,\cdot\,,y_2,z_2,u_2)\big|\leq K\big(|y_1-y_2|+\norm{z_1-z_2} + \|u_1-u_2\|_{\Lambda^2}\big).
        \]
    \end{enumerate}
    Further, assume that $g:\Om\times[0,T]\times\R\times\R^m\times\Lambda^2\to\R$ satisfies the following conditions:
    \begin{enumerate}[resume, label=\arabic*., noitemsep, topsep=0pt]
        \item 
        \label{IT:LEM:stability_jumps:convergence_g^n}
        For all $(y,z,u)\in \R\times \R^m\times \Lambda^2$, $g^n(\,\cdot\,,y,z,u)\longrightarrow g(\,\cdot\,,y,z,u)$ $\P\otimes\ell_1\text{-a.e.}$
        \item $g^n(\,\cdot\,,0,0,0)\longrightarrow g(\,\cdot\,,0,0,0)$  in $L^2(\Om\times[0,T])$.
    \end{enumerate}
    Let $(X^n)_{n\in\N}$ be a convergent sequence in $L^2(\F_T)$, and let $X\in L^2(\F_T)$ be its limit.
    If we denote the solutions to the BSDE \eqref{EQ:BSDE} with parameters $(g,T,X)$ and $(g^n,T,X^n)$, for $n\in\N$, by $(Y,Z,U)$ and $(Y^n,Z^n,U^n)$, respectively, then it results that
    \[
        g^{n}(\,\cdot\,,Y^n,Z^n,U^n)\longrightarrow g(\,\cdot\,,Y,Z,U) \quad \text{ in } L^2(\Om\times[0,T]).
    \]
\end{lemma}
\begin{proof}
    We denote for brevity $\Xi:=(Y,Z,U)$, $\Xi^n:=(Y^n,Z^n,U^n)$, and $\Delta \Psi^n:=\Psi^n-\Psi$, for $n\in\N$ and $\Psi\in\{g,X,Y,Z,U\}$.
    
\textit{Step} $1$.
    We start by showing the validity of the following estimate.
    There exists $C>0$ such that, for all $n\in\N$,
    \begin{equation}
    \label{EQ:LEM:stability_jumps_step1_begin}
        \E\left[\sup_{t\in[0,T]}|\Delta Y^n_t|^2 + \int_0^T\norm{\Delta Z^n_t}^2 + \|\Delta U^n_t\|^2_{\Lambda^2}\d t\right] 
        \leq \, C\,\E\left[|\Delta X^n|^2
        +\int_0^T|\Delta g^n(t,\Xi_t)|^2\d t \right].
    \end{equation}
    
    Let us fix $n\in\N$. 
    By subtracting, member by member, the BSDE with parameters $(g,T,X)$ and solution $(Y,Z,U)$ from the BSDE with parameters $(g^n,T,X^n)$ and solution $\Xi^n$, we obtain, for all $t\in[0,T]$, $\P$-a.s.,
    \begin{align*}
        \Delta Y^n_t 
        = \Delta X^n 
        &+\int_t^T g^n(s,\Xi^n_s)-g(s,\Xi_s)\d s- \int_t^T\Delta Z^n\cdot \d W_s - \int_{(t,T]\times\R^d_\ast}\Delta U^n_s(x) \d \tilde N(s,x)\\
        = \Delta X^n&+\int_t^T \Delta g^n(s,\Xi_s) + \a_s\Delta Y^n_s + \b_s\cdot\Delta Z^n_s + 
        \big\langle\g_s,\Delta U^n_s\big\rangle^{}_{\Lambda^2}
        \d s\\
        &- \int_t^T\Delta Z^n_s\cdot \d W_s - \int_{(t,T]\times\R^d_\ast}\Delta U^n_s(x)\d \tilde N(s,x),
    \end{align*}
    where we introduced $(\a,\b,\g):\Om\times[0,T]\to\R\times\R^m\times\Lambda^2$ such that
    \begin{align*}
        \a &:= \frac{g^n(\,\cdot\,,\Xi^n)-g^n(\,\cdot\,,Y,Z^n,U^n)}{\Delta Y^n}\1_{\R_\ast}(\Delta Y^n_{}),\\
        \b &:= \frac{g^n(\,\cdot\,,Y,Z^n,U^n)-g^n(\,\cdot\,,Y,Z,U^n)}{\norm{\Delta Z^n}^2}\Delta Z^n\1_{\R^m_\ast}(\Delta Z^n_{}),\\
        \g &:= \frac{g^n(\,\cdot\,,Y,Z,U^n)-g^n(\,\cdot\,,\Xi)}{\|{\Delta U^n_{}}\|_{\Lambda^2}^2}\Delta U^n\1_{\Lambda^2\setminus\{0\}}(\Delta U^n).
    \end{align*}
    Let us observe that $|\a|, \norm{\b}, \|\g\|_{\Lambda^2}\leq K$, $\P\otimes\ell_1\text{-a.e.}$, by means of assumption \itemref{IT:LEM:stability_jumps:uniform_lipschitz}
    Therefore, if we introduce $\tilde g:\Om\times[0,T]\times\R\times\R^m\times\Lambda^2\to \R$ such that 
    \[
        \tilde g(\om,t,y,z,u)= \Delta g^n\big(\om,t,\Xi_t(\om)\big) + \a_t(\om) y + \b_t(\om)\cdot z + \big\langle\g_t(\om),u\big\rangle^{}_{\Lambda^2},
    \]
    then $\tilde g$ is a driver satisfying the condition~\ref{IT:L_condition}.
    Consequently, $(\Delta Y^n, \Delta Z^n, \Delta U^n)$ is the unique solution to the BSDE \eqref{EQ:BSDE} with parameters $(\tilde g, T, \Delta X^n)$.
    By \cite[Proposition~2.1]{Tang+Wei_2012_Representation_dynamic_time-consistent_convex_risk_measures_jumps}, there exists $C>0$ such that
    \begin{equation}
    \label{EQ:LEM:stability_jumps_step1}
        \E\left[\sup_{t\in[0,T]}|\Delta Y^n_t|^2 + \int_0^T\norm{\Delta Z^n_t}^2 + \|\Delta U^n_t\|^2_{\Lambda^2}\d t\right] 
        \leq \, C\,\E\left[|\Delta X^n|^2 + \int_0^T|\tilde g(t,0,0,0)|^2\d t \right].
    \end{equation}
    Estimate \eqref{EQ:LEM:stability_jumps_step1_begin} follows from the linearity of the expectation and the definition of $\tilde g$.

\textit{Step} $2$.
    Let us now prove the statement of the lemma.
    Fix $n\in\N$.
    By the Young inequality and the linearity of both expectation and integral, we have:
    {\small{\begin{equation}
    \label{EQ:LEM:stability_step2_begin}
        \E\left[\int_0^T\big|g^n(t,\Xi^n_t)-g(t,\Xi_t)\big|^2\d t \right]\leq 2\E\left[\int_0^T\big|g^n(t,\Xi^n_t)-g^n(t,\Xi_t)\big|^2\d t\right] + 2\, \E\left[\int_0^T|\Delta g^n(t,\Xi_t)|^2\d t \right].
    \end{equation}}}
    The first expectation in the right-hand member can be estimated as follows:
    \begin{align*}
        \E\left[\int_0^T\big|g^n(t,\Xi^n_t)-g^n(t,\Xi_t)\big|^2\d t\right]
        &\leq 3K^2\E\left[\int_0^T|\Delta Y^n_t|^2 + \norm{\Delta Z^n_t}^2 + \|\Delta U^n_t\|^2\d t\right]\\
        &\leq C_T\E\left[\sup_{t\in[0,T]}|\Delta Y^n_t|^2 + \int_0^T\norm{\Delta Z^n_t}^2 + \|\Delta U^n_t\|^2\d t\right]\\
        &\leq C'_T\E\big[|\Delta X^n|^2\big] + C'_T\E\left[\int_0^T|\Delta g^n(t,\Xi_t)|^2\d t\right],
    \end{align*}
    where we first used the uniform Lipschitz property in assumption \itemref{IT:LEM:stability_jumps:uniform_lipschitz} together with the well-known inequality $(x+y+z)^2\leq 3(x^2+y^2+z^2)$ for $x,y,z\in\R$.
    For the second inequality, we estimated the first integrand with its supremum in time, integrated it, and introduced the constant $C_T:=3K^2(T\vee 1)>0$.
    For the last inequality, we resorted to the equation~\eqref{EQ:LEM:stability_jumps_step1} from the previous step of the proof, and renamed $C'_T:=CC_T$.
    
    Let us now insert this estimate back into equation~\eqref{EQ:LEM:stability_step2_begin}. 
    We find a new positive constant ${C''_T:=2C'_T+2}$ such that
    \begin{equation}
    \label{EQ:LEM:stability_jumps_step2}
    \begin{aligned}
        \E\left[\int_0^T\big|g^n(t,\Xi^n_t)-g(t,\Xi_t)\big|^2\d t \right]
        \leq C''_T\E\big[|\Delta X^n|^2\big] + C''_T\E\left[\int_0^T|\Delta g^n(t,\Xi_t)|^2\d t\right].
    \end{aligned}
    \end{equation}
    The first term in the last member is infinitesimal as $n\to 0$, by the assumptions of the lemma.
    Concerning the second term, we aim to apply Lebesgue's dominated convergence theorem in $L^2\big(\Om\times[0,T]\big)$ to the sequence of processes $\big(\Delta g^n(\,\cdot\,,\Xi)\big)_{n\in\N}$.
    First, we have $\Delta g^n(\,\cdot\,,\Xi)\longrightarrow 0$, $\P\otimes\ell_1$-a.e., by assumption \itemref{IT:LEM:stability_jumps:convergence_g^n}.
    Moreover, it is straightforward to verify that the convergence in assumption \itemref{IT:LEM:stability_jumps:convergence_g^n} allows us to extend the uniform Lipschitzianity in assumption \itemref{IT:LEM:stability_jumps:uniform_lipschitz} to the family $(\Delta g^n)_{n\in\N}$. 
    More precisely, there exists a constant $K>0$, possibly different from the one in assumption \itemref{IT:LEM:stability_jumps:uniform_lipschitz}, such that, for all $n\in\N$, $y_1,y_2\in\R$, $z_1,z_2\in\R^m$, $u_1,u_2\in\Lambda^2$,
    \[
        \big|\Delta g^n(\,\cdot\,,y_1,z_1,u_1)-\Delta g^n(\,\cdot\,,y_2,z_2,u_2)\big|\leq K\big(|y_1-y_2|+\norm{z_1-z_2} + \|u_1-u_2\|_{\Lambda^2}\big).
    \]
    Thanks to this property, and to the hypothesis \itemref{IT:LEM:stability_jumps:convergence_g^n} for $(y,z,u)=(0,0,0)=:\underline 0\in \R\times\R^m\times\Lambda^2$, we have, for sufficiently large $n\in\N$,
    \begin{align*}
        \big|\Delta g^n\big(\,\cdot\,,\Xi\big)\big|
        \leq  \big|\Delta g^n\big(\,\cdot\,,\Xi\big)-\Delta g^n(\,\cdot\,,\underline{0})\big| + \big|\Delta g^n(\,\cdot\,,\underline{0})\big| 
        \leq K\big(|Y|+\norm{Z} + \|U\|_{\Lambda^2} +1\big).
    \end{align*}
    The stochastic process in the last member belongs to $L^2\big(\Om\times[0,T]\big)$, because $(Y,Z,U)\in\mathcal S^2_T\times\mathcal H^2_T(\R^m)\times \mathcal H^2_T(\Lambda^2)$.
    Therefore, Lebesgue's dominated convergence theorem yields $\Delta g^n(\,\cdot\,,\Xi)\longrightarrow 0$ in $L^2\big(\Om\times[0,T]\big)$ as $n\to\infty$.
    The thesis follows from computing the limit as $n\to\infty$ to both members of equation~\eqref{EQ:LEM:stability_jumps_step2}, thanks to the squeeze theorem.
\end{proof}

\begin{proposition}
    Assume the hypotheses of Lemma~\ref{LEM:stability_jumps}, and the Notation~\ref{NOT:res}.
    We have:
    \[
        \bm\dot\rho(g^n,X^n)\longrightarrow \bm\dot\rho(g,X) \quad \text{ in } L^2(0,T).
    \]
\end{proposition}
\begin{proof}
    Let $(\rho,Z,U)$, $(\rho^n,Z^n,U^n)$, for $n\in\N$, denote the solutions to the BSDEs \eqref{EQ:BSDE} with parameters $(g,T,X)$, $(g^n,T,X^n)$, respectively.
    For all $n\in\N$,
    \begin{align*}
        \int_0^T\left|\bm\dot\rho_{t}(g^n,X^n)-\bm\dot\rho_{t}(g,X)\right|^2\d t 
        &= \int_0^T \big|\mathbb{E}\left[g^n(t,\rho^n_t,Z^n_t,U^n_t)-g(t,\rho_t,Z_t,U_t)\right]\big|^2\d t\\
        &\leq \mathbb{E}\left[\int_0^T \left|g^n(t,\rho^n_t,Z^n_t,U^n_t)-g(t,\rho_t,Z_t,U_t)\right|^2\d t\right]. 
    \end{align*}
    The first equality follows from Theorem~\ref{TH:resilience_jumps}\itemref{IT:th:resilience_jumps:deterministic} and the linearity of expectations.
    The inequality follows from Jensen's inequality and Fubini's theorem.
    The last member is infinitesimal, as $n\to\infty$, by Lemma~\ref{LEM:stability_jumps}.
\end{proof}

\subsection{Resilience-acceptance families}
\label{SEC:res_acc_fam}

In the context of static and dynamic risk measures, considerable attention has been devoted to the concept of risk-acceptance sets (e.g., \cite{Follmer+Schied_2016_Stochastic_finance,Drapeau+Kupper_2013_Risk_preferences_robust_representation}).
We can formulate a similar notion for measuring financial resilience, by means of the resilience-acceptance set.

\begin{definition}
\label{DEF:AS_jumps}
    Let $\rho$ be a dynamic risk measure on $L^p(\F_T)$, for some $p\in[1,+\infty]$. 
    For $a\in\overline\R$ and $t\in[0,T)$, we define the \textbf{resilience-acceptance set} for the dynamic risk measure $\rho$, at level $a$ and time $t$, as
    \[
        \mathcal{R}^{a}_{tT}(\rho):=\big\{X\in L^{p}(\mathcal{F}_T) \, : \, \bm\dot\rho_t(X)\leq a\big\}.
    \]
    We say that the resilience-acceptance family $\big(\mathcal R^a_{tT}(\rho)\big)_{a\in\overline\R}$ is cash-insensitive and positively homogeneous, if, for all $a\in\overline\R$ and $X\in L^p(\F_T)$, we have, respectively,
    \begin{align}
        X\in\mathcal{R}^{a}_{tT}(\rho)\quad &\Longleftrightarrow\quad X+h\in\mathcal{R}^{a}_{tT}(\rho), \qquad \forall\,h\in L^p(\F_t),\\
        X\in\mathcal{R}^{a}_{tT}(\rho)\quad &\Longleftrightarrow\quad \a X\in\mathcal{R}^{\a a}_{tT}(\rho), \qquad \forall\,\a> 0.
    \end{align}
\end{definition}
The dependence of the resilience-acceptance set on the time horizon $T$ and on the dynamic risk measure $\rho$ will be suppressed when clear from the context.
Let us remark that, as the definition has been stated in full generality for dynamic risk measures, the resilience-acceptance set $\mathcal R_t^a$ could be empty for some $t\in[0,T)$ and $a\in\overline\R$.
    
We have the following characterization for resilience-acceptance sets.
\begin{proposition}
\label{PROP:resacc_jumps}
    \begin{enumerate}[label=(\roman*), noitemsep]
        \item \label{IT:PROP:resacc_jumps_i}
        If $\rho$ is a dynamic risk measure on $L^p(\F_T)$ for some $p\in[1,+\infty]$, and if $\bm\dot\rho_t(X)$ exists for some $t\in[0,T)$ and $X\in L^p(\F_T)$, then 
        \[
            \bm\dot\rho_t(X)=\min\big\{a\in\overline\R \, : \, X\in\mathcal{R}^{a}_t(\rho)\big\}.
        \]
        \item \label{IT:PROP:resacc_jumps_ii}
        Let $g$ be a driver satisfying either the condition~\ref{IT:L_condition} or the condition~\ref{IT:Q_condition}, and let $p=2$ or $p=\infty$, respectively. 
        For any $a\in\overline\R$ and $\ell_1$-a.e.\ $t\in[0,T)$, we have that
        \[
            \mathcal{R}^a_t\big(\rho(g)\big)= \big\{X\in L^{p}(\mathcal{F}_T) \, : \, -\mathbb{E}\big[g\big(t,\rho_t(g,X),Z_t(g,X), U_t(g,X)\big)\big]\leq a\big\}.
        \]
    \end{enumerate}
\end{proposition}
\begin{proof}
    The first statement can be proved by recalling Definition~\ref{DEF:AS_jumps} and by noticing that $\min\{a\in\overline\R \, : \, X\in\mathcal{R}^{a}_t\}
        =\min\{a\in\overline\R \, : \, \bm\dot\rho_t(X)\leq a\}=\bm\dot\rho_t(X)$.
    The second statement immediately follows from Definition~\ref{DEF:AS_jumps}, Theorem~\ref{TH:resilience_jumps}\itemref{IT:th:resilience_jumps:deterministic} and Proposition~\ref{PROP:verification_deterministic}.
\end{proof}

We show in the next proposition that resilience-acceptance families can be used to characterize the $(\bm\dot\rho,\rho)$-time consistency introduced in Section~\ref{SEC:time_cons}, under the assumption of cash-additivity for the underlying risk measure.

\begin{proposition}
\label{PROP:timecons_RAF}
    Let $\rho$ be a dynamic risk measure on $L^p(\F_T)$ for some $p\in[1,+\infty]$.
    If $\rho$ is cash-additive,
    then the $(\bm\dot\rho,\rho)$-time consistency of Section \ref{SEC:time_cons} is satisfied if and only if, for $0\leq t<s<T$, we have that
    \[
        \mathcal R^a_{tT}\subseteq \mathcal R^a_{ts}+\mathcal A^0_{sT},\qquad\forall\, a\in\overline\R,
    \]
    where $\mathcal R^a_{ts}:=\big\{X\in L^p(\F_s) \ : \ \bm\dot\r_t(X)\leq a\big\}$ and $\mathcal A^0_{sT}:=\big\{X\in L^p(\F_T) \ : \ \r_s(X)\leq 0\big\}$.
\end{proposition}
\begin{proof}
    First, we prove that the $(\bm\dot\rho,\rho)$-time consistency implies the desired inclusion.
    Fix $X\in \mathcal R^a_{tT}$, and write $X=\rho_{s}(X)+X-\rho_{s}(X)$.
    In the last expression, $\rho_{s}(X)\in\mathcal R^a_{ts}$ because $\rho_s(X)\in L^p(\F_s)$ and $\bm\dot\rho_{t}\big(\rho_s(X)\big)=\bm\dot\rho_t(X)\leq a$ by $(\bm\dot\rho,\rho)$-time consistency, while $X-\rho_{s}(X)\in\mathcal A^0_{sT}$ because $\rho_s\big(X-\rho_{s}(X)\big)=\rho_s\big(X)-\rho_{s}(X)=0$ by cash-additivity of $\rho$.

    Let us now prove the converse implication.
    First, recall from Proposition \ref{PROP:res_rate_properties}\itemref{IT:PROP:res_rate_properties:cash_add} that the cash-additivity for $\rho$ implies the cash-insensitivity for $\bm\dot\rho_r(\xi)$, for any choice of $r\in[0,T)$ and $\xi\in L^p(\F_T)$ for which $\bm\dot\rho_r(\xi)$ exists.
    Let us fix $X\in L^p(\F_T)$ such that $\bm\dot\rho_t(X)$ exists, and let $a:=\bm\dot\rho_t(X)\in\overline\R$. 
    Then $X\in\mathcal R^a_{tT}$, while $\tilde X:=X-\rho_{t}(X)\in \mathcal R^a_{tT}$ because $\bm\dot\rho_{t}(\tilde X)=\bm\dot\rho_t(X)\leq a$ by cash-insensitivity of $\bm\dot\rho$, and $Z:=X-\rho_{s}(X)\in \mathcal A^0_{sT}$ by cash-additivity of $\rho$. 
    Since $\mathcal R^a_{tT}\subseteq \mathcal R^a_{ts}+\mathcal A^0_{sT}$ by assumption, we have that $\tilde X-Z\in \mathcal R^a_{ts}$, hence $\bm\dot\rho_{t}(\tilde X-Z)\leq a$, namely:
    \begin{align}
        a\geq \bm\dot\rho_{t}(\tilde X-Z)
         =\bm\dot\rho_{t}\big(X-\rho_{t}(X)-(X-\rho_{s}(X))\big)
         =\bm\dot\rho_{t}\big(\rho_{s}(X)\big).
    \end{align}
    Let us now define $Y_\pm:=\rho_{s}(X)\pm\Big(\bm\dot\rho_{t}\big(\rho_{s}(X)\big)-\bm\dot\rho_{t}(X)\Big)$. 
    Then $Y_\pm\in\mathcal R^a_{ts}$ because $\bm\dot\rho_{t}(Y_\pm)=\bm\dot\rho_{t}\big(\rho_{s}(X)\big)$ by cash-insensitivity, and $\bm\dot\rho_{t}\big(\rho_{s}(X)\big)\leq a$ as proved above. 
    Since $\mathcal R^a_{tT}\subseteq\mathcal R^a_{ts}+\mathcal A^0_{sT}$ by hypothesis, we have that $X-Y_\pm\in \mathcal A^0_{sT}$, which means:
    \begin{align}
        0\geq \rho_{s}(X-Y_\pm)
         =\rho_{s}(X)-\rho_{s}(X)\mp\bm\dot\rho_{t}(\rho_{s}(X))\pm\bm\dot\rho_{t}(X).     
    \end{align}
    The upper sign yields $\bm\dot\rho_{t}\big(\rho_{s}(X)\big)\geq \bm\dot\rho_{t}(X)$, while the lower sign gives $\bm\dot\rho_{t}(X)\geq \bm\dot\rho_{t}\big(\rho_{s}(X)\big)$.
\end{proof}

\begin{remark}
\label{REM:RAS_prop}
    The properties for the resilience rate that we discussed in Section~\ref{SEC:properties_cash-ins_pos-hom} can be trivially translated into properties for resilience-acceptance sets.
    In more detail, if we assume the setting of Definition~\ref{DEF:AS_jumps} and let $t\in[0,T)$, then we have the following statements.
    \begin{itemize}[label=\scriptsize$\bullet$, noitemsep, topsep=0pt,leftmargin=1em]
        \item 
            $\bm\dot\rho_t(X)$ exists and is cash-insensitive for all $X\in L^p(\F_T)$, if and only if the resilience-acceptance family $\big(\mathcal R^a_t\big)_{a\in\overline\R}$ is cash-insensitive.
        \item 
            $\bm\dot\rho_t(X)$ exists and is positively homogeneous for all $X\in L^p(\F_T)$, if and only if the resilience-acceptance family $\big(\mathcal R^a_t\big)_{a\in\overline\R}$ is positively homogeneous.
    \end{itemize}
\end{remark}

\begin{proposition}
\label{PROP:RAS_properties}
    In the same setting of Definition~\ref{DEF:AS_jumps}, we have the following properties:
    \begin{enumerate}[label=(\roman*)]
    \item 
    \label{IT:PROP:RAS_properties:cash_add}
        If $\rho$ is cash-additive, then, for all $t\in[0,T)$, the resilience-acceptance family $\big(\mathcal{R}^{a}_t\big)_{a\in\overline\R}$ is cash-insensitive.
    \item 
    \label{IT:PROP:RAS_properties:pos_hom}
        If $\rho$ is positively homogeneous, then for all $t\in[0,T)$, the resilience-acceptance family $\big(\mathcal{R}^{a}_t\big)_{a\in\overline\R}$ is positively homogeneous.
    \end{enumerate}
\end{proposition}
\begin{proof}
    Let us suppose that $\rho$ is cash-additive and fix $t$, $X$ and $a$.
    If $X\in \mathcal R_t^a$, then $\bm\dot\rho_t(X)$ exists by definition of resilience-acceptance set, and, by Proposition~\ref{PROP:res_rate_properties}\itemref{IT:PROP:res_rate_properties:cash_add}, we know that $\bm\dot\rho_t(X+h)$ exists for all $h\in L^p(\F_t)$ and that it equals $\bm\dot\rho_t(X)$. 
    Hence, $\bm\dot\rho_t(X+h)\leq a$ if and only if $\bm\dot\rho_t(X)\leq a$. 
    This proves \itemref{IT:PROP:RAS_properties:cash_add}. 
    Part \itemref{IT:PROP:RAS_properties:pos_hom} is then proved analogously.
\end{proof}
The next result easily follows from Remark~\ref{REM:RAS_prop} and Corollary~\ref{COR:prop_brown_res_rate}, and is a converse of Proposition~\ref{PROP:RAS_properties}.
\begin{corollary}
\label{COR:conv_RAS_prop}
    In the setting of Corollary~\ref{COR:prop_brown_res_rate}, each of the statements in Corollary~\ref{COR:prop_brown_res_rate}\itemref{IT:COR:prop_brown_res_rate:cash_ins_converse} (resp.\ Corollary~\ref{COR:prop_brown_res_rate}\itemref{IT:COR:prop_brown_res_rate:pos_hom_converse}) is equivalent to the resilience-acceptance family $\big(\mathcal R^a_t\big)_{a\in\overline\R}$ being cash-insensitive (resp.\ positively homogeneous) for all $t\in[0,T)$.
\end{corollary}

\section{Examples}
\label{sec:examples}

In this section, we provide several illustrative examples.
The examples reveal that our measure of financial resilience captures both objective statistical characteristics of the underlying stochastic financial process and subjective risk preferences. 
We divide this section into two parts. 
First, we present a few examples in the context of a Brownian filtration. 
Then, we develop examples in the setting with jumps.

\subsection{Brownian filtration}

\begin{example_text}[Value of a self-financing portfolio]
\label{EX:black_scholes}

Let us consider a financial market with  a free-risk asset with null interest rate, and a risky asset $S$ evolving according to the linear SDE $\d S_t=-\mu_tS_t\d t+\sigma_tS_t\d W_t$, with $S_0>0$, where $\mu$ and $\s$ are $\P$-a.s.\ continuous $\bm\F^W$-adapted stochastic processes such that $\s_t>0$ for any $t\geq 0$.
Let $X$ represent an $\F^W_T$-measurable payoff function, and let $\pi_t$ denote the optimal fraction of a financial portfolio invested in the risky asset $S$ at time $t$.
Then, it is well-known (see \cite[Section~1.2]{ElKaroui+Peng+Quenez_1997_Backward_stochastic_differential_equations_finance}, among others) that the self-financing portfolio $V(X)$ replicating $X$ follows the Brownian BSDE with linear driver $g:\Om\times[0,T]\times\R\ni(\om,t,z)\mapsto -\mu_t(\om) z /\sigma_t(\om)$, terminal condition $X$, and solution $(V(X),Z)$, where $Z:=\pi\sigma S$.
In particular, the replicating portfolio $V$ can be interpreted as a dynamic risk measure; see Remark~\ref{REM:induced_dyn_risk_meas}.  

We define the $\bm\F^W$-adapted process
\[
    \mathcal E_t:=\exp\left(-\int_0^t\frac{\mu_s}{\s_s}\d W_s-\frac12\int_0^t\Big(\frac{\mu_s}{\s_s}\Big)^2\d s \right),\qquad \forall\, t\geq 0,
\]
and introduce the probability measure $\Q$ on $(\Om,\F)$ whose density with respect to $\P$ is ${\mathcal{E}_T}$. 
Then $\P$ and $\Q$ are equivalent probability measures.
Denote by $\E^\Q$ the expectation with respect to $\Q$.
With this notation, the standard theory for linear BSDEs, see \cite[Section~4.1]{Zhang_2017_Backward_stochastic_differential_equations}, gives a well-known formula for the replicating portfolio: $V_t(X)=\E^\Q\left[\left.X\right|\mathcal{F}_t\right]$, $t\in[0,T]$.

In addition, if $X, \mu/\s\in \mathbb D^{1,2}$, then there exists a version of the second component of the solution, $Z$, that can be evaluated through Malliavin calculus (see \cite[Section~4.3]{DiNunno+Oksendal+Proske_2009_Malliavin_calculus_Levy_processes_applications_finance}), obtaining, for $\ell_1$-a.e.\ $t\in[0,T]$,
\begin{equation}
\label{EQ:Z_impl}
\begin{aligned}
    Z_t
    =D_tV_t(X)
    &=\E^\Q\left[\left.D_tX-X\int_t^T D_t\frac{\mu_s}{\sigma_s}\d W^{\mathbb{Q}}_s\right|\mathcal{F}_t\right],\qquad\P\text{-a.s.},
\end{aligned}
\end{equation}
where $W^{\mathbb{Q}}_t:=\int_0^t\frac{\mu_s}{\sigma_s}\d s+W_t$, $t\geq 0$, is a $\Q$-Brownian motion, according to Girsanov's theorem.

In this setting, a relevant problem for an economic agent is to determine if the capital needed to self-finance the claim $X$ can be expected to decrease or increase.
In other words, the agent seeks to determine $\bm\dot V_t(X)$, i.e., the expected rate at which the capital required to replicate the claim $X$ is increasing or decreasing at time $t\in[0,T]$. 
By Theorem~\ref{TH:resilience_jumps}\itemref{IT:th:resilience_jumps:deterministic}, and recalling that $\E^\Q [Y|\F_t]=\mathcal E^{-1}_t\E[\mathcal E_TY|\F_t]$ for all $t\in[0,T]$ and $Y\in L^1(\F)$, we have, for almost any $t\in[0,T)$,
\[
    \bm\dot V_t(X)
    =-\mathbb{E}\left[g(t,Z_t)\right]
    =\mathbb{E}\left[\frac{\mu_t}{\sigma_t}Z_t\right]
    =\E^\Q \left[\frac{\mu_t}{\sigma_t}\mathcal E^{-1}_t\left(D_tX-X\int_t^T D_t\frac{\mu_s}{\sigma_s}\d W^{\mathbb{Q}}_s\right)\right].
\]
If $X=\phi(S_T)$ for some Lipschitz function $\phi:\R\to\R$, then the hypotheses of Corollary~\ref{COR:resil_brownian}\itemref{IT:cor:resil_brown:tau} are satisfied thanks to Proposition~\ref{PROP:verification_stopptimes_brown}, hence for $\t\in\mathcal T_T$, we have that the resilience rate of $V_{\tau}(X)$ is given by
\[
    \bm\dot V_\t(X)
    =\E^\Q \left[\left.\frac{\mu_\t}{\sigma_\t}\mathcal E^{-1}_\t\left(D_\t X-X\int_\t^T D_\t\frac{\mu_s}{\sigma_s}\d W^{\mathbb{Q}}_s\right)\right|\tau<T\right].
\]

We henceforth assume that both $\mu$ and $\s$ are deterministic and constant, as in the usual setting of the Black and Scholes model (see, e.g., \cite{Hull_2022_Options_Futures_Other_Derivatives}).  
We provide an explicit expression for the resilience rate of the portfolio that replicates the payoff function of a European put option,  $X=\left(K-S_T\right)^+=\phi(S_T)$, where $\phi(x):=(K-x)^+$ is Lipschitz continuous. 
In Appendix~\ref{APP:ex_BeS_payoff}, a different payoff is also considered.
The optimal strategy at time $t\in[0,T)$ to replicate the put option is then $\pi_t=-\mathcal N\big(-d_+(t,S_t)\big)$, where $\mathcal N(x):=(2\pi)^{-1/2}\int_{-\infty}^xe^{-z^2/2}\d z$, $x\in\R$, is the cumulative distribution function of a standard normal distribution, and $d_+(t,x):=\frac{1}{\s\sqrt{T-t}}\left[\ln\frac{x}{K}+\frac{\s^2}{2}(T-t)\right]$ for $x>0$.
Therefore, recalling that $Z=\s\pi S$, we obtain, for any $\t\in\mathcal T_T$ and almost any $t\in[0,T)$,
\begin{align}
\label{EQ:resilience_put}
    \bm\dot V_\t(X)
    &=-\frac\mu\s\E\left[ Z_\t\right]=-\mu\, \E\left[\left.S_\t\mathcal N\big(-d_+(\t,S_\t)\big)\right|\tau<T\right],\\
    \bm\dot V_t(X)
    &=-\frac\mu\s\E\left[ Z_t\right]=-\mu\, \E\left[S_t\mathcal N\big(-d_+(t,S_t)\big)\right],
\label{EQ:resilience_put_t}
\end{align}
which can easily be numerically computed for practical purposes.

\begin{figure}[t]
    \hspace{-1.4cm}
    \includegraphics[width=1.2\textwidth]{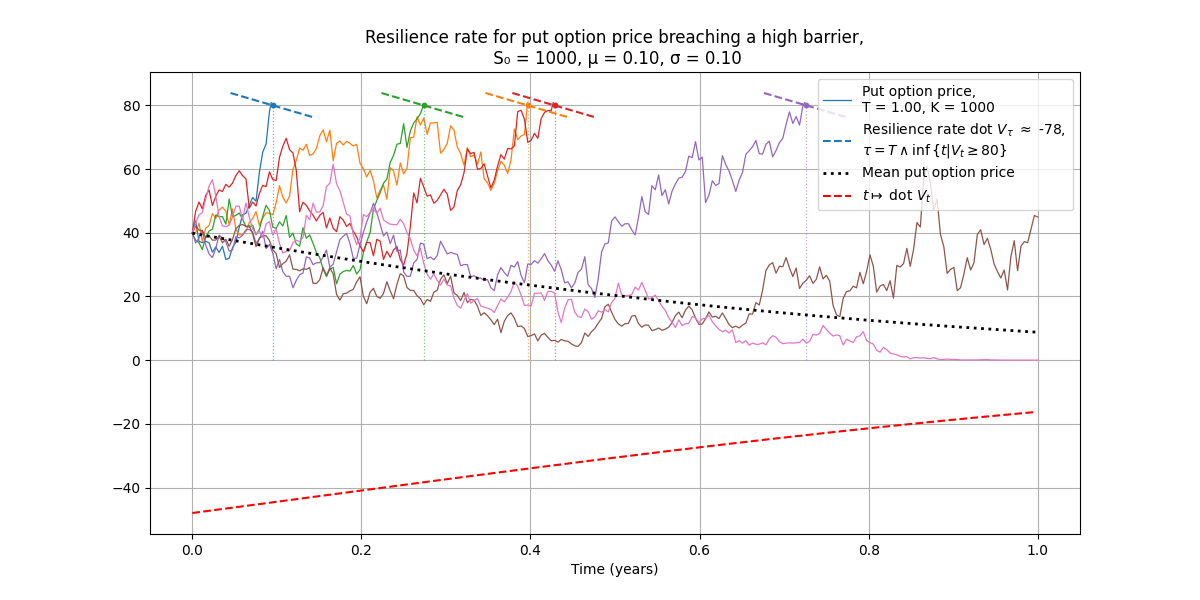}
    \caption{\footnotesize{For a time interval of one year ($1.0\, y$), we simulate $10^6$ random trajectories of the asset price $S$, which follows the Black and Scholes model with initial value $S_0 = 1.0 k\, \text{\euro}$, deterministic drift ${\mu=0.10\, y^{-1}}$, deterministic volatility $\sigma=0.10\, y^{-1/2}$, and zero risk-free interest rate.
    The time step is set to ${\d t=y/252}$.
    Based on these simulations, we use the Black and Scholes formula to calculate the value of the replicating portfolio $V(X)$ for a put option with payoff $X=(K-S_T)^+$, maturity $T=1.0 \, y$ and strike price $K=1.0 k\, \text{\euro}$.
    The colored solid lines in the graph are a selection of put option price trajectories.
    The black dotted line depicts the expected time evolution of the put option price, namely $[0,T]\ni t\mapsto \E[V_t(X)]$, numerically computed by averaging over the $10^6$ trajectories of $V(X)$. 
    The red dashed line is the time evolution of the resilience rate, i.e., $(0,T)\ni t\mapsto \bm\dot V_t(X)$, see equation~\eqref{EQ:resilience_put_t}. 
    Here, the expectation was computed via numerical integration with respect to the probability distribution of $S_t$.
    Moreover, using equation~\eqref{EQ:resilience_put}, we computed the resilience rate $\bm\dot V_\tau(X)\approx-78\, \text{\euro}y^{-1}$ at the stopping time $\tau$ defined in equation~\eqref{EQ:stopping_time}, with $c=80\, \text{\euro}$, where the expectation was numerically estimated by averaging over the $10^6$ trajectories.
    In the graph, the put option price trajectories that breached the $80\, \text{\euro}$ barrier were truncated. 
    At each truncation point, the resilience rate $\bm\dot V_\tau(X)$ is represented as the slope of an incident line.
    }}
\label{FIG:ideal}
\end{figure}

In Figure~\ref{FIG:ideal}, we display several trajectories of the price $V(X)$ of a put option with maturity one year.  
Some trajectories decline rapidly over time, while others increase due to a rise in the underlying asset's price. 
From the perspective of the option writer, a higher price for the put represents an increased risk of loss, and thus the writer hopes for a recovery towards lower option values. 
For a parameter $c>V_0(X)$, we define the stopping time
\begin{align}
\label{EQ:stopping_time}
    \tau := T\wedge \inf\{t\in[0,T] \: |\: V_t(X)\geq c\}\in\mathcal T_T.
\end{align}
We then numerically compute the quantity $\bm\dot V_\tau(X)$, which represents the expected recovery rate of the put option price when its value breaches the acceptability region and goes beyond the barrier.
This is illustrated in the graph by the slope of the colored dashed lines at the points where the trajectories first hit the barrier, set to be $c=80$.
In the same graph, we also depict the mean put option price trajectory, as a black dotted line, and the evolution of the deterministic time resilience rate $(0,T)\ni t\mapsto \bm\dot V_t(X)$, as a red dashed line. 
This allows us to interpret $\bm\dot V_t(X)$ as the slope of the mean put trajectory at time $t$. 
As a matter of fact, in the relatively simple setting of this example, the function $(0,T)\ni t\mapsto \E[V_t(X)]$ is differentiable, and therefore, its derivative at $t$ must be $\bm\dot V_t(X)$, see Remark~\ref{REM:rate_derivative}. 
We shall prove this statement in Appendix~\ref{APP:ex_BeS_proof}.
\end{example_text}

\begin{example_text}[Zero-coupon bond price]
\label{EX:vasicek}

Consider an instantaneous (short) interest rate driven by the SDE $\d r_t=a(b-r_t)\d t+\s \d W_t$, i.e., a standard Vasicek model with speed of mean reversion $a> 0$, long-term mean level $b\geq 0$, instantaneous  volatility $\s\geq 0$, and starting value $r_0>-1$, see e.g., \cite[Section~31.2]{Hull_2022_Options_Futures_Other_Derivatives} or \cite{Bazhba+Blanchet+Laeven+Zwart_2025_Large_deviations_asymptotics_unbounded_additive_functionals_diffusion_processes}.
In particular, $r_t$ is normally distributed at each $t\in[0,T]$, with mean $\mu_t:=\mathbb{E}[r_t]= r_0e^{-at}+b(1-e^{-at})$ and variance ${\Sigma^2_t:=\mathbb{E}\big[(r_t-\mu_t)^2\big]=\s^2 (1-e^{-2at})/2a}$.
If there is a zero-coupon bond available in the market, with maturity $T>0$, then its price $P$ under no-arbitrage assumptions is given by the linear Brownian BSDE
\[
    P_t=1-\int_t^Tr_sP_s\d s-\int_t^TZ_s\d W_s,
\]
with explicit solution $P_t =\E\left[\left.\exp\big(-\int_t^Tr_s\d s\big)\right|\mathcal{F}_t\right] =\exp\big(A_t-B_tr_t\big)$, where
\begin{align*}
    B_t:=\frac{1-e^{-a(T-t)}}{a},\qquad
    A_t:= -\frac{\s^2}{4a^2}\big[aB^2_t+2B_t-2(T-t)\big]+b\big[B_t-(T-t)\big].
\end{align*}
It is a simple exercise to show that $\big[aB^2_t+2B_t-2(T-t)\big]\leq 0$ for $t\leq T$, hence $A_t$ is increasing in $\s$.

Theorem~\ref{TH:resilience_jumps}\itemref{IT:th:resilience_jumps:stopping_times}, which can be applied thanks to Proposition~\ref{PROP:verification_stopptimes_brown}, yields the resilience rate of the zero-coupon bond price, as follows:
\begin{equation}
\label{EQ:vasicek_resilience_rate}
    \bm\dot P_\t =\E\left[r_\t P_\t |\tau<T\right]=\E\left[r_\t \exp\left(A_\t -B_\t r_\t \right)|\tau<T\right], \qquad \t\in\mathcal T_T.
\end{equation}
Since, for deterministic times $t\in[0,T]$, $r_t$ follows a Gaussian distribution with mean $\mu_t$ and variance $\Sigma^2_t$, after some algebra, we obtain the explicit formula
\begin{equation}
\label{EQ:vasicek_dot_Pt}
    \bm\dot P_t=
    \left(\mu_t-B_t\Sigma^2_t\right)\exp\left(A_t+\frac{\Sigma_t^2B^2_t}{2}-\mu_tB_t\right).
\end{equation}
Both $\bm\dot P_\t$ and $\bm\dot P_t$ are numerically illustrated in Figure~\ref{FIG:vas_rate}.

\begin{figure}[t]
    \hspace{-1.4cm}
    \includegraphics[width=1.2\textwidth]{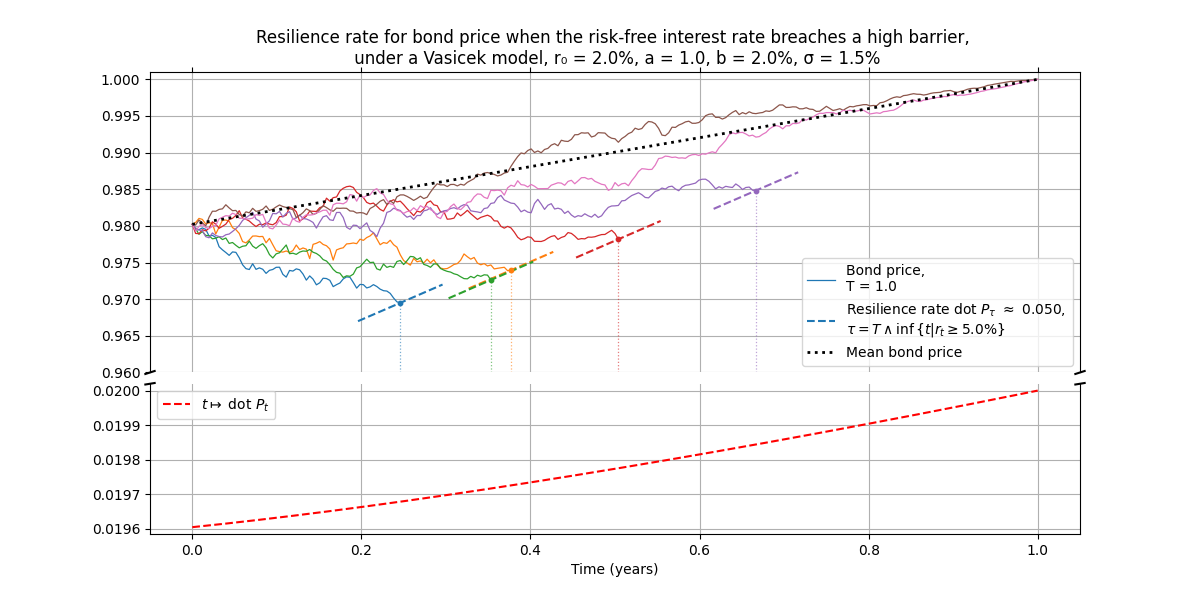}
    \caption{\footnotesize{For a time interval of one year ($1.0\, y$), we simulate $10^6$ random trajectories for a risk-free interest rate $r$ that follows a standard Vasicek model with speed of mean reversion $a=1.0 \,y^{-1}$, long-term mean level $b=2.0\% \,y^{-1}$, initial value $r_0 = 2.0\% \,y^{-1}$ and instantaneous volatility $\s = 1.0\% \, y^{-3/2}$.
    The time step is set to ${\d t=y/252}$.
    Based on these simulations, we compute the price $P$ of a zero-coupon bond available in the market.
    The colored solid lines in the graph are a selection of bond price trajectories.
    The black dotted line depicts the expected time evolution of the bond price, namely $[0,T]\ni t\mapsto \E[P_t]$, numerically computed by averaging over the $10^6$ trajectories of $P$. 
    The red dashed line is the time evolution of the resilience rate, i.e., $(0,T)\ni t\mapsto \bm\dot P_t$, see equation~\eqref{EQ:vasicek_dot_Pt}.
    Using equation~\eqref{EQ:vasicek_resilience_rate}, we computed the resilience rate $\bm\dot P_\tau\approx0.050\, \text{\euro}y^{-1}$ at the stopping time $\tau:=T\wedge \inf\{t\in[0,T] \ :\ r_t\geq 5.0\%\,y^{-1}\}$, where the expectation was numerically estimated by averaging over the $10^6$ trajectories.
    In the graph, the bond price trajectories for which the underlying risk-free interest rate breached the $5.0\%\,y^{-1}$ barrier, were truncated at the breaching time. 
    At each truncation point, the resilience rate $\bm\dot P_\tau$ is represented as the slope of an incident line.
    }}
\label{FIG:vas_rate}
\end{figure}

Close to maturity, i.e., in the limit as $t\to T^-$, the resilience rate tends to the mean interest rate $\mu_T$.
Furthermore, if $b\geq r_0$, the resilience rate is most typically increasing (in absolute value) with the speed of mean reversion $a$, keeping all the other parameters fixed: the faster the interest rate reverts to its long-term mean, the larger the recovery rate for the bond price. 
In Appendix~\ref{APP:ex_vasicek_reversion}, we provide a graphical representation to illustrate this property.

\end{example_text}
\begin{example_text}[Ambiguous interest rates]\label{EX:amb}
The previous example assumed the existence of a stochastic, yet non-ambiguous, instantaneous interest rate. 
In reality, however, a certain degree of ambiguity regarding the interest rate dynamics can exist, which may influence the decisions of an ambiguity averse agent. 
The present example, inspired by \cite[Example~7.2]{ElKaroui+Ravanelli_2009_Cash_sub-additive_risk_measures_interest_rate_ambiguity}, analyzes the resilience rate of an asset when the interest rate exhibits a certain level of ambiguity, specifically, when it fluctuates between an upper and a lower bound depending on the beliefs of the agent.
    
Assume that $(r_t)_{t\in[0,T]}$ and $(R_t)_{t\in[0,T]}$ are adapted, bounded, continuous and non-negative processes, which represent the inferior and superior bounds, respectively, for the ambiguous discount rate.
Let us consider the Brownian BSDE
    \begin{equation*}
        \r_t=X-\int_t^T\big(r_s\r^+_s-R_s\r^-_s\big)\d s -\int_t^T Z_s\d W_s,
    \end{equation*}
    where $x^\pm:=0\vee (\pm x)$, for $x\in\R$.
    
If we assume that $r$ and $R$ are induced by forward SDEs satisfying the assumptions of Proposition~\ref{PROP:verification_stopptimes_brown}, then, by a straightforward application of Theorem~\ref{TH:resilience_jumps}\itemref{IT:th:resilience_jumps:stopping_times}, we see that
    \[
        \bm\dot\rho_{\tau}(X)
        =\E[r_\tau\r_\tau^+-R_\tau\r^-_\tau|\tau<T],\qquad \forall\,\t\in\mathcal T_T,
    \]
which can be evaluated employing numerical simulation (see \cite{Agram+Rems+RosazzaGianin_2024_SIG-BSDE_Dynamic_Risk_Measures}, where this example has been extensively examined from a computational point of view). 
If we consider the case of a deterministic time $t\in[0,T]$, we can find explicit formulas. 
By uniqueness of the solution to Lipschitz BSDEs, it is straightforward to verify that
\begin{align*}
    \rho_t^{+} =X^{+}-\int_t^Tr_s\rho_t^{+}\d s-\int_t^TZ_s^+\d W_s,\qquad 
    \rho_t^{-}=X^{-}-\int_t^TR_s\rho_t^{-}\d s-\int_t^TZ_s^-\d W_s.
\end{align*}
Namely, $(\rho^{\pm}_t)_{t\in[0,T]}$ are driven by linear BSDEs with explicit solutions.
Thus, we have
\begin{align*}
    \bm\dot\r_t (X)
    =\E\left[r_t\r^+_t-R_t\r^-_t\right]
    &=\E\left[r_tX^+e^{-\int_t^Tr_s\d s }-R_tX^-e^{-\int_t^TR_s\d s }\right].
\end{align*}
\end{example_text} 
 
\begin{example_text}[Entropic risk measure, continued from Examples~\ref{EX:entropic_1}, \ref{EX:ecntropy_cont_2}, \ref{EX:entropy_cont_3}]\label{EX:entropy_cont_4}
Let us now consider a Brownian setting for the entropic risk measure, with a one-dimensional Brownian motion.
Then, it is known that $e^\g$ follows the quadratic Brownian BSDE
\[
    e^{\gamma}_t(X)=X+\int_t^T\frac{\gamma}{2}|Z_s|^2\d s-\int_t^TZ_s\d W_s,
\]
for any $\F^W_T$-measurable terminal condition $X$ such that $\mathbb{E}[\exp(pX)]<+\infty$ for any $p\geq \g$.
In particular, if $X\in \mathbb D^{1,2}$, then the solution to this BSDE is explicitly given by
\begin{align*}
    &e^{\gamma}_t(X)=\frac1\g\ln\E\left[\left.e^{\g X}\right|\F^W_t\right], \qquad
    Z_t=\frac{\mathbb{E}\left[\left.e^{X/\g} D_tX\right|\mathcal{F}^W_t\right]}{\mathbb{E}\left[\left. e^{X/\g} \right|\mathcal{F}^W_t\right]},
\end{align*}
see \cite[Example~5.1]{Kromer+Overbeck_2014_Representation_BSDE-based_dynamic_risk_measures_dynamic_capital_allocations}, with $D_t$ the Malliavin derivative operator at time $t$.
Let us now consider, for computational purposes, as terminal condition $h(X_T)$, where $h\in C^1_b(\R)$ and $X_T$ is the terminal point of the SDE
\[
    X_t=x+\int_0^t \mu(s)\d s +\int_0^t \s(s)\d W_s,
\]
for some deterministic $x\in \R$ and integrable deterministic functions $\mu,\s:[0,T]\to \R$.
We are then studying a (non-coupled) system of FBSDEs that satisfies the assumptions $(B1),\dots,(B4)$ from Appendix~\ref{APP:verif_assu_brownian}.
Then $D_t h(X_T)=\s(t)h'(X_T)$, hence, by Proposition~\ref{PROP:verification_stopptimes_brown}, we have
\[
    \bm\dot e^{\gamma}_\t\big(h(X_T)\big)
    =-\mathbb{E}\left[\left.\frac{\gamma}{2}|Z_\t|^2\right|\t<T\right]
    =-\frac{\gamma}{2}\mathbb{E}\left[\left.\left|\s(\t)\frac{\mathbb{E}\left[\left.e^{h(X_T)/\g}h'(X_T)\right|\mathcal{F}^W_\t\right]}{\mathbb{E}\left[\left.e^{h(X_T)/\g}\right|\mathcal{F}^W_\t\right]}\right|^2\right|\tau<T\right],
\]
for any stopping time $\t\in\mathcal T_T$.
Recalling the interpretation of $\gamma$ as the risk aversion coefficient, it is clear from the previous expression that the greater an agent's aversion to risk, the more negative the resilience rate of $e^{\gamma}_t$, hence the more resilient the financial position.
\end{example_text}

\subsection{Brownian-Poissonian filtration}

\begin{example_text}[Conditional expectation, $g\equiv 0$]\label{EX:gnull}
The simplest case of a BSDE is induced by the martingale representation theorem (see \cite[Example~2.3]{Laeven+Stadje_2014_Robust_portfolio_choice_indifference_valuation}, among others).
Let us assume ${X\in L^2(\mathcal{F}_T)}$ and consider the BSDE
\[
    \rho_t(X)=X-\int_t^TZ_s\d W_s-\int_{(t,T]\times\R^d_\ast}U_s(x)\d \tilde N(s,x),
\]
the unique solution of which has first component given by $\rho_t(X)=\mathbb{E}\left[\left.X\right|\mathcal{F}_t\right]$. 
In this case, the driver is identically zero and Theorem~\ref{TH:resilience_jumps}\itemref{IT:th:resilience_jumps:stopping_times} yields $\bm\dot\rho_{\tau}(X)=0=\bm\dot\rho_{\tau|\tau}(X)$ for any choice of $\t\in\mathcal T_T$.
Thus, the conditional expectation is \textit{resilience neutral}, i.e., once a certain level of risk has been reached, there is no expected trend to bounce back and recover. 
\end{example_text}

\begin{example_text}[Pricing in an incomplete market]
\label{EX:incomplete}

As in Example~\ref{EX:black_scholes}, let us consider a financial market with a free-risk asset with null interest rate, and a risky asset $S$. 
We now suppose that the risky asset evolves according to a linear SDE with jumps:
\begin{equation*}
    S_t=S_0+\int_0^t\mu_sS_s\d s+\int_0^t\sigma_sS_s\d W_s+\int_{[0,t]\times\R_\ast}\gamma_s(x)S_s\d \tilde N(s,x),
\end{equation*}
where $S_0>0$ is deterministic, while $\mu, \s:\Om\times[0,T]\to\R$ and $\g:\Om\times[0,T]\to \Lambda^2$ are $\P$-a.s.\ continuous $\bm\F$-adapted stochastic processes.
We suppose that $\mu$ is $\P$-a.s.\ bounded and that $\s$ is $\P$-a.s.\ positive and bounded away from $0$.
Further assume that the jump measure $\nu$ is finite, so that $N_t:=N\big([0,t]\times\R_\ast\big)$, for $t\geq 0$, is a standard Poisson process with rate $\nu(\R_\ast)$.

In general, this market will not be complete; see, for instance, \cite[Theorems~2.12, 2.14]{Oksendal+Sulem_2019_Applied_stochastic_control_jump_diffusions}.
However, for a claim $X$, i.e., an $\F_T$-measurable random variable, we may still define an arbitrage-free price $P_t(X)$ at time $t\in[0,T]$ by $P_t(X)=\E^\Q[X|\F_t]$, where $\E^\Q$ denotes the expectation with respect to an equivalent martingale measure $\Q$.
Suppose that
\begin{align*}
    \mathcal E_t
    :=\exp\left(-\int_0^t\frac{\mu_s}{\s_s}\d W_s-\frac12\int_0^t\left(\frac{\mu_s}{\s_s}\right)^2\d s \right)
\end{align*}
exists for $t\in[0,T]$ and satisfies $\E[\mathcal E_T]=1$. 
If we define the probability measure $\Q$ on $(\Om,\F_T)$ in such a way that its density with respect to $\P$ is ${\mathcal{E}_T}$, then $\P$ and $\Q$ are equivalent probability measures.
By the version of Girsanov's theorem in \cite[Theorems~1.31, 1.33]{Oksendal+Sulem_2019_Applied_stochastic_control_jump_diffusions}, and since $\mu/\s$ is $\P$-a.s.\ bounded, we know that $\Q$ is an equivalent martingale measure.
Moreover, if we let $W^\Q_t := \mu_t/\s_t+W_t$, for $t\in[0,T]$, then $W^\Q$ is a $\Q$-Wiener process, while  $\tilde N$ remains a compensated Poisson random measure of $N$ with respect to $\Q$, in the sense that, for all $t\in[0,T]$ and all $B\in\mathscr B(\R_\ast)$, the process $\tilde N([0,t]\times B)$ is a $\Q$-martingale.
It is simple to show (see \cite[Theorem~4.8]{Oksendal+Sulem_2019_Applied_stochastic_control_jump_diffusions}) that there exists $(Z,U)\in\mathcal H^2_T(\R)\times\mathcal H^2_T(\Lambda^2)$ such that $\big(P(X),Z,U\big)$ is the solution of the linear BSDE
\begin{align*}
    P_t(X)
    &=X-\int_t^T\frac{\mu_s}{\s_s}Z_s\d s -\int_t^TZ_s\d W_s - \int_{(t,T]\times\R_\ast}U_s(x)\d\tilde N(s,x)\\
    &=X-\int_t^TZ_s\d W^\Q_s - \int_{(t,T]\times\R_\ast}U_s(x)\d\tilde N(s,x),
\end{align*}
hence $P(X)$ can be seen as a dynamic risk measure on $L^\infty(\F_T)$, see Remark~\ref{REM:induced_dyn_risk_meas}.

We now resort to the Clark-Ocone formula in \cite[Theorem~12.22]{DiNunno+Oksendal+Proske_2009_Malliavin_calculus_Levy_processes_applications_finance} applied to $X$, to obtain, for $s\in[0,T]$,
\begin{align*}
    Z_s
    &= \E^\Q\left[\left.D_sX-X\int_s^TD_s\frac{\mu_r}{\s_r}\d W^\Q_r\right|\F_s\right], \qquad\Q\text{-a.s.},
\end{align*}
where $D_s$, $s\geq 0$, is the Malliavin derivative with respect to $W^\Q$.
We recall that the following property holds for any $Y\in L^1(\F)$ and $t\in[0,T]$: $\E^\Q [Y|\F_t]=\mathcal E^{-1}_t\E[Y\mathcal E_T|\F_t]$.
Therefore, for almost any $t\in[0,T)$, we have by Theorem~\ref{TH:resilience_jumps}\itemref{IT:th:resilience_jumps:deterministic}:
\begin{gather}
    \bm\dot P_t(X)
    =-\E\left[-\frac{\mu_t}{\s_t}Z_t\right]
    =\E\left[\frac{\mathcal E_T}{\mathcal E_t}\frac{\mu_t}{\s_t}\bigg(D_tX-X\int_t^TD_t\frac{\mu_s}{\s_s}\d W^\Q_s\bigg)\right].
\end{gather}
The formula significantly simplifies if $\mu$ and $\s$ are supposed to be deterministic.
Were this the case, we would obtain
\[
    \bm\dot P_t(X)=\E\left[\frac{\mathcal E_T}{\mathcal E_t}\frac{\mu_t}{\s_t}D_tX\right], \qquad \ell_1\text{-a.e. } t\in[0,T).
\]

Let us give an explicit computation for a European plain vanilla call option, where $X=(S_T-K)^+$ for some $K>0$, and with the additional assumption of $\mu,\s,\g$ being constants, with $\g>-1$.
By the chain rule for the Malliavin derivative, we know that 
\begin{align*}
    D_{t}X = D_t (S_T-K)^+=\s S_T\1_{[K,+\infty)}(S_T).
\end{align*}
Moreover, by direct inspection: $\mathcal E_T\mathcal E^{-1}_t = \exp\left(-\frac\mu\s(W_T-W_t)-\frac12\left(\frac\mu\s\right)^2(T-t)\right)$.
Therefore, for a.e.\ $t\in[0,T)$,
\begin{equation}
\label{EQ:formula}
    \bm\dot P_t(X)
    =\mu \exp\left(-\frac12\left(\frac\mu\s\right)^2(T-t)\right)\E\left[\exp\bigg(-\frac\mu\s(W_T-W_t)\bigg)S_T\1_{[K,+\infty)}(S_T)\right].
\end{equation}
We recall that
\[
    S_T = S_0\exp\left(\Big(\mu -\frac12\s^2\Big)T+\s W_T\right)(1+\g)^{N_T},    
\]
and that $N_T \sim N_{T-t}+N_t$, $W_T \sim W_{T-t}+W_t$, and $W_T-W_t\sim W_{T-t}$, where $W_{T-t}$, $W_t$, $N_t$, $N_{T-t}$ are independent random variables. 
Thus, the expectation in equation~\eqref{EQ:formula} can be analytically computed in terms of the distributions of the Wiener and standard Poisson processes.
The result will be expressed in terms of numerical series and Gaussian integrals, which can be numerically simulated for practical purposes.
\end{example_text}

\section{Financial Interpretation and Implications}
\label{SEC:interpretation}
Before concluding, we aim to provide some further insights into the financial interpretation and implications of the resilience rate, both from a theoretical perspective and in view of its applications for financial regulatory purposes.

\subsection{Establishing resilience neutrality}
\label{SEC:interpretation:1}
Suppose that $g,X$ satisfy any of the assumptions in Proposition~\ref{PROP:verification_deterministic}.
Then, thanks to Proposition~\ref{PROP:resacc_jumps}\itemref{IT:PROP:resacc_jumps_ii}, we have, for $\ell_1$-a.e.\ $t\in[0,T)$, that
\begin{align*}
    \bm\dot\rho_t(g,X)
    &=\min\big\{a\in\overline\R \, : \, -\mathbb{E}\big[g\big(t,\rho_t(g,X),Z_t(g,X), U_t(g,X)\big)+a\big]\leq 0\big\},
\end{align*}
where we invoked Notation~\ref{NOT:res}.
This expression provides a very natural and appealing interpretation of the resilience rate.
Indeed, it shows that $\bm\dot\rho_t(g,X)$ is the smallest amount $a\in\overline\R$ that, when added to the BSDE driver, ensures that the resulting bouncing drift is negative.
In other words, it suggests that we can properly modify the driver so as to establish a \emph{resilience-neutral} attitude for the claim $X$. 

Let us clarify the last statement. 
If we introduce $\tilde g:\Om\times[0,T)\times\R\times\R^m\times\Lambda^2\to\R$ such that
\begin{equation}
\label{EQ:REM:interpretation_resilience}
    \tilde g(\om,t,y,z,u):=g\left(\om, t, y - \int_t^T\bm\dot\rho_s(g,X)\d s , z, u\right)+\bm\dot\rho_t(g,X), 
\end{equation}
then it is easy to see that the solution to the BSDE with parameters $(\tilde g,T,X)$ is given by
\begin{equation}
            \big(\rho(\tilde g, X),\ Z(\tilde g, X),\ U(\tilde g, X)\big) = 
    \bigg(\rho(g, X)+\int_{\bm\cdot}^T\bm\dot\rho_s(g,X)\d s ,\ Z(g, X),\ U(g, X)\bigg), \label{EQ:adjusted_price}
\end{equation}
the equality holding in $\mathcal S^2_T\times\mathcal H^2_T(\R^m)\times\mathcal H^2_T(\Lambda^2)$.
Therefore, the resilience rate of the risk measure $\rho(\tilde g, X)$ at $\ell_1$-a.e.\ time $t\in[0,T)$ will equal
\begin{align}
    \bm\dot\rho_t(\tilde g, X)
    &=-\E\big[\tilde g\big(t,\rho_t(\tilde g,X), Z_t(\tilde g,X), U_t(\tilde g,X)\big)\big]\\
    &=-\E\left[g\bigg(t,\rho_t(\tilde g, X)- \int_t^T\bm\dot\rho_s(g,X)\d s, Z_t(\tilde g,X), U_t(\tilde g,X)\bigg)+ \bm\dot\rho_t(g,X) \right]\\
    &=-\E\left[g\big(t,\rho_t(g,X), Z_t(g,X), U_t(g,X)\big)\right]- \bm\dot\rho_t(g,X)=0,
\end{align}
where we first used Theorem~\ref{TH:resilience_jumps}\itemref{IT:th:resilience_jumps:deterministic}, then equation \eqref{EQ:REM:interpretation_resilience}, followed by equation \eqref{EQ:adjusted_price}, and Theorem~\ref{TH:resilience_jumps}\itemref{IT:th:resilience_jumps:deterministic} again for the last equality. 

In particular, if $g$ is independent of $y$ --- hence, the induced dynamic risk measure is cash-additive --- then $\tilde g$ simplifies to $\tilde g = g + \bm\dot\rho(g,X)$.
In this case, we can interpret the resilience rate as the precise amount to be added to the driver $g$ to ensure that the claim $X$ exhibits resilience-neutral behavior, i.e., the BSDE with parameters $\big(g + \bm\dot\rho(g,X),T,X\big)$ has null resilience rate. 
Whereas cash-additive risk measures yield a \textit{level} of capital that, when added to a financial position (and invested in a risk-less asset), makes a position risk-acceptable, cash-insensitive measures of resilience (see Proposition~\ref{PROP:res_rate_properties}\itemref{IT:PROP:res_rate_properties:cash_add}) yield an additional \textit{drift} that, when added to the bouncing drift, make the position resilience-acceptable.\footnote{ 
This interpretation is somewhat reminiscent  of the so-called \textit{net profit condition} in the actuarial literature on Cram\'er-Lundberg risk processes, ruin theory and large deviations theory. 
In particular, our bouncing drift resembles to some extent the net profit condition.
However, in our setting, not the risk process itself but the risk measurement process is subject to scrutiny; furthermore, the net profit condition ensures that ruin will not occur with probability one, whereas a resilience-acceptable position is expected to recover at a sufficiently high rate.}

In Examples~\ref{EX:black_scholes} and \ref{EX:incomplete}, the driver $g$ is independent of the $y$-component, hence the resilience rates in equations \eqref{EQ:resilience_put_t} and \eqref{EQ:formula} can be interpreted as the minimal additional drifts required to ensure a non-positive resilience rate, and consequently, an expected recovery towards the risk-acceptance region.
On the other hand, Example~\ref{EX:vasicek} does not allow for the interpretation in terms of minimal additional drift, owing to the explicit dependence of the driver on the $y$-component. 
Nevertheless, the resilience rate $\bm\dot P$ computed in equation~\eqref{EQ:vasicek_dot_Pt} can be used to define the modified driver $\tilde g(\om, t, y)=-r_t(\om)\left(y-\int_t^T\bm\dot P_s\d s \right) + \bm\dot P_t$ as in equation~\eqref{EQ:REM:interpretation_resilience}.
Then, the risk evaluation $\rho_t(\tilde g)=P_t+\int_t^T\bm\dot P_s\d s$ of the claim induced by the modified driver $\tilde g$, see equation~\eqref{EQ:adjusted_price}, will exhibit resilience-neutral behavior, i.e., $\bm\dot\rho(\tilde g )=0$.\footnote{Similar comments apply to Examples~\ref{EX:amb} ($y$-dependent $g$), \ref{EX:entropy_cont_4} ($y$-independent $g$) and \ref{EX:gnull} ($y$-independent $g$).}
We provide further insights into the applications of this point in the next subsection.
    
\subsection{The resilience risk adjustment (RRA)}
\label{SEC:RRA}
Let us consider a risk measure $\rho$ induced by a Brownian BSDE.  
When evaluating a contingent claim $X$, it is of interest to analyze the expected evolution over time of its risk measure under the real-world probability measure $\mathbb{P}$ --- that is, the probability measure under which the risk measure evolves physically.  
BSDE-induced risk measures are, by construction, designed to reflect the (genuine) preferences of an agent, and naturally translate global risk preferences encoded in $(g,T,X)$ into local risk preferences through the infinitesimal generator $g$ (\cite{Barrieu+ElKaroui_2009_Pricing_hedging_optimally_designing_derivatives_minimization_risk_measures,ElKaroui+Ravanelli_2009_Cash_sub-additive_risk_measures_interest_rate_ambiguity}).
From our theoretical results (see Remark~\ref{REM:rate_derivative} and Corollary~\ref{COR:resil_brownian}), we know that --- under suitable assumptions --- the following expansion holds, as  $t \to s^+$:
\begin{align}
    \mathbb{E}[\rho_t(X)] 
    &= \mathbb{E}[\rho_s(X)] + \bm\dot\rho_s(X) \cdot (t-s) + O\big((t-s)^2\big) \\
    &= \mathbb{E}[\rho_s(X)] - \mathbb{E}\left[g\big(s,\rho_s(X),Z_s\big)\right] \cdot (t-s) + O\big((t-s)^2\big), \quad s \in [0,T).
\end{align}
By neglecting the $O$-term, the right-hand side provides an approximation of $\mathbb{E}[\rho_t(X)]$ for $t$ sufficiently close to $s$.  
When evaluated at $s=0$, this expansion becomes particularly explicit: 
as $ t \to 0^+$,
\begin{align}
    \E[\rho_t(X)] 
    = \E[\rho_0(X)] + \bm\dot\rho_0(X) \cdot t + O(t^2)
    = \rho_0(X) - g\big(0,\rho_0(X),Z_0\big) \cdot t + O(t^2).
\end{align}
The expressions above highlight that the driver $g$ captures the agent’s local preferences over future realizations of the risk of $X$. 
More specifically, the term $\bm\dot\rho_{s}(X)$ clarifies the role of $g$ as an infinitesimal generator.  
As noted in \cite{Barrieu+ElKaroui_2009_Pricing_hedging_optimally_designing_derivatives_minimization_risk_measures}, $g$ may be interpreted as the infinitesimal expected value of a market participant.  
The associated resilience rate quantifies the sensitivity of the risk functional to short-term time fluctuations under $\mathbb{P}$.  
The expansion can also be used to obtain a linear, first-order approximation of the expected recovery time.

Importantly, as already discussed in Section~\ref{SEC:interpretation:1}, $g$ can be modified to induce a desired expected behavior --- e.g., by adjusting the drift --- thus offering a tool to control the time profile of risk.  
Such modifications are useful when modeling internal policies or regulatory constraints that impose upper bounds on the resilience rate, limiting exposure to risk fluctuations over time.  
From a modeling perspective, these constraints can be enforced by suitably altering the driver $g$.  
From a financial perspective, such adjustments reflect changes in asset pricing or risk to ensure compliance with regulatory restrictions. 
Suppose a firm must hedge a claim under resilience constraints.  
It may then modify its portfolio (cf.\ \cite[Section~4]{Laeven+Stadje_2014_Robust_portfolio_choice_indifference_valuation}) so that the associated bouncing drift --- i.e., the resilience rate --- is below a regulatory threshold.  
This leads to a state of resilience neutrality.  
Naturally, such portfolio modifications incur an additional cost, which must be accounted for in the total risk margin.  
We refer to this cost as the \emph{Resilience Risk Adjustment (RRA)}. 
The full resilience-adjusted risk measure is given by $\rho_t(\tilde{g}, X)$, where $\tilde{g}$ is the modified driver ensuring resilience neutrality (cf.\ equation~\eqref{EQ:adjusted_price}).  
More generally, we define the RRA for the claim $X$ at time $t\in[0,T]$ as
\begin{equation}
    \text{RRA}_t(X) := \int_t^T c_s\bm\dot\rho_s(X)\d s,
\end{equation}
where $c_s\geq 0$ is a rescaling factor, possibly different from unity, such that $c\in L^\infty(0,T)$.
Then, the resilience-adjusted risk measure can be expressed as
\[
    \rho^c_t(X) = \rho_t(X) + \text{RRA}_t(X), \quad t \in [0,T].
\]

Proceeding as in equation~\eqref{EQ:REM:interpretation_resilience}, it can be shown that $\rho^c(X)$ is induced by the Brownian BSDE with parameters $(g^c, T, X)$, where
\[
    g^c(\om,t, y, z) := g\left(\om, t, y - \int_t^T c_s\bm\dot\rho_s(X)\d s, z\right) + c_t\bm\dot\rho_t(X).
\]
In particular, $c_t$ may be externally imposed by a regulator to mitigate adverse short-term risk fluctuations. 
The expected resilience-adjusted risk measure satisfies
\begin{equation}
 \label{EQ:stable_price}
    \mathbb{E}[\rho^c_t(X)] = \mathbb{E}[\rho^c_s(X)] + (1-c_s)\bm\dot\rho_s(X)(t-s) + O\big((t-s)^2\big), \quad \text{as } t \to s^+, \ s \in [0,T).
\end{equation}
When $c_s = 0$, no adjustment is made; if $\bm\dot\rho_s(X)>0$ and $c_s = 1$, a full correction is performed, resulting in a resilience-neutral risk measure --- i.e., 
$
\bm\dot{\rho_t}(g^c,X) = \bm\dot{{\rho_t}}(\tilde g,X) = 0,
$
for all $t \in [0,T]$. 
By letting $c_s=0$ whenever $\bm\dot\rho_s(X)\leq0$, and $c_s\in(0,1]$ otherwise, one can allow for an RRA correction to the risk only when the resilience rate is positive, thereby accepting a negative resilience rate as a conservative and acceptable outcome. 
This shows that the threshold $c_s$ can be tuned to protect the risk measure $\rho$ against adverse short-term fluctuations under the real-world measure $\mathbb{P}$.

\section{Conclusion}\label{sec:con}

We have introduced the resilience rate as a measure of financial resilience.
It evaluates the expected rate at which a dynamic risk measure recovers from a certain stress scenario.
We have established that, under appropriate and natural conditions, the resilience rate takes the form of a suitable expectation of the generator of a BSDE, which we refer as the bouncing drift.
We have analyzed the properties of the resilience rate and introduced resilience-acceptance sets.
Furthermore, we have illustrated our results in several examples and highlighted their financial interpretation and implications.

A very extensive literature across different domains analyzes the optimal portfolio choice problem, which seeks to trade off financial return and risk.
A natural question that arises now is how to construct portfolios that optimally trade off the triple of return, risk and resilience.
We intend to analyze this question in future work.

Furthermore, our examples have illustrated that our measure of financial resilience is amenable to numerical evaluation. 
Future research may provide a more comprehensive analysis of associated numerical methods, and suitable approaches for its statistical estimation. 

We hope that our proposed measure of financial resilience aids in accomplishing a shift from measuring financial risk ----- and the design of (robust) risk management policies and regulatory measures such as adequate levels of risk capital ---, to (also) measuring financial resilience --- and the design of resilience management policies and contingent measures for corrective actions such as suitable portfolio adjustments and exploring innovation opportunities.\footnote{In \cite{Eiling+Laeven+Xu_2024_Coping_Unexpected_Forward-Looking_Measure_Firm_Resilience}, innovation and R\&D are empirically found to be key characteristics among U.S.\ listed firms that are resilient.} 

\appendix
\section*{Online Appendix}
\addcontentsline{toc}{section}{Appendix}
\renewcommand{\thesubsection}{\Alph{subsection}}
\renewcommand{\theequation}{\thesubsection.\arabic{equation}}
\numberwithin{equation}{subsection}
\refstepcounter{subsection}
\subsection*{\thesubsection\quad Example~\ref{EX:entropic_1} (cont.)}
\label{APP:example}
In this appendix, we justify the application of Proposition~\ref{PROP:integral_average}\itemref{IT:PROP:integral_average_1} in Example~\ref{EX:entropic_1}, by showing that the process integrated in time in equation~\eqref{EQ:entropic} is in $L^1_T$.
Assume that $X$ is bounded. 
Then, by monotonicity of the conditional expectation, there exist constants $c,\ C>0$ such that $0<c\leq M_t\leq C$ for all $t$.
Therefore,
\[
    \E\left[\int_0^T\left|-\frac12\frac{\norm{H_s}^2}{M^2_{s^-}}\right|\d s\right]
    \leq \frac{1}{2c^2}\E\left[\int_0^T\norm{H_s}^2\d s\right],
\]
where the last expectation is finite because $H\in\mathcal H^2_T(\R^m)$.
As far as the second term in the integrand is concerned, we will resort to the following property, for $\delta>0$:
\begin{equation}
\label{EQ:ineq_real}
    |\ln(1+u)-u|\leq \frac{u^2}{2\delta^2}, \qquad \forall\, u\geq -1+\delta.
\end{equation}
Hence, we first need to show that there exists $\delta>0$ such that
\begin{equation}
\label{EQ:ineq}
    \frac{K}{M_-}\geq -1+\delta,\qquad\P\otimes\ell_1\otimes\nu\text{-a.e.},
\end{equation}
where $M_-$ denotes the process $(M_{t^-})_{t\in(0,T]}$.

Let us introduce the $\mathcal P\otimes\mathscr B(\R^d_\ast)$-measurable set
\[
    A:=\left\{(\om,t,x)\in\Om\times(0,T]\times\R^d_\ast \ : \ \frac{K_t(\om)(x)}{M_{t^-}(\om)}<-1+\frac cC\right\},
\]
and its $\mathscr B\big([0,T]\times\R^d_\ast\big)$-measurable sections
\[
    A_\om:=\left\{(t,x)\in[0,T]\times\R^d_\ast \ : \ (\om,t,x)\in A\right\}, \qquad \forall\,\om\in\Om.
\]
The sought lower bound for $K/M_-$ follows if we prove that $(\P\otimes\ell_1\otimes\nu)(A)=0$.
Let us suppose, by way of contradiction, that $(\P\otimes\ell_1\otimes\nu)(A)>0$.
Since, for $\P$-a.e.\ $\om\in\Om$,
\[
    N(\om,A_\om)=\int_{[0,T]\times\R^d_\ast}\1_{A_\om}(s,x) N(\om,\d(s,x))=\int_{[0,T]\times\R^d_\ast}\1_{A}(\om,s,x)N(\om, \d (s,x)),
\]
we have
\begin{align}
    \int_\Om N(\om,A_\om)\d \P(\om)
    =\E\left[\int_{[0,T]\times\R^d_\ast}\1_A\d N\right]
    =\E\left[\int_0^T\int_{\R^d_\ast}\1_A \d\nu\d\ell_1\right]
    =(\P\otimes\ell_1\otimes\nu)(A)>0,
\end{align}
where we used the compensation formula for the predictable integrand $\1_A$ (cf.\ \cite[Theorem~1.8, Chapter~II]{Jacod+Shiryaev_2003_Limit_theorems_stochastic_processes}).
Hence, there exists $B\in\F$ such that $\P(B)>0$ and $N(\om,A_\om)\geq 1$ for all $\om\in B$. 
Let us denote by $\tilde\Om$ the event of full probability such that $M_s(\om)\leq C$ for all $\om\in\tilde\Om$ and $s\in[0,T]$.
Fix now $\om\in B\cap \tilde\Om$.
Since $N(\om,A_\om)\geq 1$, there exists $(t,x)\in A_\om$ that marks a jump for $M$, i.e.,
\[
    K_t(\om)(x)<\left(-1+\frac cC\right)M_{t^-}(\om), \quad \text{and }\quad M_t(\om)=M_{t^-}(\om)+K_t(\om)(x).
\]
Consequently,
\[
    \inf_{s\in[0,T]}M_s(\om)\leq M_t(\om)=M_{t^-}(\om)+K_t(\om)(x)<\frac cC M_{t^-}(\om)\leq c,
\]
where we used the upper bound for $M(\om)$ at the end.
We have demonstrated that
\[
    \P\left(\inf_{s\in(0,T]}M_s<c\right)\geq \P(B\cap\tilde\Om)>0,
\]
which contradicts the lower bound for $M$ and thus proves the validity of equation~\eqref{EQ:ineq} with $\delta:=c/C$.

We can now proceed to show that the second term in the Lebesgue integral of equation~\eqref{EQ:entropic} is also a process in $L^1_T$:
{\small{\begin{align}
    \E\left[\int_0^T\left|\int_{\R^d_\ast}\left[\ln\left(1+\frac{K_s(x)}{M_{s^-}}\right)-\frac{K_s(x)}{M_{s^-}}\right]\d \nu(x)\right|\d s \right]
    &\leq \frac{1}{2\delta^2}\E\left[\int_0^T\int_{\R^d_\ast}\left(\frac{K_s(x)}{M_{s^-}}\right)^2\d \nu (x)\d s\right]\\
    &\leq \frac{1}{2\delta^2c^2}\E\left[\int_0^T\|K_s\|^2_{\Lambda^2}\d s\right]. 
\end{align}}}
\hspace{-0.14cm}We first used the absolute value inequality in the spatial integral, together with the inequality~\eqref{EQ:ineq_real}, which can be applied thanks to equation~\eqref{EQ:ineq}.
We then resorted to the lower bound for $M$ and rewrote the spatial integral as the norm in $\Lambda^2$. 
The last term in the chain of inequalities is finite because $K\in \mathcal H^2_T(\Lambda^2)$.
To conclude, the process integrated in time in equation~\eqref{EQ:entropic} is in $L^1_T$, because it is the sum of two processes in $L^1_T$.

\refstepcounter{subsection}

\subsection*{\thesubsection\quad Stopping times in a Brownian filtration}
\label{APP:verif_assu_brownian}
In this appendix, we discuss special conditions for the verification of the assumptions in Corollary~\ref{COR:resil_brownian}\itemref{IT:cor:resil_brown:tau}, similar to what we did in Section~\ref{SEC:verif_assu_brown-poisson} for the Brownian-Poissonian filtration.
Let $n\in\N$ and assume the following:
\begin{gather}
    \mu:[0,T]\times\R^n\times\R\to\R^n,\qquad \s:[0,T]\to\R^{n\times m},\qquad h:\R^n\to\R,\\
    f:[0,T]\times\R^n\times\R\times \R^m\to\R.
\end{gather}
We consider the following system of Brownian forward-backward stochastic differential equations (FBSDE) with coefficients $\mu,\s,h,f$, initial point $x\in\R^n$ and horizon $T>0$:
\begin{equation}
\label{EQ:FBSDE}
    \begin{cases}
        \ds X_t=x+\int_0^t\mu(s,X_s, Y_s)\d s +\int_0^t\s_s \d W_s,\\
        \ds Y_t=h(X_T)+\int_t^Tf(s,X_s,Y_s,Z_s)\d s -\int_t^TZ_s\cdot \d W_s.
    \end{cases}
\end{equation}
The second equation of the above system is a special case of the Brownian BSDE \eqref{EQ:brownian_BSDE}, where the driver $g$ takes the form
\[
    g(\om,t,y,z)=f\big(t,X_t(\om),y,z\big), \qquad \forall\, (\om,t,y,z)\in\Om\times[0,T]\times\R\times\R^m,
\]
and the terminal condition is $h(X_T)$.
Assume the following:
\begin{enumerate}[label=$(B\arabic*)$, noitemsep, topsep=0pt]
    \item 
        $\mu$ is $\mathscr B\big([0,T]\times\R^n\times\R\big)$-measurable.
        For each $t\in[0,T]$, 
        $\mu(t,\,\cdot\,,\,\cdot\,)$ is continuous.
        There exist $k_1,\l_1\geq 0$ such that for all $(t,x,y)\in[0,T]\times\R^n\times\R$ and $(x',y')\in\R^n\times\R$:
        \begin{align}
            \norm{\mu(t,x,y)-\mu(t,x',y')}&\leq k_1(\norm{x-x'}+|y-y'|),\\
            \norm{\mu(t,x,y)}&\leq \l_1(1+|y|).
        \end{align}
    \item 
        $\s$ is $\mathscr B([0,T])$-measurable and bounded, and there exist $\l_2> 0$ such that ${\norm{\s_t^{\ss{\top}}x}\geq \l_2\norm x}$ for all $(t,x)\in[0,T]\times\R^n$.
    \item   
        $h$ is Lipschitz continuous and bounded.
    \item 
        $f$ is continuous. 
        There exist $k_2,k_3,k_4, \l_3\geq 0$, a sufficiently large $K\geq 0$, and a non-decreasing function $\rho:\R_+\to\R_+$ such that, for all $(t,x,y,z)\in[0,T]\times\R^n\times\R\times\R^m$, and all $(x',y',z')\in\R^n\times\R\times\R^m$:
        \begin{align}
            &\hspace{-6mm} |f(t,x,y,z)-f(t,x',y,z)|\leq k_2\norm{x-x'},\\
            &\hspace{-6mm} |f(t,x,y,z)-f(t,x,y',z')|\leq k_3|y-y'|+\rho\big(\norm{z}\vee\norm{z'}\big)\norm{z-z'},\\
            &\hspace{-6mm} |f(t,x,y,z)|\leq \l_3\big(1+|y|+\rho(\norm{z})\norm{z}\big),\\
            &\hspace{-6mm} (y-y')\big(f(t,x,y,z)-f(t,x,y',z)\big)\leq -K(y-y')^2,\\
            &\hspace{-6mm} \big|f(t,x,y,z)-f(t,x',y,z)-f(t,x,y',z')+f(t,x',y',z')\big|\leq k_4\norm{x-x'}(|y-y'|+\norm{z-z'}).
        \end{align}
\end{enumerate}

The following result is taken from \cite[Theorem~2.5]{Kupper+Luo+Tangpi_2019_Multidimensional_Markovian_FBSDEs_super-quadratic_growth}. 
For any $x\in \R^n$, there exists a unique solution $(X,Y,Z)$ to the Brownian FBSDE~\eqref{EQ:FBSDE} with coefficients $\mu$, $\s$, $h$, $f$, starting point $x$ and horizon $T$ s.t.\ $(X,Y,Z)\in \mathcal S^2_T\times \mathcal S^\infty_T\times \mathcal S^\infty_T$ are $\mathcal P^W$-measurable and $\P$-a.s.\ continuous processes.

Under this setting, we can apply Corollary~\ref{COR:resil_brownian}\itemref{IT:cor:resil_brown:tau} to the solution $(Y,Z)$ of the backward equation in the system~\eqref{EQ:FBSDE}.
Specifically, we have the following result, which is now straightforward to prove.
\begin{proposition}
\label{PROP:verification_stopptimes_brown}
    Assume $(B1),\dots,(B4)$ above, and fix $x\in\R^n$. 
    Let $(X,Y,Z)$ be the solution to the Brownian FBSDE \eqref{EQ:FBSDE} with coefficients $\mu$, $\s$, $h$, $f$, starting point $x$ and horizon $T$. 
    Then, the process $f(\,\cdot\,, X,Y,Z)$ is $\P$-a.s.\ continuous and for any $q\geq 1$,
    \[
        \E\left[\sup_{t\in[0,T]}|f(t,X_t,Y_t,Z_t)|^q\right]<+\infty.
    \]
    In particular, for any $\tau,\s\in\mathcal{T}_T$ such that $\s\leq \t$, we have:
    \begin{align}
        \bm\dot Y_{\tau|\s}\big(h(X_T)\big)&=-\frac{\E\left[\ds\1_{\{\t<T\}}\left.f\big(\t,X_\t,Y_\t,Z_{\t}\big)\right|\F_\s\right]}{\P(\t<T)},\qquad\P\text{-a.s.},\\
        \bm\dot Y_{\tau}\big(h(X_T)\big)&=-\E\left[f\big(\t,X_\t,Y_\t,Z_{\t}\big){\Big|\tau<T}\right].
    \end{align}
\end{proposition}

\refstepcounter{subsection}
\subsection*{\thesubsection\quad Example~\ref{EX:black_scholes} (cont.)}
\subsubsection{Exponential payoff}
\label{APP:ex_BeS_payoff}

In the setting of Example~\ref{EX:black_scholes}, and with $\mu$ and $\s$ deterministic and constant in time, we may consider an alternative payoff and compute the resilience rate of its replicating portfolio.
Let 
\[
    \tilde X = \exp(\s W_T)=\tilde \phi(S_T),
\]
where $\tilde \phi(x)=x S_0^{-1}\exp[-(\mu-{\s^2}/{2})T]$ is Lipschitz continuous. 
We can use the results presented in \cite[Example~4.10]{DiNunno+Oksendal+Proske_2009_Malliavin_calculus_Levy_processes_applications_finance} to evaluate the expectation in equation~\eqref{EQ:Z_impl}. 
For any $t\in[0,T]$, we have $\P$-a.s.
\[
    \tilde Z_t=\s\exp\left(\s  W_t+\frac{\mu^2}{2\sigma^2}T+\Big(\frac{\s^2}{2}-\mu\Big)(T-t)\right),
\]
from which it follows that
\begin{align*}
    \bm\dot V_\t(\tilde X)
    &=\mu \exp\left(\frac{\mu^2}{2\sigma^2}T\right)\E\left[\left.\exp\left(\Big(\frac{\s^2}{2} - \mu\Big)(T-\t)+\s W_\t\right)\right|\tau<T\right],\qquad\forall\,\t\in\mathcal T_T,\\
    \bm\dot V_t(\tilde X)
    &=
    \mu \exp\left(\frac12\Big(\frac{\mu^2}{\s^2}+\s^2\Big)T-\mu(T-t)\right), \qquad\forall\, t\in[0,T).
\end{align*}
Let us note that the resilience rate depends on the expected return $\mu$, the volatility $\sigma$, and, for deterministic times, the current time $t\in[0,T)$.
Specifically, a higher positive expected return corresponds to a larger positive expected instantaneous rate of change (as soon as $\s^2 T <4)$, whereas an increase in (small) volatility (namely, for $\s\in(0,\sqrt\mu]$), results in a resilience rate closer to $0$.
Moreover, as the asset approaches maturity, $\bm\dot V_t(\tilde X)$ increases, indicating that the rate of change for $\E[V_t(\tilde X)]$ speeds up as expiration draws near.

\subsubsection{More on equation~(\ref{EQ:resilience_put_t})}
\label{APP:ex_BeS_proof}

As an alternative to the proof provided in Example~\ref{EX:black_scholes}, we derive equation~\eqref{EQ:resilience_put_t} by differentiating in time the mean put value, which, in addition, proves its validity for any $t\in(0,T)$. 
First, for fixed $t\in(0,T)$, we explicitly compute the expectation of the quantity $V_t(X)$ with respect to the distribution of $S_t$, by means of the Black and Scholes formula:
\begin{align}
\label{EQ:integrand}\tag{A.1}
    \E[V_t(X)]
    =&\,\int_0^{+\infty}\left[K\mathcal N\big(-d_-(t,x)\big)-x\,\mathcal N\big(-d_+(t,x)\big)\right]p_{\ss{S_t}}(x)\d x,
\end{align}
where, for $x>0$,
\begin{equation*}
    p_{\ss{S_t}}(x)
    :=\frac{1}{x\s\sqrt{2\pi t}}\exp{\left(-\frac{1}{2\s^2 t }\left[\ln\frac{x}{S_0}-\Big(\mu-\frac{\s^2}{2}\Big)t\right]^2\right)},
\end{equation*}
is the probability density function of $S_t$.
If we denote the integrand in equation~\eqref{EQ:integrand} as $\Phi(t,x)$, then it is easy to verify that $\Phi\in C^{1,0}\big((0,T)\times(0,+\infty)\big)$ and that, for any ${\eps\in(0,T/2)}$, an integrable function $\theta_\eps:(0,+\infty)\to[0,T]$ exists such that ${|\partial_t\Phi(t,x)|\leq \theta_\eps(x)}$ for any $x>0$ and any $t\in[\eps,T-\eps]$.
This allows us to differentiate under the integral sign, see \cite[Theorem~2.27]{Folland_1999_Real_analysis}, and infer that $[\eps,T-\eps]\ni t\mapsto \E[V_t(X)]$ is differentiable and that, for $t\in[\eps,T-\eps]$,
\[
    \bm\dot V_t(X)
    =\frac{\d}{\d t}\E[V_t(X)]
    =\int_0^{+\infty}\partial_t\Phi(t,x)\d x 
    =-\mu\,\E\big[S_t\mathcal N\big(-d_+(t,S_t)\big)\big].
\]
By the arbitrariness of $\eps\in(0,T/2)$, we conclude that $(0,T)\ni t\mapsto \E[V_t(X)]$ is everywhere differentiable and that the above formula holds for any $t\in(0,T)$.

\refstepcounter{subsection}
\subsection*{\thesubsection\quad Example~\ref{EX:vasicek} (cont.)}
\label{APP:ex_vasicek_reversion}
\begin{figure}[ht]
    \centering
    \includegraphics[width=0.8\textwidth]{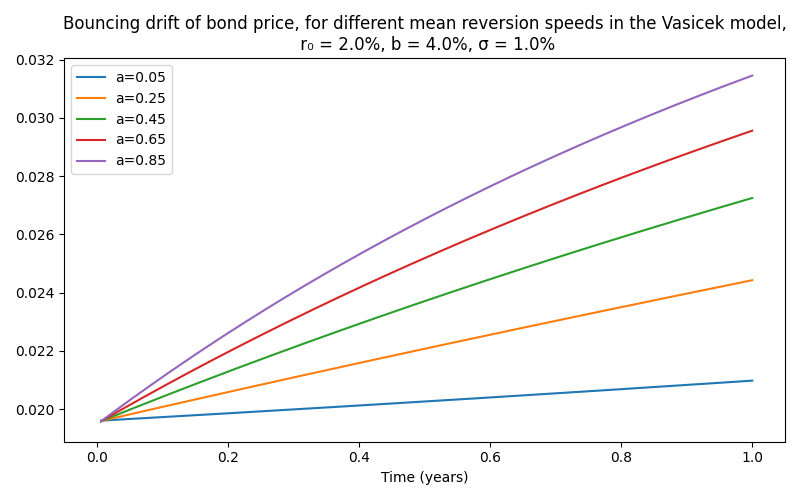}
    \caption{\footnotesize{Dependence of the resilience rate on the speed of mean reversion in Example~\ref{EX:vasicek}.}}
\label{FIG:vasicek_speed}
\end{figure}

In Figure~\ref{FIG:vasicek_speed}, we plot several instances of the bouncing drift function $[0,T)\in t \mapsto \bm\dot P_t\in\R$ defined in equation~\eqref{EQ:vasicek_dot_Pt}, for the price of a zero-coupon bond $P$ with maturity one year ($T=1.0\, y$).
The risk-free interest rate follows the Vasicek model with fixed parameters ${r_0=2.0\%y^{-1}}$, ${\s=1.0\%y^{-3/2}}$ and ${b=4.0  \%y^{-1}}$, while the different lines correspond to different values of the speed of mean reversion $a$, which are reported in the figure and are expressed in $y^{-1}$. 
As is evident from the figure, the resilience rate is strictly related to the speed of mean reversion $a$: the greater the force that pulls the risk-free rate to its long-term mean $b$, the greater the bouncing drift for the zero-coupon bond.

\printbibliography[heading=bibintoc]

\end{document}